\pdfoutput=1
\documentclass[a4paper,onecolumn,nopdfoutputerror]{quantumarticle}
\usepackage{cite} 
\usepackage{amsmath, amssymb}
\usepackage{mathrsfs}
\usepackage{xcolor}
\usepackage{url}
\usepackage{dirtytalk}
\usepackage{graphicx}

\usepackage{wrapfig} 
\usepackage[T1]{fontenc}
\usepackage{tcolorbox}
\usepackage{mathtools}
\usepackage{svg}
\usepackage{bm}
\usepackage{multirow}
\usepackage[ruled,lined,linesnumbered]{algorithm2e}
\DontPrintSemicolon                       
\usepackage{xfrac}
\usepackage{braket}
\usepackage{amsthm}
\usepackage{hyperref}
\usepackage{quantikz}
\usepackage{booktabs}
\usepackage{array}
\usepackage{float}

\usepackage{glossaries-extra}
\setabbreviationstyle[acronym]{long-short} 
\makenoidxglossaries 
\glssetcategoryattribute{acronym}{nohyperfirst}{true}



\newacronym{pmp}{PMP}{Pontryagin's Maximum Principle}
\newacronym{lqr}{LQR}{linear quadratic regulator}
\newacronym[longplural={Non-Linear Programs}, shortplural={NLPs}]
           {nlp}{NLP}{non-linear programming}
\newacronym[
  longplural={ordinary differential equations},
  shortplural={ODEs}
]{ode}{ODE}{ordinary differential equation}
\newacronym[
  longplural={partial differential equations},
  shortplural={PDEs}
]{pde}{PDE}{partial differential equation}
\newacronym{qpe}{QPE}{quantum phase estimation}
\newacronym{qft}{QFT}{quantum Fourier transform}
\newacronym{aa}{AA}{amplitude amplification}
\newacronym{qrw}{QRW}{quantum random walk}
\newacronym{lcu}{LCU}{linear combination of unitaries}
\newacronym{qsp}{QSP}{quantum signal processing}
\newacronym{qsvt}{QSVT}{quantum singular value transformation}
\newacronym{qlsa}{QLSA}{quantum linear-systems algorithm}
\newacronym{hs}{HS}{hamiltonian simulation}
\newacronym[
  longplural={algebraic Riccati equations},
  shortplural={AREs}
]{are}{ARE}{algebraic Riccati equation}
\newacronym[
  longplural={differential Riccati equations},
  shortplural={DREs}
]{dre}{DRE}{differential Riccati equation}
\newacronym{hjb}{HJB}{Hamilton--Jacobi--Bellman}
\newacronym[
  longplural={initial-value problems},
  shortplural={IVPs}
]{ivp}{IVP}{initial-value problem}
\newacronym[
  longplural={two-point boundary-value problems},
  shortplural={TPBVPs}
]{tpbvp}{TPBVP}{two-point boundary-value problem}
\newacronym[
  longplural={boundary-value problems},
  shortplural={BVPs}
]{bvp}{BVP}{boundary-value problem}
\newacronym{be}{BE}{block encoding}
\newacronym{dp}{DP}{dynamic programming}
\newacronym{hhl}{HHL}{Harrow–-Hassidim-–Lloyd}
\newacronym{qubo}{QUBO}{Quantum Approximate Optimization Algorithm}

\newacronym{hpm}{HPM}{Hamiltonian Propagator Method}
\newacronym{bdhm}{BDHM}{Block-Discretized Hamiltonian Method}

\usepackage{cleveref}
\allowdisplaybreaks

\hypersetup{
    colorlinks=true,
    linkcolor=blue,
    citecolor=magenta,
    urlcolor=magenta
}

\newtheorem{theorem}{Theorem}

\newtheorem{definition}{Definition}
\newtheorem{assumption}{Assumption}

\newtheorem{lemma}{Lemma}
\newtheorem{corollary}{Corollary}
\newtheorem{remark}{Remark}

\newtheorem{proposition}{Proposition}
\newtheorem{benchmark}{Benchmark}
\newtheorem{problem}{Problem}

\begin{document}
\title[]{Quantum algorithms for optimal control: large-scale finite-horizon linear-quadratic regulation. }

\author{Xiajie Huang\texorpdfstring{\textsuperscript{*,\textdagger}}{}}
\affiliation{School of Mathematical Sciences, Shanghai Jiao Tong University, Shanghai 200240, China}

\author{Cristina Álvarez Yuste\texorpdfstring{\textsuperscript{*}}{}}
\affiliation{Quantum Mads, Calle Larrauri 1, Edificio A, piso 3, puerta 28, 48160 Derio, Spain}
\affiliation{Department of Physical Chemistry, University of the Basque Country UPV/EHU, Apartado 644, 48080 Bilbao, Spain}

\author{Nikita Guseynov\texorpdfstring{\textsuperscript{\textdaggerdbl}}{}}
\affiliation{Global College, Shanghai Jiao Tong University, Shanghai 200240, China}

\author{Nana Liu}
\affiliation{School of Mathematical Sciences, Shanghai Jiao Tong University, Shanghai 200240, China}
\affiliation{Global College, Shanghai Jiao Tong University, Shanghai 200240, China}
\affiliation{Institute of Natural Sciences, Shanghai Jiao Tong University, Shanghai 200240, China.}
\affiliation{Ministry of Education Key Laboratory in Scientific and Engineering Computing, Shanghai Jiao Tong University, Shanghai 200240, China}

\author{Javier González-Conde\texorpdfstring{\textsuperscript{\textsection}}{}}
\affiliation{Quantum Mads, Calle Larrauri 1, Edificio A, piso 3, puerta 28, 48160 Derio, Spain}
\affiliation{Department of Physical Chemistry, University of the Basque Country UPV/EHU, Apartado 644, 48080 Bilbao, Spain}
\affiliation{EHU Quantum Center, University of the Basque Country UPV/EHU, Apartado 644, 48080 Bilbao, Spain}

\begin{abstract}
We develop quantum algorithms for optimal control over a prescribed finite time interval (finite horizon), focusing on continuous-time linear–quadratic regulation (LQR). This foundational control problem often requires solving a nonlinear matrix Riccati equation, which can be costly for large systems and long control durations on conventional classical solvers. We instead work directly with the state-costate optimality equations, reformulating the \gls{lqr} problem as a multiple-shooting linear system. Importantly, we prove a condition-number bound for the system that is independent of the horizon, under stabilizability and detectability assumptions. By incorporating the initial and final conditions into a single linear system, we can solve the subsequent linear system more efficiently with a quantum linear systems subroutine. For a given initial state, our quantum algorithms returns classical estimates of the optimal cost and the optimal control inputs at a selected time, without constructing a Riccati solution operator. Given sparse-matrix block-encodings, an initial-state preparation oracle and a fixed number of control inputs, the query and gate complexities scale polylogarithmically with system size and control duration, provided that sparsity, relevant norms, and stability parameters are bounded. The dependence on the target error is inverse-linear up to logarithmic factors. We establish the required bounds uniformly in system size and duration for a family of damped mass–spring chains with two endpoint controls. Numerical experiments support the analysis.

\end{abstract}

\renewcommand{\thefootnote}{\fnsymbol{footnote}}

\footnotetext[1]{These authors contributed equally to this work.}

\footnotetext[2]{Corresponding author. Emails: xj.huang@sjtu.edu.cn}

\footnotetext[3]{Corresponding author. Emails: guseynov.nm@gmail.com}

\footnotetext[4]{Corresponding author. Emails: javier.gonzalezc@ehu.eus}

\renewcommand{\thefootnote}{\arabic{footnote}}

\maketitle

\section{Introduction}
\label{se:introduction}
Optimal control theory studies how to choose control actions over time in order to steer a dynamical system toward a prescribed objective while optimizing a performance criterion \cite{kirk1970}. It has found broad applications in electric power systems \cite{molzahn2017}, aerospace engineering \cite{conway2010}, process control \cite{qin2003} and mathematical finance \cite{merton1971} etc. A particularly important and analytically tractable class of optimal control problems is the \gls{lqr} model \cite{kirk1970,anderson1990,liberzon2012,BennerMena2018LQR}, in which the dynamics are linear and the cost functional is quadratic. Owing to its simple yet widely applicable linear--quadratic structure, \gls{lqr} models have been extensively studied. Moreover, it serves as a fundamental computational primitive in nonlinear optimal control \cite{LiTodorov2004ILQR,SiderisBobrow2005SLQ}. These connections make \gls{lqr} both a classical benchmark problem and a practically relevant subroutine in modern optimal-control algorithms. 

\gls{lqr} problems are commonly studied in two closely related but conceptually distinct settings: finite-horizon and infinite-horizon control. The finite-horizon formulation optimizes system performance over a prescribed time interval, often with a terminal objective, and therefore leads to a time-dependent optimal feedback law. This setting is particularly natural for transient control tasks, such as maneuvering, tracking, terminal-state regulation, and receding-horizon \cite{kirk1970,anderson1990,liberzon2012, rawlings2017model}. By contrast, the infinite-horizon formulation concerns long-time regulation and, under standard stabilizability and detectability assumptions, yields a stationary stabilizing feedback law. From a computational perspective, these two formulations lead to \gls{dre}: classical finite-horizon \gls{lqr} methods typically require solving a \gls{dre}, whereas the infinite-horizon case is usually reduced to an \gls{are} \cite{kirk1970,anderson1990,liberzon2012}. Since the finite-horizon problem involves a time-dependent matrix differential equation rather than a stationary algebraic equation, it is often computationally more involved, especially for large-scale systems or when high accuracy is required over long time intervals \cite{BennerMena2018LQR,BennerBujanovicKuerschnerSaak2020}. Motivated by both its computational challenges and its relevance to time-limited control tasks, this work focuses on the finite-horizon \gls{lqr} problem.

Recent advances in quantum computing have opened new possibilities for high-dimensional scientific computation. Foundational quantum algorithms such as Shor's factoring algorithm \cite{Shor1994Algorithms} and Grover's search algorithm \cite{Grover1996Search} demonstrate that quantum computers can offer substantial complexity advantages for certain computational tasks. In scientific computing, \glspl{qlsa} \cite{Harrow2009,Ambainis2012VTAA,childs2017,chakraborty2018power,Gilyen2019,Costa2022} and quantum algorithms for ordinary and partial differential equations \cite{Berry2014,Montanaro2016FEM,Berry2017LinearODE,childs2020quantum,Childs2021PDE,PhysRevLett.133.230602,PhysRevA.108.032603,AnLiuLin2023LCHS,ShangGuoAnZhao2025Lindbladians} have emerged as important algorithmic primitives. These developments naturally raise the question of whether quantum algorithms can also accelerate large-scale optimal-control computations. 

Recently, Clayton~\emph{et al.}~\cite{Clayton2024DifferentiableQuantumControl} developed a differentiable quantum-control framework that combines quantum Lyapunov-equation solvers with policy-gradient methods for infinite-horizon \gls{lqr} problems. In parallel, recent quantum algorithms for Lyapunov, Sylvester, algebraic Riccati, and matrix-function problems provide potential building blocks for infinite-horizon \gls{lqr} and related control applications \cite{Liu2025MatrixGeometricMeans,Somma2025LinearMatrixEquations,Benedetti2026Lyapunov,WangLiu2026SignEmbedding, RodenasRuizZhaoLee2026NonlinearMatrixEquations}. Block-encoding and \gls{qsvt} techniques have also been used to formulate quantum pipelines for discrete-time  Linear-Quadratic-Gaussian (LQG) control \cite{Dehaghani2025LQG}, and this discrete-time framework has subsequently been extended to more general discrete-time nonlinear constrained optimal-control problems \cite{Dehaghani2025BarrierSQP}. These works illustrate the growing interest in quantum approaches to optimal control, while also highlighting that the literature remains at an early stage.

Quantum algorithms are particularly well-developed for linear problems, most notably linear systems and linear differential equations. By contrast, quantum algorithms for genuinely nonlinear problems remain less mature, despite recent progress based on linearization or embedding techniques \cite{liu2021dissipative,krovi2023,JinLiu2025QuantumHJ}. This limitation is especially relevant for the \gls{dre}, which is fact a nonlinear matrix differential equation. Recently, Krovi~\cite{Krovi2024OptimalControl} proposed a quantum algorithm for the \gls{dre} based on a M\"obius transformation that linearizes the Riccati dynamics. However, this approach requires solving a matrix-inversion problem whose condition number can grow exponentially with the time horizon \(T\), leading to an exponential dependence on \(T\) in the overall complexity \cite{SchiffShnider1999MobiusRiccati,BehrBennerHeiland2021DavisonMaki}. This becomes prohibitive for long-horizon finite-time control. 

More recently, Wang et al. \cite{wang2026near} developed a quantum method for Riccati equations based on weighted Riesz operators.\footnote{We became aware of this
work while completing the present manuscript.} By treating stable and unstable components of the dynamics separately, their method avoids directly encoding exponentially growing forward propagation. It constructs quantum encodings of Riccati solution matrices and can estimate classical feedback controls given quantum access to the current state. This provides a complementary route to optimal control: this formulation works through the Riccati operator, whereas in our method, we directly use  the state–costate boundary-value system.

In this work, we develop quantum algorithms for continuous-time
finite-horizon \gls{lqr} with a fixed initial state. The outputs
are classical estimates of the optimal cost or the optimal
control vector at one selected temporal grid point. We formulate the state--costate boundary-value problem as a
multiple-shooting linear system and prove a bound on its condition number. Combining this bound with Gauss--Legendre
discretization~\cite{wu2024structure} and quantum estimation
procedures gives the requested classical outputs. In the stated
input model, under suitable bounds on the input and stability
parameters and for a fixed control dimension, the query and
additional gate complexities depend polylogarithmically on the
state dimension and horizon. We also establish the required dimension-independent bounds
on the stability parameters for a family of coupled mass--spring
chains with two endpoint controls under the stated assumptions. Classical numerical experiments examine the conditioning and
discretization estimates for three families of control problems.

The remainder of the paper is organized as follows.
In Section~\ref{sec:informal-results}, we give an overview
of our results and the main ideas behind them.
We formulate the problem and state the assumptions in
Section~\ref{sec:problem-formulation}, and present the main
theorems in Section~\ref{sec:main-results}.
In Section~\ref{sec:algorithm-proof-overview}, we describe
the algorithms and outline the proofs.
We present numerical experiments in Section~\ref{sec:numerics}
and conclude in Section~\ref{sec:conclusion}.
The appendices contain the quantum primitives, detailed
constructions, and proofs.

\section{Overview of results and key ideas}
\label{sec:informal-results}
Optimal control seeks to minimize a cost function for a system
whose \(n\)-dimensional state evolves under \(d\) time-dependent
control inputs. In \gls{lqr}, the dynamics is linear and the
cost function is quadratic. We consider a fixed initial state
and a finite time horizon \(T\). Our algorithms estimate 
the optimal cost and the optimal control vector at one selected
temporal grid point. Both outputs are classical, with a target
error \(\varepsilon\) measured in absolute value for the cost
and in Euclidean norm for the control vector.
The formal problem is stated in
Section~\ref{sec:problem-formulation}.


\paragraph{Informal complexity bounds.}
We assume the sparse-matrix block-encodings and initial-state
preparation oracle specified in
Assumption~\ref{ass:be}, together with stabilizability and
detectability as in Assumption~\ref{ass:sta&det}.
Under suitable bounds on the input and stability parameters,
both the oracle-query counts and the additional gate counts
admit the following simplified bounds.

\[
\begin{aligned}
\text{Optimal-cost estimation:}\qquad
&\widetilde{\mathcal O}\!\left(
\varepsilon^{-1}\operatorname{polylog}(n,d,2+T)
\right),\\
\text{Optimal-control estimation:}\qquad
&\widetilde{\mathcal O}\!\left(
d^{3/2}\varepsilon^{-1}\operatorname{polylog}(n,d,2+T)
\right).
\end{aligned}
\]
Here logarithmic factors in the inverse error and failure
probability are suppressed. Theorems~\ref{thm:main-cost}
and~\ref{thm:main-control} give the full parameter dependence
and distinguish query costs from additional gates.
For fixed \(d\), both bounds are polylogarithmic in \(n\)
and \(T\) within this access model. 

\paragraph{Conditioning from the boundary conditions.}



\gls{pmp} \cite{pontryagin1962,liberzon2012} introduces an auxiliary vector
called the costate. Along an optimal trajectory, it describes
the sensitivity of the minimum remaining cost to changes in
the state. For LQR, the state of the dynamics system and costate together satisfy a
linear differential equation, with conditions imposed at both
the initial and terminal times.

We collect the state and costate at the temporal grid points
and impose the evolution equations between adjacent points,
together with the endpoint conditions. This construction,
known as multiple shooting, gives a block-structured linear
system. Its condition number measures the sensitivity of its
solution to perturbations and enters the cost of the quantum
algorithm.

Although the state--costate dynamics can contain growing modes,
the boundary conditions allow these modes to be treated in
the appropriate direction. A change of variables separates
a part that decays under forward propagation from a part that
decays under backward propagation. The initial and terminal
conditions supply data at the corresponding ends. The influence
of one grid point on another therefore decays with their
separation, which allows us to bound the inverse of the
assembled system.

For a fixed grid spacing \(h\), we obtain a condition-number
bound independent of the time horizon \(T\), with dependence
on \(h\) and the underlying \gls{lqr} problem. This dependence
includes the problem-dependent constant \(M\), defined in
Section~\ref{sec:main-results}. The growth of \(M\) with the
state dimension \(n\) depends on the problem family and need
not be mild, Section~\ref{subsec:state_dimension_dependence} illustrates how the
conditioning depends on the mode of dimensional growth.
This is why the informal complexity bounds above require
suitable bounds on the stability parameters. However, in this work we also provide a concrete example for which
\(M\) admits an upper bound independent of the state
dimension \(n\).



\paragraph{Dependence on the state dimension.}
Our example is a coupled mass--spring chain with two endpoint
controls. Each mass contributes a position and a velocity to
the state, so \(n\) is twice the number of masses, while \(d=2\).

For fixed positive damping, on-site and coupling stiffnesses,
and cost weights, with masses bounded above and away from zero,
we prove that the stability parameters, including \(M\),
as well as the coefficient norms and matrix sparsity, admit
bounds independent of the chain length. With a uniformly
bounded initial-state norm, both estimation tasks therefore
have polylogarithmic dependence on \(n\) and \(T\) in the
stated input model. The formal statement is given in
Section~\ref{sec:main-results}.

\section{Problem formulation and assumptions}
\label{sec:problem-formulation}

We consider the continuous-time finite-horizon \gls{lqr} problem
\begin{equation}\label{eq: LQR payoff}
\min_u J[u]
:=\frac12\int_0^T
\bigl(x(t)^\top Qx(t)+u(t)^\top Ru(t)\bigr)\,\mathrm{d}t
+\frac12x(T)^\top Gx(T),
\end{equation}
subject to
\begin{equation}\label{eq: LQR dynamics}
\dot{x}(t)=Ax(t)+Bu(t),\qquad x(0)=x^0.
\end{equation}
Here \(T>0\), \(x(t)\in\mathbb R^n\), \(u(t)\in\mathbb R^d\),
\(A\in\mathbb R^{n\times n}\), \(B\in\mathbb R^{n\times d}\),
\(Q,G\in\mathbb R^{n\times n}\), and \(R\in\mathbb R^{d\times d}\).
Throughout the paper, we assume
\[
Q=Q^\top\succeq0,\qquad G=G^\top\succeq0,\qquad R=R^\top\succ0.
\]
We write \(\|\cdot\|\) for the Euclidean norm of a vector and the
spectral norm of a matrix, \(\|X\|_{\max}:=\max_{i,j}|X_{ij}|\),
and \(\kappa(X):=\|X\|\|X^{-1}\|\) for a nonsingular square matrix. Further notation and conventions are collected in
Appendix~\ref{app:notation}.

For a fixed initial state \(x^0\), we study the following two tasks.
The temporal grid is \(t_k=kh\), where \(h=T/(m-1)\), \(m\geq2\),
and \(k=0,\ldots,m-1\).

\begin{problem}[Optimal-cost estimation]
\label{pro:optimal-cost-estimation}
Let \(J^*\) be the optimal value in
\eqref{eq: LQR payoff}--\eqref{eq: LQR dynamics}.
Given \(\varepsilon>0\) and \(0<\delta<1\), compute a classical
estimate \(J^*_{\rm est}\) such that
\[
\Pr\bigl(|J^*-J^*_{\rm est}|\leq\varepsilon\bigr)\geq1-\delta.
\]
\end{problem}

\begin{problem}[Optimal-control estimation]
\label{pro:optimal-control-estimation}
Let \(u^*(t)\) be the optimal control in
\eqref{eq: LQR payoff}--\eqref{eq: LQR dynamics}.
Given a grid index \(k\), \(\varepsilon>0\), and \(0<\delta<1\),
compute a classical vector \(u_{\rm est}(t_k)\in\mathbb R^d\) such that
\[
\Pr\bigl(\|u^*(t_k)-u_{\rm est}(t_k)\|\leq\varepsilon\bigr)
\geq1-\delta.
\]
\end{problem}

\paragraph{Input model.}
We assume block-encoding access to the system matrices. An exact
\((\alpha_X,a_X)\)-block-encoding of a matrix \(X\) is a unitary
\(U_X\) satisfying
\[
(\langle0|_{a_X}\otimes I)U_X(|0\rangle_{a_X}\otimes I)
=X_{\rm pad}/\alpha_X,
\]
where \(X_{\rm pad}\) is obtained by zero-padding \(X\) to a square
matrix of power-of-two dimension. The definition allowing
approximation error is given in Appendix~\ref{sec:block-encoding}.

\begin{assumption}[Input access]\label{ass:be}
Each of \(A,B,Q,R,G\) is \(s\)-sparse, meaning that every row and
column has at most \(s\) nonzero entries. For each nonzero
\(X\in\{A,B,Q,R,G\}\), we have an exact
\((\alpha_X,a_X)\)-block-encoding \(U_X\), with
\(\alpha_X=s\|X\|_{\max}\). Writing \(l_n=\lceil\log_2n\rceil\) and
\(l_d=\lceil\log_2d\rceil\), using the normalization and ancilla counts of
the sparse-matrix construction in Lemma~\ref{lem:be-sparse} we take
\[
a_A=a_Q=a_G=l_n+3,\qquad
a_R=l_d+3,\qquad a_B=\max\{l_n,l_d\}+3.
\]
For \(x^0\neq0\), its norm is known and a unitary \(U_{x^0}\)
prepares \(|x^0\rangle=x^0/\|x^0\|\), with zero padding if needed.
The controlled versions and inverses of these oracles are available.
\end{assumption}

This access model uses the sparse-matrix construction in Lemma~\ref{lem:be-sparse}. Here, for ease of analysis, we assume that the block-encodings specified in the input access model are exact. We count oracle queries and additional
one-qubit and CNOT gates separately. 

\paragraph{Stability assumptions.}
To obtain conditioning estimates uniform in the horizon, we impose
the following assumptions.

\begin{assumption}[Stabilizability and detectability]
\label{ass:sta&det}
The pair \((A,BR^{-1/2})\) is stabilizable and the pair
\((Q^{1/2},A)\) is detectable. This means, there exist matrices
\(K\in\mathbb R^{d\times n}\) and \(F\in\mathbb R^{n\times n}\)
such that \(A+BR^{-1/2}K\) and \(A+FQ^{1/2}\) are Hurwitz -- all the eigenvalues of these matrices have strictly negative
real parts.
\end{assumption}

These assumptions are used for the horizon-uniform estimates;
they are not required for finite-horizon solvability under the
positivity conditions above. 

\paragraph{Boundary-value formulation.}
Applying the \gls{pmp} gives the following necessary and sufficient
conditions for optimality \cite{pontryagin1962,liberzon2012}:
\begin{equation}\label{eq: ODE for LQR}
\frac{\mathrm{d}}{\mathrm{d}t}
\begin{bmatrix}x(t)\\ \lambda(t)\end{bmatrix}
=\underbrace{
\begin{bmatrix}A&-BR^{-1}B^\top\\-Q&-A^\top\end{bmatrix}
}_{=:H}
\begin{bmatrix}x(t)\\ \lambda(t)\end{bmatrix},
\end{equation}
with
\begin{equation}\label{eq: BC for LQR}
x(0)=x^0,\qquad \lambda(T)=Gx(T).
\end{equation}
Here \(x(t)\) and \(\lambda(t)\) are the optimal state and costate. We denote the system matrix by \(H\) to reflect its Hamiltonian
structure, namely,
\[
H^\top \mathbb J_{2n}+\mathbb J_{2n}H=0,\qquad
\mathbb J_{2n}:=\begin{bmatrix}
0&I_n\\
-I_n&0
\end{bmatrix}.
\]
The term ``Hamiltonian matrix'' here refers to this algebraic
property and should be distinguished from a Hamiltonian operator in quantum mechanics.
In particular, \(H\) need not be Hermitian. Then, the optimal control and cost satisfy
\begin{equation}\label{eq: u and lambda relation}
u^*(t)=-R^{-1}B^\top\lambda(t),
\qquad J^*=\frac12(x^0)^\top\lambda(0).
\end{equation}
The cost identity is proved in
Appendix~\ref{app:optimal-cost-identity-proof-v2}.
Thus both tasks reduce to extracting particular quantities from
the solution of this linear \gls{tpbvp}.

\section{Main results}
\label{sec:main-results}

Let \(y_k=(x(t_k)^\top,\lambda(t_k)^\top)^\top\) and
\(\Phi_{\rm ex}=e^{hH}\). The exact nodal values satisfy
\[
y_k-\Phi_{\rm ex}y_{k-1}=0,\qquad k=1,\ldots,m-1,
\]
together with
\[
B_0y_0+B_Ty_{m-1}=\begin{bmatrix}x^0\\0\end{bmatrix},
\qquad
B_0=\begin{bmatrix}I_n&0\\0&0\end{bmatrix},\quad
B_T=\begin{bmatrix}0&0\\-G&I_n\end{bmatrix}.
\]
Placing the boundary condition first and then the propagation
equations gives \(L_{\rm ex}q_{\rm ex}=b\), where
\(q_{\rm ex}=(y_0^\top,\ldots,y_{m-1}^\top)^\top\).
The full matrix is given in Section~\ref{sec:symplectic-preserving}.

\paragraph{Conditioning.}
Let \(P(t)\) be the solution of the finite-horizon \gls{dre} in Eq.~\eqref{eq:finite-horizon-riccati} with terminal
value \(P(T)=G\), and let \(\Psi(t,s)\) be the propagator of
\(\dot v=(A-BR^{-1}B^\top P(t))v\).
Under Assumption~\ref{ass:sta&det}, there are constants
\(M_c\geq1\), \(\alpha_c>0\), and \(M_P\geq0\), independent of
\(T\), such that
\[
\|\Psi(t,s)\|\leq M_ce^{-\alpha_c(t-s)},\quad 0\leq s\leq t\leq T,
\qquad \sup_{t\in[0,T]}\|P(t)\|\leq M_P.
\]
These estimates are established in
Propositions~\ref{prop:finite-horizon-closed-loop-stability}
and~\ref{prop:horizon-uniform-riccati-bound}. Set
\begin{equation}
\label{eq:def-M}
M:=M_c^2
\left(1+\frac{M_c^2\|B\|^2\|R^{-1}\|}{2\alpha_c}\right)^2
(1+M_P)^2.
\end{equation}

\begin{theorem}[Horizon-uniform conditioning]
\label{thm:main-conditioning}
Under Assumption~\ref{ass:sta&det}, \(L_{\rm ex}\) is nonsingular and
\[
\|L_{\rm ex}^{-1}\|
\leq M\frac{1-e^{-\alpha_chm}}{1-e^{-\alpha_ch}}.
\]
Moreover,
\begin{equation}\label{eq:condition-number-scaling}
\kappa(L_{\rm ex})
\leq M\left(1+\max\left\{\|e^{hH}\|,\sqrt{1+\|G\|^2}\right\}\right)
\frac{1-e^{-\alpha_chm}}{1-e^{-\alpha_ch}}.
\end{equation}
In particular, for fixed system matrices and a fixed step size
\(h>0\), this bound is independent of \(T\).
If \(L\) is obtained by replacing \(\Phi_{\rm ex}\) with its approximation \(\Phi\) 
and \(\|\Phi-\Phi_{\rm ex}\|\|L_{\rm ex}^{-1}\|\leq1/2\), then
\(L\) is nonsingular and \(\kappa(L)\leq3\kappa(L_{\rm ex})\).
\end{theorem}

\noindent The proof uses a Riccati--Lyapunov transformation to separate
forward-stable and backward-stable propagation; see
Section~\ref{sec:condition_error}. The Riccati and Lyapunov
solutions are used only in the conditioning analysis and are not computed by the quantum algorithm.

\paragraph{Quantum estimation.}
For the quantum algorithms below, choose
\[
m=1+\max\{1,\lceil8\pi T\|H\|\rceil\},\qquad h=\frac{T}{m-1},
\]
so that \(h\|H\|<1/(4\pi)\). We approximate \(e^{hH}\) by a
Gauss--Legendre Runge--Kutta method and construct a block-encoding
of \(L^{-1}\) using \gls{qsvt}
(Lemma~\ref{lem:inverse-be-approx}).
The required quantities are then estimated using Hadamard tests
and amplitude estimation (Lemma~\ref{lem:Hadamard-test}).
The stage number is increased to meet the target accuracy while
the temporal grid is kept fixed. The Runge--Kutta coefficient
oracles depend only on the stage number and can be constructed
with polynomial cost in that number; this cost is absorbed in
the logarithmic factors below.

For this grid, Theorem~\ref{thm:main-conditioning} gives
\(\|L_{\rm ex}^{-1}\|=\mathcal O(M(1+\alpha_c^{-1}\|H\|))\).
To state the two complexity bounds, define
\[
\mathcal C
:=M^2(1+\alpha_c^{-1}\|H\|)^2s^3\kappa(R)^2
\bigl(s^2\kappa(R)+s\|G\|_{\max}+1\bigr).
\]
The notation \(\widetilde{\mathcal O}\) suppresses logarithmic
factors in the precision and problem parameters.

\begin{theorem}[Optimal-cost estimation]
\label{thm:main-cost}
Under Assumptions~\ref{ass:be} and~\ref{ass:sta&det}, there exists
a quantum algorithm that solves Problem~\ref{pro:optimal-cost-estimation}
using
\[
\widetilde{\mathcal O}\left(
\frac{\mathcal C\|x^0\|^2}{\varepsilon}\log\frac{2}{\delta}
\right)
\]
oracle queries.
\end{theorem}

\begin{theorem}[Optimal-control estimation]
\label{thm:main-control}
Under Assumptions~\ref{ass:be} and~\ref{ass:sta&det}, there exists
a quantum algorithm that solves Problem~\ref{pro:optimal-control-estimation}
at any selected grid node using
\[
\widetilde{\mathcal O}\left(
\frac{\mathcal C\,s\,d^{3/2}\|x^0\|\|B\|_{\max}\|R^{-1}\|}
{\varepsilon}\log\frac{2d}{\delta}
\right)
\]
oracle queries.
\end{theorem}

In both theorems, the additional one-qubit and CNOT gate count
is bounded by the corresponding query bound multiplied by
\(\log(2d)\log^2(2n)\log^2(2m)\), up to the logarithmic factors
already suppressed. The algorithm and error choices are given in
Section~\ref{sec:quantum-algorithm}; the complexity bounds follow
from Appendices~\ref{app:optimal-cost-estimation-proof-v2}
and~\ref{app:optimal-control-estimation-proof-v2}.
The factor \(d^{3/2}\) comes from estimating \(d\) components to
accuracy of order \(\varepsilon/\sqrt d\).

These bounds have polylogarithmic dependence on \(n\) and \(T\)
when the remaining parameters are bounded or grow at most
polylogarithmically, including \(d\) for the control-vector output.
They concern the cost and the control at one selected time;
obtaining the full classical control trajectory requires readout
at each requested node.

\paragraph{A dimension-uniform example.}
The dependence on \(n\) also requires control of the stability
constants as the system grows. The coupled mass--spring chain
provides one such family.

\begin{corollary}[Mass--spring chain]
\label{cor:main-mass-spring}
Consider Benchmark~\ref{bench:mass_spring_chain}, with two endpoint
controls and fixed positive damping, on-site and coupling
stiffnesses, and cost weights. Suppose that
\[
0<m_-\leq m_i\leq m_+<\infty,\qquad i=1,\ldots,n_m,
\]
where \(m_-\) and \(m_+\) are independent of \(n_m\).
Then \(M_c\), \(M_P\), and \(\alpha_c^{-1}\) admit bounds
independent of \(n_m\) and \(T\). Consequently, at a fixed
admissible step size, the condition-number bound is uniform in
both \(n=2n_m\) and \(T\).
If \(\|x^0\|\) is also uniformly bounded, both estimation
algorithms have polylogarithmic dependence on \(n\) and \(T\)
in the stated access model.
\end{corollary}

This follows from Theorem~\ref{thm:Mc-alphac-Mp-bound} in
Appendix~\ref{app:Proofs for the benchmark}, together with the
uniform bounds on the matrix norms and sparsity of this family.

Figure~\ref{fig:pipeline} provides a roadmap of the proposed quantum
algorithms for optimal-cost estimation and optimal-control estimation
at a selected temporal node, together with the horizon-uniform
conditioning analysis underlying their complexity bounds.

\begin{figure*}[!htbp]
    \centering
    \makebox[\linewidth][c]{%
        \includegraphics[width=\linewidth]{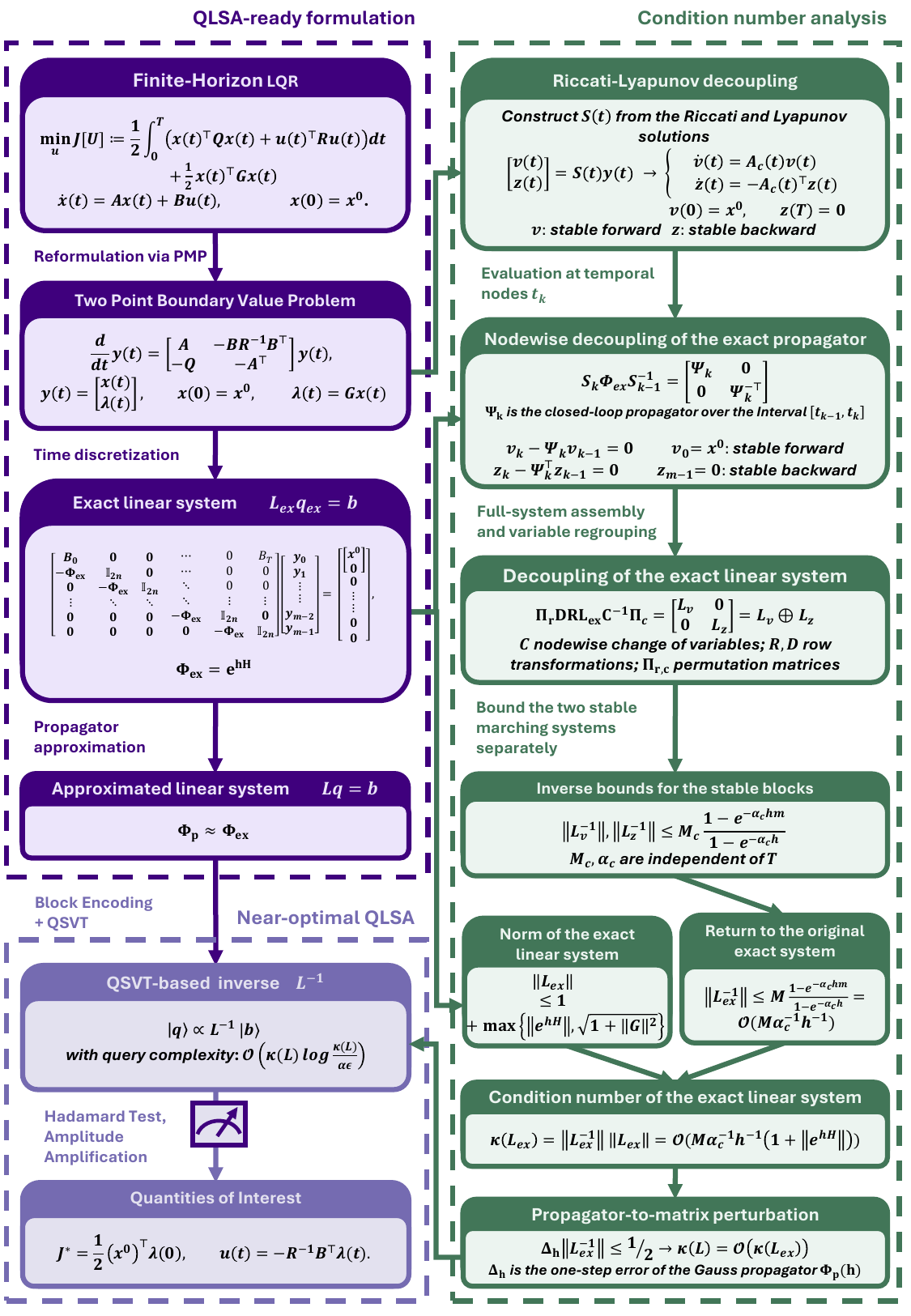}%
    }
    \caption{Roadmap of the proposed quantum algorithms and their
    conditioning analysis.
    Left: reformulation and discretization of the finite-horizon LQR
    problem, followed by quantum estimation of the optimal cost and
    the optimal control at a selected temporal node.
    Right: Riccati--Lyapunov decoupling and perturbation analysis
    establishing horizon-uniform condition-number bounds.}
    \label{fig:pipeline}
\end{figure*}

\section{Algorithm and proof overview}
\label{sec:algorithm-proof-overview}

We describe the construction of the quantum algorithms and outline the
proofs of the results in Section~\ref{sec:main-results}. We first analyse
the exact nodal system, then introduce the propagator approximation
and the quantum implementation.

\subsection{Multiple-shooting formulation}
\label{sec:symplectic-preserving}

With the notation of Section~\ref{sec:main-results}, set
\begin{equation}\label{eq:def_q_ex}
q_{\rm ex}=(y_0^\top,\ldots,y_{m-1}^\top)^\top
\in\mathbb R^{2nm}.
\end{equation}
The boundary conditions and the propagation equations give
\(L_{\rm ex}q_{\rm ex}=b\), where
\begin{equation}\label{eq:exact-shooting-operator-block-form}
L_{\rm ex}=
\begin{bmatrix}
B_0&0&\cdots&B_T\\
-\Phi_{\rm ex}&I_{2n}&&0\\
&\ddots&\ddots&\\
0&&-\Phi_{\rm ex}&I_{2n}
\end{bmatrix},
\end{equation}
and
\begin{equation}\label{eq:def-b}
b=(b_0^\top,0^\top,\ldots,0^\top)^\top,
\qquad b_0=\begin{bmatrix}x^0\\0\end{bmatrix}.
\end{equation}
The first block row imposes the two boundary conditions; the remaining
rows match the solution on adjacent intervals. Since
\(\Phi_{\rm ex}=e^{hH}\), this system gives the exact nodal values of
the continuous solution. Replacing \(\Phi_{\rm ex}\) by an approximation
\(\Phi\) gives the system \(Lq=b\) used by the algorithm.

\subsection{Horizon-uniform conditioning}
\label{sec:condition_error}

The Hamiltonian state--costate dynamics may contain modes that grow
under forward propagation. The two-point boundary conditions allow
these modes to be controlled by propagating the appropriate variables
from opposite endpoints. We make this mechanism explicit through a
Riccati--Lyapunov change of variables, following the decoupling approach
in
\cite{TrelatZhangZuazua2018Turnpike,YeZengZhang2026PeriodicDichotomy}.
The Riccati and Lyapunov solutions introduced below are used only in
the analysis; the quantum algorithm does not construct them.

Let \(P(t)\) solve
\begin{equation}\label{eq:finite-horizon-riccati}
-\dot P=A^\top P+PA-PBR^{-1}B^\top P+Q,
\qquad P(T)=G,
\end{equation}
and write \(A_c(t)=A-BR^{-1}B^\top P(t)\).
Define \(W(t)\) by
\begin{equation}\label{eq:lyapunov-W}
\dot W=A_cW+WA_c^\top+BR^{-1}B^\top,
\qquad W(0)=0.
\end{equation}
The variables
\[
z=\lambda-Px,
\qquad
v=x+Wz
\]
satisfy
\begin{equation}\label{eq:vz-decoupled-dynamics}
\dot v=A_c(t)v,
\qquad
\dot z=-A_c(t)^\top z,
\qquad
v(0)=x^0,
\quad
z(T)=0.
\end{equation}
If \(\Psi(t,s)\) denotes the propagator of
\(\dot v=A_c(t)v\), then, for \(0\leq s\leq t\leq T\),
\[
v(t)=\Psi(t,s)v(s),
\qquad
z(s)=\Psi(t,s)^\top z(t).
\]
The backward-propagation identity is established in
Lemma~\ref{lem:propagation-of-z} of
Appendix~\ref{app:transformed-one-step-map}.
Thus \(v\) is propagated forward from the initial endpoint and \(z\)
backward from the terminal endpoint.

Under Assumption~\ref{ass:sta&det},
Propositions~\ref{prop:finite-horizon-closed-loop-stability}
and~\ref{prop:horizon-uniform-riccati-bound} in
Appendix~\ref{app:horizon-uniform-bounds} give constants
\(M_c\geq1\), \(\alpha_c>0\), and \(M_P\geq0\), independent of \(T\),
such that
\[
\|\Psi(t,s)\|\leq M_c e^{-\alpha_c(t-s)},
\qquad
\sup_{t\in[0,T]}\|P(t)\|\leq M_P.
\]
Explicit choices of these constants are given in
Appendices~\ref{app:ad-M_c-alpha_c}
and~\ref{app::horizon-uniform-riccati-bound}.
The propagation in each of the two directions is therefore
exponentially decaying, with constants that remain valid as the
horizon varies.

The change of variables is
\begin{equation}\label{eq:riccati-lyapunov-transform}
\begin{bmatrix}v\\z\end{bmatrix}
=
\mathcal S(t)
\begin{bmatrix}x\\\lambda\end{bmatrix},
\qquad
\mathcal S(t)=
\begin{bmatrix}
I_n-W(t)P(t)&W(t)\\
-P(t)&I_n
\end{bmatrix}.
\end{equation}
Its factorization
\[
\mathcal S(t)=
\begin{bmatrix}I_n&W(t)\\0&I_n\end{bmatrix}
\begin{bmatrix}I_n&0\\-P(t)&I_n\end{bmatrix}
\]
shows that it is invertible.
The transformation and the decoupled dynamics are proved in
Theorem~\ref{thm:riccati-lyapunov-decoupling} of
Appendix~\ref{app:riccati-lyapunov-proof}.

Although the optimal solution has \(z(t)=0\), a bound on
\(\|L_{\rm ex}^{-1}\|\) must control arbitrary right-hand sides,
including perturbations of the boundary and matching equations.
Both transformed components are therefore needed.
Let \(\Psi_k=\Psi(t_k,t_{k-1})\).
Applying \(\mathcal S(t_k)\) to the nodal variables and the
corresponding matching rows gives the interior residuals
\[
v_k-\Psi_kv_{k-1},
\qquad
z_k-\Psi_k^{-\top}z_{k-1}.
\]
Multiplying the second residual by \(\Psi_k^\top\) puts it in
backward form,
\[
\Psi_k^\top z_k-z_{k-1}.
\]
The boundary residuals become \(v_0\) and \(z_{m-1}\).
After grouping the variables and equations, the transformed matrix
is block diagonal, with blocks denoted by \(L_v\) and \(L_z\).
The precise row and variable transformations are given in
Appendix~\ref{app:transformed-one-step-map}, culminating in the
factorization~\eqref{eq:L-diag-relationship}.

The forward and backward solution operators have kernels
\(\Psi(t_i,t_j)\) for \(i\geq j\) and
\(\Psi(t_j,t_i)^\top\) for \(i\leq j\), respectively.
Their exponential decay controls the operator norms through a
geometric sum:
\begin{equation}\label{eq:Lv-Lz-inverse-bound}
\max\{\|L_v^{-1}\|,\|L_z^{-1}\|\}
\leq M_c\gamma_h,
\qquad
\gamma_h:=\sum_{j=0}^{m-1}e^{-\alpha_c jh}.
\end{equation}
Lemma~\ref{lem:Lv-Lz-inverse-bounds} in
Appendix~\ref{app:triangular-inverse-bounds-proof} proves this
estimate by applying weighted Cauchy--Schwarz to the forward and
backward recursions.

It remains to bound the transformations uniformly in the horizon.
Following the proof of Theorem~\ref{thm:inverse-bound-L} in
Appendix~\ref{app:inverse-bound-L},
equation~\eqref{eq:lyapunov-W} gives
\[
W(t)=\int_0^t
\Psi(t,s)BR^{-1}B^\top\Psi(t,s)^\top\,\mathrm ds,
\qquad
\|W(t)\|\leq
\frac{M_c^2\|B\|^2\|R^{-1}\|}{2\alpha_c}.
\]
The triangular factorization of \(\mathcal S(t)\) implies
\[
\max\{\|\mathcal S(t)\|,\|\mathcal S(t)^{-1}\|\}
\leq
(1+\|P(t)\|)(1+\|W(t)\|)
\leq a_S,
\]
where
\[
a_S:=
(1+M_P)
\left(
1+\frac{M_c^2\|B\|^2\|R^{-1}\|}{2\alpha_c}
\right).
\]
The inverse variable transformation and the matching-row
transformation each have norm at most \(a_S\).
The additional backward-row factors satisfy
\(\max\{1,\|\Psi_k\|\}\leq M_c\).
Combining these bounds with
\eqref{eq:Lv-Lz-inverse-bound} yields
\begin{equation}\label{eq:L-inverse-bound-explicit}
\|L_{\rm ex}^{-1}\|
\leq M_c^2 a_S^2\gamma_h
=M\gamma_h
=M\frac{1-e^{-\alpha_c hm}}{1-e^{-\alpha_c h}},
\end{equation}
where \(M\) is the constant defined in
\eqref{eq:def-M}.
Theorem~\ref{thm:norm-bounds-L} in
Appendix~\ref{app:norm-bounds-L} further gives
\[
\|L_{\rm ex}\|
\leq
1+\max\{\|e^{hH}\|,\sqrt{1+\|G\|^2}\}.
\]
Together, these estimates prove the exact-system conditioning bound
in Theorem~\ref{thm:main-conditioning}.

The dependence on the grid is explicit:
\[
\gamma_h
\leq
\min\left\{m,\frac{1}{1-e^{-\alpha_c h}}\right\}
\leq
\min\left\{m,1+\frac{1}{\alpha_c h}\right\}.
\]
Consequently, for fixed system matrices and fixed \(h>0\), the
condition-number bound is uniform in \(T\).
Under grid refinement, the same analysis gives
\(\|L_{\rm ex}^{-1}\|=\mathcal O(h^{-1})\) and
\(\kappa(L_{\rm ex})=\mathcal O(h^{-1})\) as \(h\to0\);
see also \eqref{eq:condition-number-order}.
This distinction motivates the accuracy strategy in the next
subsection: keep the temporal grid fixed and increase the
approximation order.
Dependence on the state dimension remains encoded in the system
norms and stability constants; the mass--spring case is treated in
Section~\ref{subsec:dimension-uniform} and
Appendix~\ref{app:Mc-alphac-Mp-bound}.

\subsection{Gauss discretization and perturbation bounds}
\label{subsec:gauss-perturbation}

We keep the temporal grid specified in
Section~\ref{sec:main-results} fixed and approximate \(e^{hH}\)
by the \(p\)-stage Gauss--Legendre Runge--Kutta method
\cite{wu2024structure}.
The stage number \(p\) will be chosen to meet the required
accuracy.
Let \(c_i\in(0,1)\) be the Gauss nodes and \(\ell_j\) their
Lagrange basis polynomials.
The coefficient matrix and weights are
\begin{equation}\label{eq:def-A-r}
\mathcal A_{ij}=\int_0^{c_i}\ell_j(s)\,\mathrm ds,
\qquad
r_j=\int_0^1\ell_j(s)\,\mathrm ds,
\qquad
0\leq i,j<p.
\end{equation}
Set
\(
\mathcal G_h=I_p\otimes I_{2n}-\mathcal A\otimes hH.
\)
Eliminating the stage derivatives gives the one-step map
\begin{equation}\label{eq:rk-propagator-definition}
\Phi=I_{2n}
+(r^\top\otimes I_{2n})\mathcal G_h^{-1}
(\mathbf 1_p\otimes hH),
\end{equation}
where \(\mathbf 1_p=(1,\ldots,1)^\top\).
The stage equations leading to this formula are given in
Appendix~\ref{app:gauss-stage-system}.

The coefficient bounds in
Lemmas~\ref{lem:uniform-bou-pr}
and~\ref{lem:mathcal-A-norm-bound} give
\[
\|\mathcal A\|
\leq p\|\mathcal A\|_{\max}
\leq2\pi,
\qquad
\sqrt p\,\|r\|\leq\sqrt{4\pi/3}.
\]
Consequently, the condition \(h\|H\|<1/(4\pi)\) imposed in
Section~\ref{sec:main-results} ensures
\[
\|\mathcal G_h\|<\frac32,
\qquad
\|\mathcal G_h^{-1}\|<2,
\qquad
\kappa(\mathcal G_h)<3,
\qquad
\|\Phi\|\leq2,
\]
uniformly in \(p\), see
Lemma~\ref{lem:gauss-rk-propagator-v2}.
The same coefficient estimates also control the normalizations
used in the quantum construction:
\[
\alpha_{\mathcal A}
=
2^{\lceil\log_2p\rceil}\|\mathcal A\|_{\max}
\leq4\pi,
\qquad
\|r\|\,\|\mathbf 1_p\|
=
\sqrt p\,\|r\|
\leq\sqrt{4\pi/3}.
\]
The first bound is \eqref{eq:gauss-coefficient-normalization};
the corresponding block-encodings are constructed in
Appendix~\ref{app:auxiliary-factors} and used in
Appendix~\ref{app:be-Phi-construction}.
These bounds allow the stage number to increase while the
step-size restriction, stage-system conditioning, and these
normalization factors remain independent of \(p\).

The Gauss map preserves the symplectic form associated with the
Hamiltonian matrix \(H\):
\[
\Phi^\top\mathbb J_{2n}\Phi=\mathbb J_{2n}.
\]
This property is proved in Lemma~\ref{lem:gauss-symplecticity}.
Its approximation accuracy follows from a separate property:
the Gauss stability function is the diagonal \([p/p]\)
Pad\'e approximant of the exponential.
Lemma~\ref{lem:rk-global-error-v2} uses the matrix
remainder formula to obtain
\begin{equation}\label{eq:rk-one-step-error}
\Delta_h:=\|\Phi-\Phi_{\rm ex}\|
\leq
\frac{8(p!)^2}{(2p)!(2p+1)!}
(h\|H\|)^{2p+1}.
\end{equation}
This estimate holds for \(h\|H\|<1\) and does not require
\(H\) to be normal or diagonalizable.
The stronger step-size restriction above additionally supplies
the uniform stage-system bounds needed for quantum implementation.

We next transfer the conditioning estimate from \(L_{\rm ex}\)
to \(L\).
Let
\(
S_m=\sum_{k=1}^{m-1}E_{k,k-1},
\)
where
\(
E_{ij}=\ket{i}^{(m)}\bra{j}^{(m)}.
\)
is a \(m\times m\) matrix unit.
Since the boundary row is unchanged, the propagation blocks give
\begin{equation}\label{eq:shooting-matrix-perturbation}
L-L_{\rm ex}
=
-S_m\otimes(\Phi-\Phi_{\rm ex}),
\qquad
\|L-L_{\rm ex}\|=\Delta_h,
\end{equation}
where we used \(\|S_m\|=1\).
Thus the perturbation of the assembled matrix is controlled by
the one-step error without an additional factor \(m\). Define
\(
\theta:=\Delta_h\|L_{\rm ex}^{-1}\|.
\)
If \(\theta<1\), the factorization
\[
L=L_{\rm ex}
\bigl(I+L_{\rm ex}^{-1}(L-L_{\rm ex})\bigr)
\]
and the Neumann series imply that \(L\) is nonsingular and
\[
\|L^{-1}\|
\leq
\frac{\|L_{\rm ex}^{-1}\|}{1-\theta}.
\]
In particular, if \(\theta\leq1/2\), then
\[
\|L^{-1}\|\leq2\|L_{\rm ex}^{-1}\|,
\qquad
\kappa(L)\leq3\kappa(L_{\rm ex}).
\]
These statements are proved in
Theorem~\ref{thm:approx-shooting-conditioning}. Together with the preceding subsection, they complete the proof
of Theorem~\ref{thm:main-conditioning}.

For the nodal solution, the identity
\(
q-q_{\rm ex}
=
L^{-1}(L_{\rm ex}-L)q_{\rm ex}
\)
gives
\[
\|q-q_{\rm ex}\|
\leq
\frac{\theta}{1-\theta}\|q_{\rm ex}\|.
\]
This is the estimate~\eqref{eq:appendix-q-qhat-first-bound}
in Appendix~\ref{app:global-error}.
When \(\theta\leq1/2\), this yields
\begin{equation}\label{eq:overview-nodal-error}
\|q-q_{\rm ex}\|
\leq
2\Delta_h\|L_{\rm ex}^{-1}\|^2\|x^0\|.
\end{equation}
A further estimate exploiting the exponential decay of the
optimal trajectory is given in Theorem~\ref{thm:global-error}.

\subsection{Quantum implementation and output estimation}
\label{sec:quantum-algorithm}

Both requested outputs can be expressed through matrix elements
of operators constructed from \(L^{-1}\). We first construct
block-encodings of these operators and then estimate the required
matrix elements using Hadamard tests and amplitude estimation.

\paragraph{Inverse block-encoding.}
We use the temporal grid specified in
Section~\ref{sec:main-results}. Let
\[
K_{\rm ex}
=
M\frac{1-e^{-\alpha_chm}}{1-e^{-\alpha_ch}}
\]
be the bound on \(\|L_{\rm ex}^{-1}\|\) established in
Theorem~\ref{thm:main-conditioning}. The algorithms take the
required norm and stability bounds as classical input.
Complexity estimates expressed in terms of exact norms assume
that constant-factor upper bounds are available; otherwise,
the supplied bounds replace the corresponding norm parameters.

We choose the Gauss stage number \(p\) so that
\[
\Delta_h K_{\rm ex}\leq\frac12.
\]
The perturbation estimate then gives
\[
\|L^{-1}\|\leq 2K_{\rm ex},
\qquad
\|q-q_{\rm ex}\|
\leq 2\Delta_hK_{\rm ex}^2\|x^0\|.
\]
We therefore use the available inversion bound
\[
\nu:=2K_{\rm ex}.
\]

Under Assumption~\ref{ass:be}, we first construct a block-encoding
of \(R^{-1}\), then combine it with those of \(A,B,Q\) to encode
\(H\). Equation~\eqref{eq:rk-propagator-definition} reduces the
construction of \(\Phi\) to the inversion of \(\mathcal G_h\).
The time shift and the boundary blocks then give a block-encoding
of \(L\). These constructions are given in
Appendices~\ref{app:be-H-construction},
\ref{app:be-Phi-construction}, and~\ref{app:be-L-construction}.

Set
\[
\Gamma_L=s^2\kappa(R)+s\|G\|_{\max}+1.
\]
The normalization of the block-encoding of \(L\) constructed in Appendix~\ref{app:be-L-construction} satisfies
\(\alpha_L=\mathcal O(\Gamma_L)\). Applying
Lemma~\ref{lem:inverse-be-approx}, with sufficiently accurate
input encodings, gives an
\((\alpha_{L^{-1}},a_{L^{-1}},\varepsilon_{L^{-1}})\)-block-encoding
\(U_{L^{-1}}\), where
\[
\alpha_{L^{-1}}=8\nu,
\qquad
0<\varepsilon_{L^{-1}}\leq\nu.
\]
One use of \(U_{L^{-1}}\) requires
\begin{equation}
\label{eq:overview-inverse-query-cost}
Q_{\rm inv}
=
\widetilde{\mathcal O}
\bigl(\Gamma_Ls^3\kappa(R)^2\nu\bigr)
\end{equation}
queries to the input matrix oracles. The suppressed factors
include logarithmic dependence on
\(\varepsilon_{L^{-1}}^{-1}\). The precise internal accuracies
are given in Appendix~\ref{app:be-L-construction}.

\paragraph{Output quantities and error allocation.}
Let \(S_{\lambda,k}\in\mathbb R^{n\times2nm}\) select the costate
at node \(k\):
\[
S_{\lambda,k}
=
e_k^\top\otimes[\,0\ \ I_n\,],
\qquad k=0,\ldots,m-1,
\]
where \(e_k\) is the \(k\)th standard basis vector of
\(\mathbb R^m\). The discrete outputs are
\begin{equation}
\label{eq:overview-discrete-outputs}
J_h
=
\frac12(x^0)^\top S_{\lambda,0}L^{-1}b,
\qquad
u_h(t_k)
=
-R^{-1}B^\top S_{\lambda,k}L^{-1}b.
\end{equation}

We allocate \(\varepsilon/3\) to each of the discretization,
block-encoding, and estimation errors. The discretization
errors satisfy
\begin{equation}
\label{eq:overview-output-discretization}
\begin{aligned}
|J_h-J^*|
&\leq
\frac{\|x^0\|}{2}\|q-q_{\rm ex}\|
\leq
\Delta_hK_{\rm ex}^2\|x^0\|^2,
\\
\|u_h(t_k)-u^*(t_k)\|
&\leq
\|R^{-1}B^\top\|\|q-q_{\rm ex}\|
\\
&\leq
2\|R^{-1}B^\top\|
\Delta_hK_{\rm ex}^2\|x^0\|.
\end{aligned}
\end{equation}
These inequalities determine the additional accuracy required
of the Gauss approximation for each task.

For a unitary \(U\) with \(a\) block-encoding ancillas, write
\[
\mathcal B(U)
=
(\bra0_a\otimes I)U(\ket0_a\otimes I).
\]
Thus \(\alpha\mathcal B(U)\) approximates the encoded operator
when \(U\) has normalization \(\alpha\). Zero padding is
understood in the matrix-element expressions below.

\paragraph{Optimal-cost estimation.}
Suppose first that \(\|x^0\|>0\), and let
\[
c=S_{\lambda,0}^\top x^0.
\]
Since
\[
\|b\|=\|c\|=\|x^0\|,
\qquad
\ket b=\frac{b}{\|x^0\|},
\qquad
\ket c=\frac{c}{\|x^0\|},
\]
the estimated cost can be written as
\[
J_h
=
\frac{\|x^0\|^2}{2}\bra c L^{-1}\ket b.
\]
Both states can be prepared with one call to \(U_{x^0}\).
Preparing \(\ket b\) adds zero-initialized registers, while
preparing \(\ket c\) also shifts the populated indices into
the initial-costate block; see
Appendix~\ref{app:optimal-cost-estimation-proof-v2}.

The encoded quantity is
\[
\widetilde J_h
=
\frac{\alpha_{L^{-1}}\|x^0\|^2}{2}
\operatorname{Re}
\bigl[\bra c\,\mathcal B(U_{L^{-1}})\ket b\bigr].
\]
Its encoding error satisfies
\[
|\widetilde J_h-J_h|
\leq
\frac{\|x^0\|^2}{2}\varepsilon_{L^{-1}}.
\]
Choosing
\[
\varepsilon_{L^{-1}}
=
\min\left\{
\nu,\frac{2\varepsilon}{3\|x^0\|^2}
\right\}
\]
therefore makes this contribution at most \(\varepsilon/3\).
The remaining estimation error is controlled by the precision
of the matrix-element estimate, as specified in
Algorithm~\ref{alg:optimal-cost}.

\begin{algorithm}[!h]
\KwIn{The input oracles and \(\|x^0\|\) from
Assumption~\ref{ass:be}; the norm and stability bounds described
above; \(T>0\), \(\varepsilon>0\), and \(0<\delta<1\).}
\KwOut{An estimate \(J^*_{\rm num}\) such that
\(\Pr(|J^*_{\rm num}-J^*|\leq\varepsilon)\geq1-\delta\).}

\If{\(\|x^0\|=0\)}{
    \KwRet{\(J^*_{\rm num}=0\)}\;
}

Choose \(m,h\) as in Section~\ref{sec:main-results},
evaluate \(K_{\rm ex}\), and set \(\nu=2K_{\rm ex}\)\;

\If{\(\varepsilon\geq K_{\rm ex}\|x^0\|^2/2\)}{
    \KwRet{\(J^*_{\rm num}=0\)}\;
}

Using the Gauss approximation bound in
\eqref{eq:rk-one-step-error}, choose \(p\) so that
\[
\Delta_hK_{\rm ex}\leq\frac12,
\qquad
\Delta_hK_{\rm ex}^2\|x^0\|^2
\leq\frac{\varepsilon}{3}
\]
\;

Set
\[
\alpha_{L^{-1}}=8\nu,
\qquad
\varepsilon_{L^{-1}}
=
\min\left\{
\nu,\frac{2\varepsilon}{3\|x^0\|^2}
\right\}
\]
\;

Construct \(U_H,U_\Phi,U_L\), and use \gls{qsvt} to obtain
\(U_{L^{-1}}\) with normalization \(\alpha_{L^{-1}}\) and
error at most \(\varepsilon_{L^{-1}}\), using the internal
accuracies from Appendix~\ref{app:be-L-construction}\;

Set \(b\) as in \eqref{eq:def-b} and
\(c=S_{\lambda,0}^\top x^0\); prepare \(\ket b\) and \(\ket c\)\;

Set
\[
\eta_J
=
\frac{2\varepsilon}
{3\alpha_{L^{-1}}\|x^0\|^2}
\]
\;

Use a Hadamard test and amplitude estimation to estimate
\[
\mu_J
=
\operatorname{Re}
\bigl[\bra c\,\mathcal B(U_{L^{-1}})\ket b\bigr]
\]
to additive error \(\eta_J\) and failure probability
\(\delta\); denote the result by \(\widehat\mu_J\)\;

\KwRet{\(
J^*_{\rm num}
=
\dfrac{\alpha_{L^{-1}}\|x^0\|^2}{2}\widehat\mu_J
\)}\;

\caption{Quantum estimation of the optimal cost}
\label{alg:optimal-cost}
\end{algorithm}

The early return is justified by
\[
|J^*|
\leq
\frac12K_{\rm ex}\|x^0\|^2.
\]
On the branch requiring quantum estimation, the choice of
\(\eta_J\) gives
\[
|J^*_{\rm num}-\widetilde J_h|
\leq
\frac{\alpha_{L^{-1}}\|x^0\|^2}{2}\eta_J
=
\frac{\varepsilon}{3}
\]
with probability at least \(1-\delta\). Consequently,
\[
\begin{aligned}
|J^*_{\rm num}-J^*|
\leq
|J^*_{\rm num}-\widetilde J_h|
+
|\widetilde J_h-J_h|
+
|J_h-J^*|
\leq\varepsilon.
\end{aligned}
\]

\paragraph{Optimal-control estimation.}
Let \(\rho_R\geq\|R^{-1}\|\) be the supplied inverse-norm bound,
and set
\[
\alpha_B=s\|B\|_{\max},
\qquad
\beta=\alpha_B\rho_R.
\]
Then \(\|R^{-1}B^\top\|\leq\beta\). For the selected node \(k\),
define
\[
C_k=R^{-1}B^\top S_{\lambda,k}L^{-1},
\qquad
u_h(t_k)=-C_kb.
\]

Combining the inverse encoding of \(L\), the exact encodings
of \(S_{\lambda,k}\) and \(B^\top\), and an inverse encoding
of \(R\) with normalization \(8\rho_R\), gives an encoding
\(U_{C_k}\) with normalization
\[
\alpha_{C_k}
=
8\alpha_B\rho_R\alpha_{L^{-1}}
=
64\beta\nu.
\]
All factors are extended to a common system register before
multiplication, as described in
Appendix~\ref{app:optimal-control-estimation-proof-v2}.

The precision of the inverse encoding of \(R\) used here is
chosen separately from that used in constructing \(H\).
If its error is \(\varepsilon_{R^{-1}}^{(u)}\), the product
encoding satisfies
\[
\varepsilon_{C_k}
\leq
\alpha_B\alpha_{L^{-1}}\varepsilon_{R^{-1}}^{(u)}
+
8\beta\varepsilon_{L^{-1}}.
\]
The choices in Algorithm~\ref{alg:optimal-control} ensure
\[
\varepsilon_{C_k}
\leq
\frac{\varepsilon}{3\|x^0\|}.
\]

Let \(\widetilde C_k\) be the leading
\(d\times2nm\) block of
\(\alpha_{C_k}\mathcal B(U_{C_k})\), and set
\[
\widetilde u_h=-\operatorname{Re}(\widetilde C_kb).
\]
The encoding error is then bounded directly in Euclidean norm:
\[
\|\widetilde u_h-u_h(t_k)\|
\leq
\|\widetilde C_k-C_k\|\|b\|
\leq
\varepsilon_{C_k}\|x^0\|
\leq
\frac{\varepsilon}{3}.
\]
The algorithm estimates the \(d\) components of
\(\widetilde u_h\) separately.

\begin{algorithm}[!h]
\KwIn{The input oracles and \(\|x^0\|\) from
Assumption~\ref{ass:be}; the norm and stability bounds described
above; \(T>0\), a valid grid index \(k\),
\(\varepsilon>0\), and \(0<\delta<1\).}
\KwOut{An estimate \(u_{\rm num}(t_k)\in\mathbb R^d\) such that
\(\Pr(\|u_{\rm num}(t_k)-u^*(t_k)\|\leq\varepsilon)
\geq1-\delta\).}

\If{\(\|x^0\|=0\) or \(B=0\)}{
    \KwRet{\(u_{\rm num}(t_k)=0\in\mathbb R^d\)}\;
}

Choose \(m,h\) as in Section~\ref{sec:main-results},
evaluate \(K_{\rm ex}\), and set
\[
\nu=2K_{\rm ex},
\qquad
\beta=\alpha_B\rho_R,
\qquad
\alpha_{L^{-1}}=8\nu
\]
\;

\If{\(\varepsilon\geq\beta K_{\rm ex}\|x^0\|\)}{
    \KwRet{\(u_{\rm num}(t_k)=0\in\mathbb R^d\)}\;
}

Using the Gauss approximation bound in
\eqref{eq:rk-one-step-error}, choose \(p\) so that
\[
\Delta_hK_{\rm ex}\leq\frac12,
\qquad
2\beta\Delta_hK_{\rm ex}^2\|x^0\|
\leq\frac{\varepsilon}{3}
\]
\;

Set
\[
\varepsilon_{L^{-1}}
=
\min\left\{
\nu,\frac{\varepsilon}{48\beta\|x^0\|}
\right\},
\qquad
\varepsilon_{R^{-1}}^{(u)}
=
\min\left\{
\rho_R,
\frac{\varepsilon}
{6\alpha_B\alpha_{L^{-1}}\|x^0\|}
\right\}
\]
\;

Construct \(U_H,U_\Phi,U_L\), and use \gls{qsvt} to obtain
\(U_{L^{-1}}\) with normalization \(\alpha_{L^{-1}}\) and
error at most \(\varepsilon_{L^{-1}}\)\;

Construct the readout encoding of \(R^{-1}\) with
normalization \(8\rho_R\) and error at most
\(\varepsilon_{R^{-1}}^{(u)}\). Combine it with the encodings
of \(B^\top,S_{\lambda,k},L^{-1}\) to obtain \(U_{C_k}\),
with normalization \(\alpha_{C_k}=64\beta\nu\)\;

Set \(b\) as in \eqref{eq:def-b} and prepare \(\ket b\)\;

Set
\[
\eta_u
=
\frac{\varepsilon}
{3\sqrt d\,\alpha_{C_k}\|x^0\|}
\]
\;

\For{\(j=0,\ldots,d-1\)}{
    Use a Hadamard test and amplitude estimation to estimate
    \[
    \mu_{k,j}
    =
    \operatorname{Re}
    \bigl[\bra j\,\mathcal B(U_{C_k})\ket b\bigr]
    \]
    to additive error \(\eta_u\) and failure probability
    \(\delta/d\); denote the result by
    \(\widehat\mu_{k,j}\)\;

    Set
    \[
    (u_{\rm num}(t_k))_j
    =
    -\alpha_{C_k}\|x^0\|\widehat\mu_{k,j}
    \]
    \;
}

\KwRet{\(u_{\rm num}(t_k)\)}\;

\caption{Quantum estimation of the optimal control at \(t_k\)}
\label{alg:optimal-control}
\end{algorithm}

The early return follows from
\[
\|u^*(t_k)\|
\leq
\beta K_{\rm ex}\|x^0\|.
\]
Otherwise, the union bound ensures that all \(d\) matrix-element
estimates satisfy their error bounds with probability at least
\(1-\delta\). On this event,
\[
\|u_{\rm num}(t_k)-\widetilde u_h\|
\leq
\alpha_{C_k}\|x^0\|\sqrt d\,\eta_u
=
\frac{\varepsilon}{3}.
\]
Adding the encoding and discretization errors gives
\[
\begin{aligned}
\|u_{\rm num}(t_k)-u^*(t_k)\|
\leq
\|u_{\rm num}(t_k)-\widetilde u_h\|
+
\|\widetilde u_h-u_h(t_k)\|
+
\|u_h(t_k)-u^*(t_k)\|
\leq\varepsilon.
\end{aligned}
\]

\paragraph{Complexity.}
On the branches requiring quantum estimation,
Lemma~\ref{lem:Hadamard-test} gives the respective readout
call counts
\[
\begin{aligned}
N_J
&=
\mathcal O\left(
\frac{\nu\|x^0\|^2}{\varepsilon}
\log\frac{2}{\delta}
\right),
\\
N_u
&=
\mathcal O\left(
\frac{d^{3/2}\beta\nu\|x^0\|}{\varepsilon}
\log\frac{2d}{\delta}
\right).
\end{aligned}
\]
These counts include calls to the corresponding operator
encodings, state-preparation unitaries, and their inverses.
For the control, the factor \(d^{3/2}\) comes from estimating
\(d\) components, each to accuracy
\(\varepsilon/(3\sqrt d)\).

One call to \(U_{C_k}\) uses one call to \(U_{L^{-1}}\),
together with the encodings of \(B^\top\), \(S_{\lambda,k}\),
and \(R^{-1}\). The additional matrix-oracle cost is bounded
by the same asymptotic expression as
\eqref{eq:overview-inverse-query-cost}; see
Appendix~\ref{app:optimal-control-estimation-proof-v2}.
Multiplying the readout call counts by \(Q_{\rm inv}\) gives
the total query bounds. In particular, one factor of \(\nu\)
comes from constructing the inverse encoding, and a second
comes from its normalization in the readout.

For the grid in Section~\ref{sec:main-results},
\[
\nu
=
\mathcal O\bigl(M(1+\alpha_c^{-1}\|H\|)\bigr).
\]
Substituting this estimate gives the query bounds in
Theorems~\ref{thm:main-cost} and~\ref{thm:main-control}.
The additional gate counts have the further factor
\[
\log(2d)\log^2(2n)\log^2(2m),
\]
up to the suppressed logarithmic factors. The Runge--Kutta
coefficient oracles have implementation costs polynomial in
\(p\). Since the required stage number grows at most
logarithmically with the reciprocal approximation accuracy,
these costs are included in the suppressed factors.
The detailed query and gate counts are given in
Appendices~\ref{app:optimal-cost-estimation-proof-v2}
and~\ref{app:optimal-control-estimation-proof-v2}.

\subsection{Dimension-uniform bounds for the mass--spring chain}
\label{subsec:dimension-uniform}

The horizon-uniform estimate alone does not control the dependence
of \(M_c,\alpha_c,M_P\) on the state dimension. We verify this
dependence for the following family.

\begin{benchmark}[Coupled mass--spring chain]
\label{bench:mass_spring_chain}
Consider \(n_m\geq2\) masses with displacements
\(\omega(t)\in\mathbb R^{n_m}\) and endpoint forces
\(u(t)\in\mathbb R^2\), governed by
\begin{equation}\label{eq:coupled-oscillator}
\mathcal M\ddot\omega+\mathcal D\dot\omega+\mathcal K\omega
=\mathcal F u,
\end{equation}
where
\[
\mathcal M=\operatorname{diag}(m_1,\ldots,m_{n_m}),
\quad \mathcal D=\delta I_{n_m},
\quad \mathcal F=[e_1,e_{n_m}],
\]
and
\[
\mathcal K=k_0I_{n_m}
+k_1\operatorname{tridiag}(-1,2,-1).
\]
Here \(e_1,e_{n_m}\) denote the first and last standard basis
vectors. Assume \(0<m_-\leq m_i\leq m_+<\infty\), with
\(m_-,m_+\) and the positive parameters
\(\delta,k_0,k_1,\mu,\gamma\) independent of \(n_m\).
With \(x=(\rho^\top,\omega^\top)^\top\) and
\(\rho=\mathcal M\dot\omega\), the first-order system is
\begin{equation}\label{eq:ode-for-oscillator}
\dot x=Ax+Bu,
\qquad
A=\begin{bmatrix}
-\delta\mathcal M^{-1}&-\mathcal K\\
\mathcal M^{-1}&0
\end{bmatrix},
\quad
B=\begin{bmatrix}\mathcal F\\0\end{bmatrix}.
\end{equation}
We take
\begin{equation}\label{eq:def-Q-benchmark}
Q=\operatorname{diag}(\mathcal M^{-1},\mathcal K),
\qquad R=\mu I_2,\qquad G=\gamma Q.
\end{equation}
Thus \(n=2n_m\), \(d=2\), and \(\|\mathcal F\|=1\).
\end{benchmark}

The proof of Corollary~\ref{cor:main-mass-spring} starts from
\[
m_-I\preceq\mathcal M\preceq m_+I,
\qquad
k_0I\preceq\mathcal K\preceq(k_0+4k_1)I.
\]
These bounds and the fixed damping give a Lyapunov energy
uniformly equivalent to \(\|x\|^2\), and hence
\[
\|e^{At}\|\leq C_0e^{-\alpha_0t},
\]
with \(C_0,\alpha_0^{-1}\) independent of \(n_m\);
see Lemma~\ref{lem:open-loop-stable}. In particular,
Assumption~\ref{ass:sta&det} holds.

Comparing the infinite-horizon optimal cost with the zero-control
cost gives a uniform upper bound for the stabilizing algebraic
Riccati solution \(P_+\). Since \(Q\) is uniformly positive
definite and \(\|A-BR^{-1}B^\top P_+\|\) is bounded, the integral
representation of \(P_+\) also gives a uniform lower bound. Thus
\[
\underline p I\preceq P_+\preceq\overline p I,
\]
where \(\underline p,\overline p>0\) are independent of \(n_m\).
The auxiliary Lyapunov estimates and the terminal correction
then give constants \(C,\alpha_*>0\), independent of \(n_m\) and
\(T\), such that
\[
M_c\leq C,\qquad M_P\leq C,\qquad \alpha_c\geq\alpha_*.
\]
The full argument and explicit constants are given in
Appendix~\ref{app:Proofs for the benchmark}.

The matrix norms and sparsity are uniformly bounded for this
family, while \(d=2\) and \(\kappa(R)=1\). Thus \(\mathcal C\)
in Section~\ref{sec:main-results} is bounded independently of
\(n_m\) and \(T\). If \(\|x^0\|\) is uniformly bounded, the
complexity bounds have polylogarithmic dependence on \(n\) and
\(T\) in the stated input model, as claimed in
Corollary~\ref{cor:main-mass-spring}. The assumption on the initial-state norm has a direct physical
interpretation. In the momentum--displacement coordinates used above,
the total initial mechanical energy is
\[
E_{\mathrm{mech}}^0
=
\frac12(\rho^0)^\top M^{-1}\rho^0
+
\frac12(\omega^0)^\top K\omega^0
=
\frac12(x^0)^\top Qx^0.
\]
The uniform bounds on the masses and stiffness matrix imply
\[
\frac{\min\{m_+^{-1},k_0\}}{2}\|x^0\|^2
\le E_{\mathrm{mech}}^0
\le
\frac{\max\{m_-^{-1},k_0+4k_1\}}{2}\|x^0\|^2.
\]
Since the constants in these inequalities are independent of the
chain length, uniform boundedness of \(\|x^0\|\) is equivalent to
uniform boundedness of the total initial mechanical energy.
Thus, for fixed additive accuracy, the stated polylogarithmic
complexity applies to families with bounded total initial energy
under the assumed input-access model. A bound on the energy per
mass alone does not ensure this condition, since the total energy
may grow with the number of masses.

\section{Numerical experiments}
\label{sec:numerics}

We examine the accuracy and conditioning of the shooting systems
used by the quantum algorithms. The experiments (i) validate the
shooting formulation against a Riccati reference, (ii) confirm the
Gauss convergence orders and perturbation bounds, (iii) show that
\(\kappa(L_{\rm ex})\) saturates as \(T\) grows at fixed \(h\) but
grows like \(h^{-1}\) under grid refinement, and (iv) illustrate how
the dependence on \(n\) is determined by the mode of dimensional growth.

\subsection{Benchmark problems and numerical setup}
\label{subsec:numerical-setup}
\label{app:benchmark}

We use the mass--spring chain, vehicle platoon, and
reaction--diffusion models shown in
Fig.~\ref{fig:benchmark-models}. In the dimension study in
Section~\ref{subsec:state_dimension_dependence}, we increase
the number of masses in the chain, the number of independent
local regulators in the platoon, or the number of spatial
grid points on a fixed domain.

\begin{figure}[!h]
\centering
\begin{minipage}[b]{0.33\linewidth}
\centering
\includegraphics[width=\linewidth]{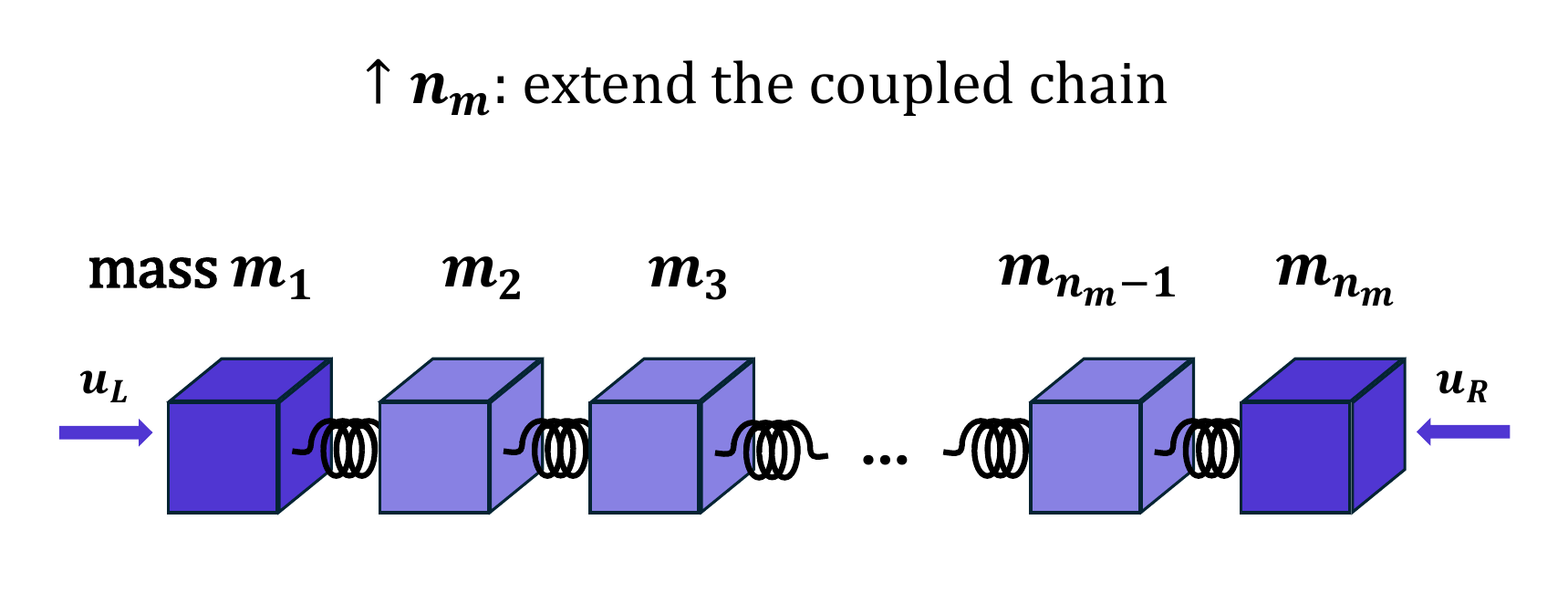}
\par\smallskip
{\footnotesize
(a) Mass--spring chain\\
\(n=2n_m\), \(d=2\)\par}
\end{minipage}\hfill
\begin{minipage}[b]{0.33\linewidth}
\centering
\includegraphics[width=\linewidth]{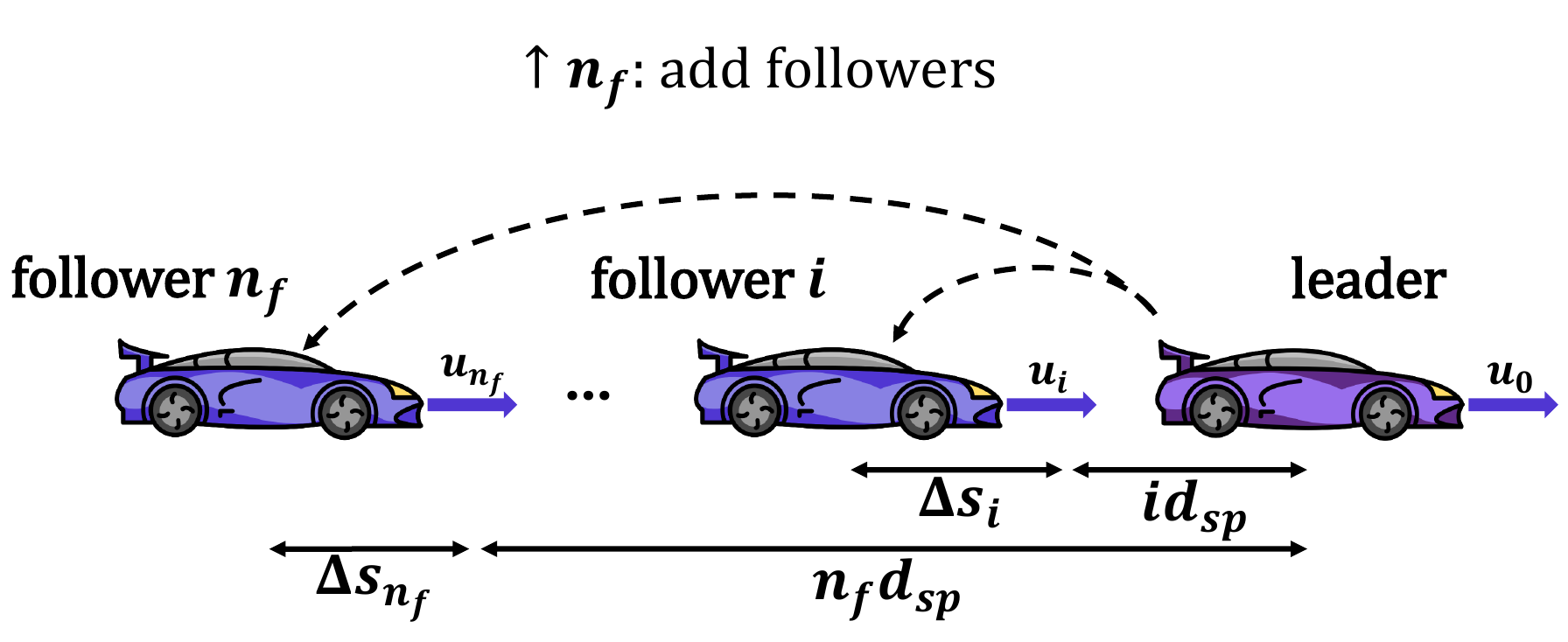}
\par\smallskip
{\footnotesize
(b) Vehicle platoon\\
\(n=3n_f\), \(d=n_f\)\par}
\end{minipage}\hfill
\begin{minipage}[b]{0.33\linewidth}
\centering
\includegraphics[width=\linewidth]{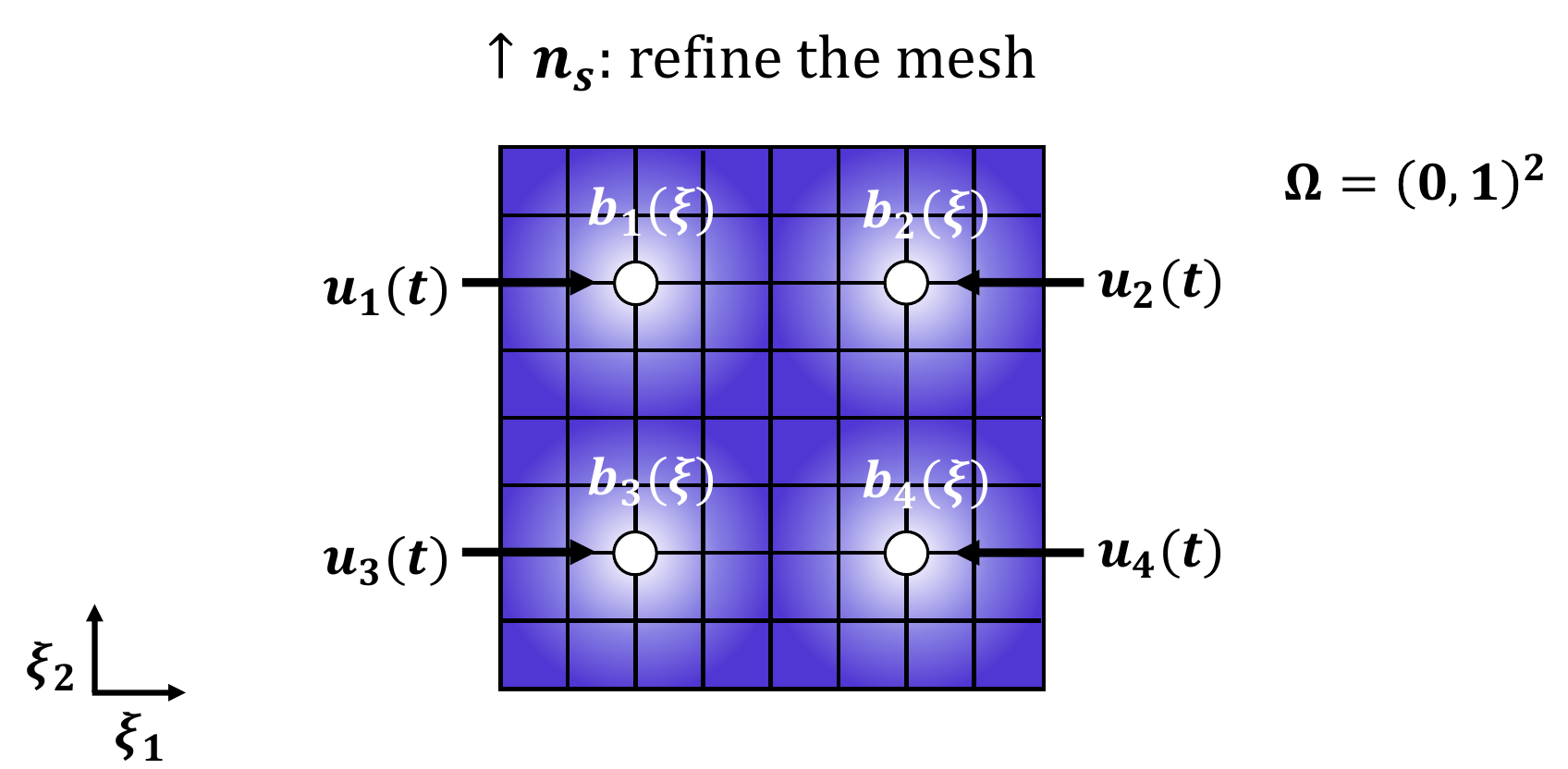}
\par\smallskip
{\footnotesize
(c) Reaction--diffusion\\
\(n=n_s^2\), \(d=4\)\par}
\end{minipage}
\caption{Three LQR benchmark problems. (a) Forces applied to the two end masses are used to bring a coupled mass--spring chain to rest at equilibrium. Damping and ground attachments are omitted from the schematic. (b) Each follower is controlled to reduce spacing, velocity, and acceleration errors relative to a prescribed leader trajectory. (c) Four distributed Gaussian inputs regulate a reaction--diffusion field toward zero; the circles mark the actuator centers.}
\label{fig:benchmark-models}
\end{figure}

The mass--spring model is defined in
Benchmark~\ref{bench:mass_spring_chain} and illustrated in
Fig.~\ref{fig:benchmark-models}(a). Its state consists of
the momenta and signed displacements from equilibrium of
\(n_m\) masses. For the representative instance, we take
six unit masses and apply controls at the two endpoints,
giving \((n,d)=(12,2)\). We set
\[
\delta=0.4,\qquad k_0=k_1=1,\qquad
\mu=0.5,\qquad \gamma=1.
\]
The initial momenta are zero, and the initial displacements are
\[
\omega_i^0=\sqrt{\frac{2}{n_m+1}}
\sin\!\left(\frac{\pi i}{n_m+1}\right),
\qquad i=1,\ldots,n_m,
\]
so that \(\|x^0\|=1\).

\begin{benchmark}[Vehicle platoon]
\label{bench:platoon}
The leader in Fig.~\ref{fig:benchmark-models}(b) supplies a
prescribed reference trajectory. Each of the \(n_f\) followers
has a three-dimensional state describing position, velocity,
and acceleration errors relative to its desired motion.
To fix the sign convention, let \(s_i,v_i,a_i\) denote these
physical quantities for vehicle \(i\), with \(i=0\) denoting
the leader. For a desired intervehicle spacing \(d_{\rm sp}\),
write
\[
\Delta s_i=s_0-s_i-i d_{\rm sp},
\qquad
x_i=\begin{bmatrix}-\Delta s_i\\v_i-v_0\\a_i-a_0\end{bmatrix}.
\]
Thus a zero error state represents the desired moving formation.
The controls are deviations from the nominal follower inputs
along the reference trajectory. The local error dynamics are
determined by
\[
A_0=
\begin{bmatrix}
0&1&0\\0&0&1\\0&0&-\tau^{-1}
\end{bmatrix},
\qquad
b_0=\begin{bmatrix}0\\0\\\tau^{-1}\end{bmatrix}.
\]
Stacking the follower states gives
\[
\begin{aligned}
A&=I_{n_f}\otimes A_0,& B&=I_{n_f}\otimes b_0,\\
Q&=I_{n_f}\otimes Q_0,& R&=I_{n_f}\otimes R_0,\\
G&=I_{n_f}\otimes G_0.&&
\end{aligned}
\]
The followers are coupled to the prescribed reference, but their
error dynamics and costs are independent of one another.
Hence \(n=3n_f\) and \(d=n_f\).
We take \(\tau=2\), \(Q_0=G_0=I_3\), and \(R_0=0.5\).
The representative instance has \(n_f=5\).
Its initial position and velocity errors are sampled uniformly from
\([-3,3]\) and \([-1,1]\), respectively, with zero acceleration
errors.
\end{benchmark}

\begin{benchmark}[Reaction--diffusion system]
\label{bench:reaction_diffusion}
We discretize a reaction--diffusion field on an
\(n_s\times n_s\) interior grid of \((0,1)^2\), as illustrated
in Fig.~\ref{fig:benchmark-models}(c). The state consists of
the field values at the interior grid points. Set
\[
\begin{aligned}
h_s&=(n_s+1)^{-1},\qquad
D_1=h_s^{-2}\operatorname{tridiag}(1,-2,1),\\
D_2&=I_{n_s}\otimes D_1+D_1\otimes I_{n_s},\qquad
A=\nu_{\rm d} D_2+\rho_{\rm r} I_n.
\end{aligned}
\]
Here \(n=n_s^2\), and the discrete Laplacian imposes homogeneous
Dirichlet boundary conditions. Each of the four controls acts
through a Gaussian spatial profile with standard deviation \(0.10\).
The set of actuator centers is
\[
\{(0.2,0.2),\ (0.2,0.8),\ (0.8,0.2),\ (0.8,0.8)\}.
\]
The columns of \(B\) are the sampled profiles, each normalized
to unit Euclidean norm. We take
\[
\nu_{\rm d}=0.01,\qquad \rho_{\rm r}=0.15,\qquad
Q=I_n,\qquad R=0.3I_4,\qquad G=I_n.
\]
These state weights act on the unscaled vector of nodal values.
The representative instance has \(n_s=5\), hence \((n,d)=(25,4)\).
The initial state samples
\[
\zeta(\xi,0)=\sin(\pi\xi_1)\sin(\pi\xi_2)
+0.25\sin(2\pi\xi_1)\sin(3\pi\xi_2)
\]
at the interior grid points.
\end{benchmark}

Unless the horizon is varied, we use \(T=10\).
We write \(\vartheta=h\|H\|\) for the normalized temporal step.
Given a target \(\vartheta_{\rm tar}\), the fixed-horizon grid is
\begin{equation}\label{eq:numerical-grid-choice}
m=1+\left\lceil\frac{T\|H\|}{\vartheta_{\rm tar}}\right\rceil,
\qquad h=\frac{T}{m-1},
\end{equation}
so that \(\vartheta\leq\vartheta_{\rm tar}\).
The fixed-target comparisons use \(0.05\) or \(0.075\), both
below the sufficient block-encoding step bound \(1/(4\pi)\).
The algorithmic grid in Section~\ref{sec:main-results} uses the
more conservative choice \(h\|H\|\leq1/(8\pi)\).

\subsection{Accuracy of the shooting discretization}
\label{subsec:numerical-accuracy}

\paragraph{Consistency with the Riccati formulation.}
We first check the shooting formulation using \(e^{hH}\)
with that obtained by numerically
solving the differential Riccati equation. We integrate
\eqref{eq:finite-horizon-riccati} backward and then integrate
the closed-loop state equation forward. The reference costate and control are
\[
\lambda_{\rm ref}(t)=P(t)x_{\rm ref}(t),
\qquad
u_{\rm ref}(t)=-R^{-1}B^\top\lambda_{\rm ref}(t).
\]
The reference cost is evaluated from the original objective
\eqref{eq: LQR payoff}. The shooting outputs are obtained
independently from \eqref{eq:overview-discrete-outputs}, with
\(L\) replaced by \(L_{\rm ex}\).

All three representative instances use \(m=21\) for this comparison.
For \(v\in\{x,\lambda,u\}\), let \(\widehat v_j\) denote
the values recovered from the numerical shooting solution and set
\[
E_v=
\frac{\left(\sum_{j=0}^{m-1}
\|\widehat v_j-v_{\rm ref}(t_j)\|^2\right)^{1/2}}
{\left(\sum_{j=0}^{m-1}\|v_{\rm ref}(t_j)\|^2\right)^{1/2}}.
\]
The state, costate, and control discrepancies are below
\(1.8\times10^{-12}\). Both the relative cost discrepancy
\(|\widehat J_{\rm ex}-J_{\rm ref}|/|J_{\rm ref}|\) and the
normalized residual
\(\|L_{\rm ex}\widehat q_{\rm ex}-b\|/\|b\|\)
are below \(3.4\times10^{-15}\).
The close agreement between the two numerical solutions provides
a numerical check of the shooting formulation and the formulas
for the optimal cost and control.

\paragraph{Gauss discretization.}
We next replace \(e^{hH}\) by the \(p\)-stage Gauss--Legendre
approximation \(\Phi_p(h)\) introduced in
Section~\ref{subsec:gauss-perturbation}.
Let the history vectors \(q_{\rm ex}\) and \(q\) collect the nodal states and costates obtained with the exact and approximate propagators on the same grid.  With
\(\mathbb J_{2n}
=\bigl[\begin{smallmatrix}
0&I_n\\-I_n&0
\end{smallmatrix}\bigr]\), define
\begin{equation}\label{eq:numerical-gauss-errors}
\begin{aligned}
\Delta_h^{(p)}
&=\|\Phi_p(h)-e^{hH}\|,
&
E_q
&=\frac{\|q-q_{\rm ex}\|}{\|q_{\rm ex}\|},
\\[3pt]
\eta_h
&=\Delta_h^{(p)}\|L_{\rm ex}^{-1}\|,
&
\delta_{\rm sp}^{(p)}
&=\|\Phi_p(h)^\top\mathbb J_{2n}\Phi_p(h)-\mathbb J_{2n}\|.
\end{aligned}
\end{equation}
These quantities measure the one-step error, the relative
history-vector error, the size of the shooting perturbation,
and the defect in the symplectic identity, respectively.
Whenever \(\eta_h<1\), the perturbation estimate in
\eqref{eq:appendix-q-qhat-first-bound} gives
\begin{equation}\label{eq:numerical-relative-error-bound}
E_q\leq\frac{\eta_h}{1-\eta_h}.
\end{equation}

We first examine convergence as the time step decreases.
For fixed system matrices, \(T\), and \(p\),
Lemma~\ref{lem:rk-global-error-v2} gives
\(\Delta_h^{(p)}=\mathcal O(h^{2p+1})\), while
Theorem~\ref{thm:inverse-bound-L} gives
\(\|L_{\rm ex}^{-1}\|=\mathcal O(h^{-1})\).
Thus, \eqref{eq:numerical-relative-error-bound} yields
\(E_q=\mathcal O(h^{2p})\) as \(h\to0\).
We test these orders on the representative mass--spring
instance with \(p=1,2,3\) and target normalized steps
\[
\vartheta_{\rm tar}\in
\{0.04,0.05,0.06,0.07,0.075,0.1,0.2,0.3,0.5,0.8\}.
\]
For the one-step error, we set
\(h=\vartheta_{\rm tar}/\|H\|\), they all satisfy \(h\|H\|<1\), as required by Lemma~\ref{lem:rk-global-error-v2}. For the history vector, we compute \(E_q\) using the grid specified by \eqref{eq:numerical-grid-choice} and plot it against the resulting normalized step \(h\|H\|\).

For \(p=1,2,3\), the fitted slopes for the one-step error
in Fig.~\ref{fig:gauss_convergence}(a) are
\(2.996\), \(4.997\), and \(6.992\), respectively.
The corresponding slopes for \(E_q\) in panel~(b) are
\(1.997\), \(3.996\), and \(6.019\).
These slopes agree with the predicted orders \(2p+1\)
and \(2p\).
We restrict the fits to \(h\|H\|\leq0.5\) and exclude
points with errors below \(5\times10^{-14}\) to reduce
the influence of roundoff.

We then keep the time grid fixed and increase the number
of Gauss stages.
For the same mass--spring instance, we take \(T=10\)
and \(m=1032\), so that
\[
\vartheta:=h\|H\|\simeq0.07499<\frac{1}{4\pi},
\]
and consider \(p=1,\ldots,8\).
At this fixed step size,
Lemma~\ref{lem:rk-global-error-v2} gives
\[
\Delta_h^{(p)}
\leq B_p(\vartheta)
:=
\frac{8(p!)^2}{(2p)!(2p+1)!}\,
\vartheta^{2p+1}.
\]
At this fixed time step, increasing \(p\) from \(1\) to \(8\) reduces the one-step error from \(5.72\times10^{-6}\) to \(1.26\times10^{-46}\), as shown in Fig.~\ref{fig:gauss_convergence}(c). Over the same range, the relative history-vector error in panel~(d) decreases from \(3.23\times10^{-5}\) to \(3.09\times10^{-47}\). For every tested \(p\), the one-step error lies below \(B_p(\vartheta)\), and the history-vector error lies below the perturbation bound \eqref{eq:numerical-relative-error-bound}.

\begin{figure*}[t]
\centering
\includegraphics[width=0.46\textwidth]{%
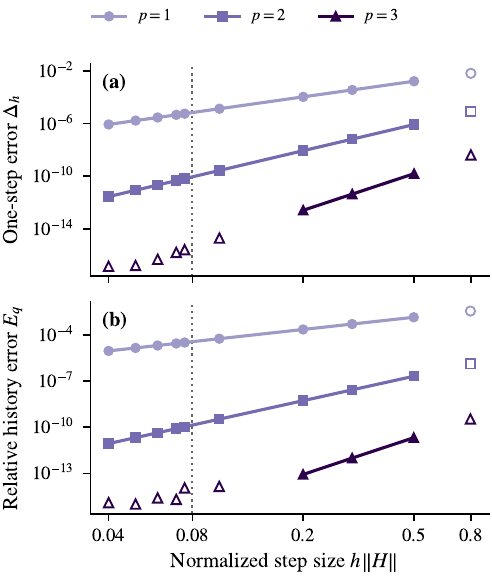}
\hfill
\includegraphics[width=0.46\textwidth]{%
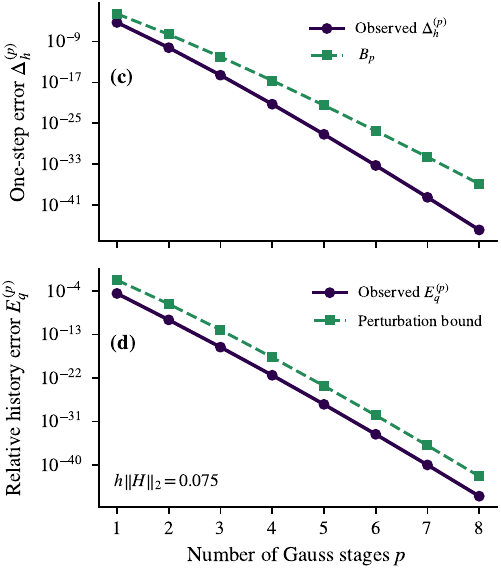}
\caption{Convergence of the Gauss--Legendre discretization
for the representative mass--spring instance.
Panels~(a) and~(b) show the one-step propagator error
\(\Delta_h^{(p)}\) and the relative history-vector error
\(E_q\), respectively, as functions of the normalized step
\(h\|H\|\) for \(p=1,2,3\).
Filled markers indicate the points used in the slope fits,
which satisfy \(h\|H\|\leq0.5\) and have errors at least
\(5\times10^{-14}\); open markers indicate the excluded points.
The vertical dotted line marks the sufficient block-encoding
threshold \(h\|H\|=1/(4\pi)\).
Panels~(c) and~(d) show the corresponding errors as the number
of Gauss stages increases on a fixed grid with \(m=1032\)
and \(h\|H\|\simeq0.07499\), together with the one-step bound
\(B_p(\vartheta)\) and the perturbation bound
\eqref{eq:numerical-relative-error-bound}.
The history-vector errors in panels~(b) and~(d) are
computed at \(T=10\).}
\label{fig:gauss_convergence}
\end{figure*}


\paragraph{Accuracy across the benchmarks.}
Table~\ref{tab:symplectic_accuracy_across_benchmarks} reports the results for the three representative instances with \(p=2\) and \(\vartheta_{\rm tar}=0.05\). In all three cases, \(\eta_h<5.8\times10^{-8}\), so the perturbation condition \(\eta_h\leq1/2\) is satisfied. The computed condition numbers of \(L_{\rm ex}\) and \(L\) agree to the precision shown in the table. The measured symplecticity defects are below \(7.8\times10^{-16}\), consistent with the symplecticity of the Gauss map in exact arithmetic.

\begin{table}[t]
\centering
\small
\setlength{\tabcolsep}{4pt}
\renewcommand{\arraystretch}{1.12}
\begin{tabular}{@{}lrcccccc@{}}
\toprule
Benchmark & \(m\) & \(\Delta_h^{(2)}\) & \(\delta_{\rm sp}\) & \(E_q\)
& \(\eta_h\) & \(\kappa(L_{\rm ex})\) & \(\kappa(L)\)\\
\midrule
Mass--spring & 1548
& \(8.53\times10^{-12}\) & \(3.33\times10^{-16}\) & \(2.10\times10^{-11}\)
& \(4.20\times10^{-8}\) & \(2.46\times10^4\) & \(2.46\times10^4\)\\
Vehicle platoon & 342
& \(3.54\times10^{-11}\) & \(3.33\times10^{-16}\) & \(6.22\times10^{-10}\)
& \(1.00\times10^{-8}\) & \(5.77\times10^2\) & \(5.77\times10^2\)\\
Reaction--diffusion & 704
& \(1.04\times10^{-10}\) & \(7.77\times10^{-16}\) & \(1.21\times10^{-10}\)
& \(5.72\times10^{-8}\) & \(1.12\times10^3\) & \(1.12\times10^3\)\\
\bottomrule
\end{tabular}
\caption{Accuracy and conditioning of the two-stage Gauss
discretization for the three representative instances at
\(T=10\). Each grid is chosen by
\eqref{eq:numerical-grid-choice} with
\(\vartheta_{\rm tar}=0.05\). }
\label{tab:symplectic_accuracy_across_benchmarks}
\end{table}


\subsection{Temporal-resolution and horizon dependence}
\label{subsec:temporal-resolution-horizon}

We study \(\kappa(L_{\rm ex})\) for the representative
mass--spring instance through the two experiments in
Fig.~\ref{fig:conditioning_mechanism}.
In the first, we refine the time grid while keeping \(T\) fixed;
in the second, we increase \(T\) while keeping \(h\) fixed.

We first fix \(T=10\) and vary \(m\), with \(h=T/(m-1)\).
Figure~\ref{fig:conditioning_mechanism}(a) shows that
\(\kappa(L_{\rm ex})\) initially decreases and then increases
as the grid is refined. The smallest observed condition number
occurs at \(m=12\), where \(h\|H\|\simeq7.03\).
This is well above the sufficient block-encoding threshold
\(1/(4\pi)\); the corresponding block-encoding step condition \(h\|H\|<1/(4\pi)\) is first
satisfied at \(m=973\).
On finer grids, the measured condition number grows
approximately in proportion to \(h^{-1}\), consistent with
\eqref{eq:condition-number-order}. The purple curve shows the numerically computed condition
number \(\kappa(L_{\rm ex})\).
The green band is obtained by multiplying the lower and upper
bounds on \(\|L_{\rm ex}\|\) in
Theorem~\ref{thm:norm-bounds-L} by the numerically computed
\(\|L_{\rm ex}^{-1}\|\).
Thus, comparing the curve with this band checks the bounds
on \(\|L_{\rm ex}\|\). It does not independently test the
theoretical bound on \(\|L_{\rm ex}^{-1}\|\) used in
Theorem~\ref{thm:main-conditioning}.


We next vary the horizon while keeping \(h\|H\|=0.075\) fixed. For each target horizon \(T_{\rm tar}\), we set
\[
m=1+\operatorname{round}(T_{\rm tar}/h),
\qquad T=(m-1)h.
\]
This keeps \(h\) fixed while allowing \(T\) to differ slightly from \(T_{\rm tar}\). The target horizons range from \(1\) to \(60\), and we plot the actual values of \(T\). 

Theorem~\ref{thm:main-conditioning} gives
\[
U_h=1+\max\{\|e^{hH}\|,\sqrt{1+\|G\|^2}\},
\qquad 
\kappa(L_{\rm ex})
\leq MU_h\frac{1-e^{-\alpha_c(T+h)}}{1-e^{-\alpha_ch}}
\leq\frac{MU_h}{1-e^{-\alpha_ch}}, 
\]
where the final bound is independent of \(T\). Figure~\ref{fig:conditioning_mechanism}(b) shows that \(\kappa(L_{\rm ex})\) grows more slowly as the horizon increases, changing by about \(3.8\%\) between \(T\simeq40\) and \(T\simeq60\). This apparent leveling off is consistent with Theorem~\ref{thm:main-conditioning}, which guarantees uniform boundedness in \(T\).

\begin{figure}[!h]
\centering
\includegraphics[width=0.96\linewidth]{%
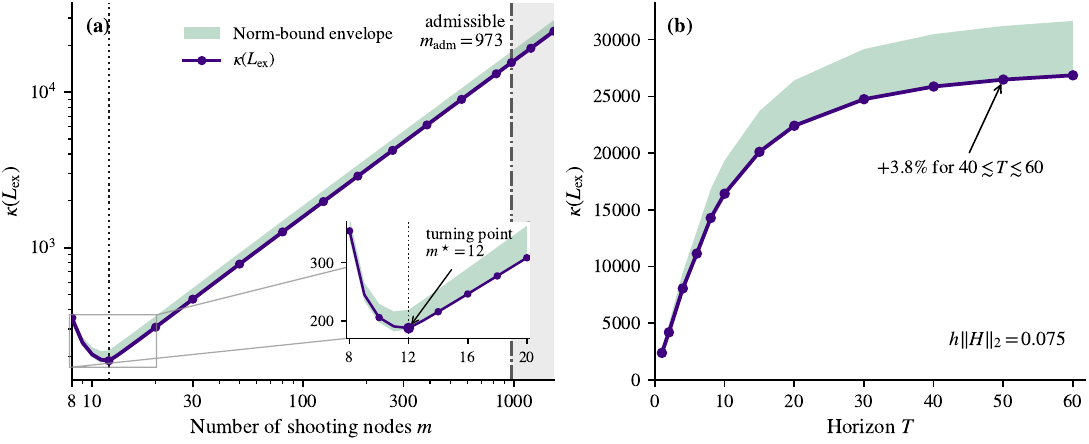}
\caption{Conditioning of the exact mass--spring shooting system.
Panel~(a): temporal refinement at \(T=10\), with the sampled
minimum and the first node count satisfying
\(h\|H\|<1/(4\pi)\). Panel~(b): horizon variation at the fixed
normalized step \(h\|H\|=0.075\). The shaded envelopes are
\([\ell_h\|L_{\rm ex}^{-1}\|,U_h\|L_{\rm ex}^{-1}\|]\), using
the numerically computed inverse norm; they display the norm
bounds of Theorem~\ref{thm:norm-bounds-L}.}
\label{fig:conditioning_mechanism}
\end{figure}

Together, the two experiments support the conditioning analysis: the observed dependence on \(T\) is consistent with a uniform bound at fixed \(h\), whereas reducing \(h\) at fixed \(T\) can worsen conditioning. These results motivate improving the discretization accuracy by increasing the number of Gauss stages \(p\) while keeping the time grid fixed.

\subsection{State-dimension and actuation dependence}
\label{subsec:state_dimension_dependence}

Theorem~\ref{thm:main-conditioning} gives a bound that is uniform in the horizon for a fixed system. As \(n\) increases, the stability parameters, coefficient norms, and selected time step may also change. We examine how \(\kappa(L_{\rm ex})\) depends on \(n\) in the three benchmark families.

We fix \(T=10\) and choose the time grid using \eqref{eq:numerical-grid-choice} with \(\vartheta_{\rm tar}=0.075\). We increase the state dimension \(n\) in each of the three \gls{lqr} families. Within each family, we keep the physical parameters fixed and define the cost weights in the same way at every dimension. Figure~\ref{fig:state_dimension_dependence} shows the results.

For the mass--spring chain, increasing \(n\) from \(16\) to \(100\) changes \(m\) from \(1052\) to \(1080\) and \(\kappa(L_{\rm ex})\) from \(1.8843\times10^4\) to \(2.0841\times10^4\). For \(n\geq64\), the condition number varies by only about \(0.34\%\). This weak dependence on \(n\) is consistent with Theorem~\ref{thm:Mc-alphac-Mp-bound}, which bounds the stability parameters independently of the chain length. The coefficient norms and matrix sparsity also remain uniformly bounded, with \(d=2\) endpoint controls throughout.

For the vehicle platoon, reordering the rows and columns of \(L_{\rm ex}\) gives a direct sum of identical shooting matrices, one for each independently controlled follower. The Hamiltonian norm and selected time grid are unchanged as \(n_f\) increases, so adding followers only repeats the singular values. The condition number therefore remains \(3.9067\times10^2\) for the tested dimensions \(18\leq n\leq99\). The cost of estimating the full control vector still depends on \(n_f\), since \(d=n_f\). The initial errors are sampled separately for each follower, so \(\|x^0\|\) can also grow with \(n_f\). This norm enters the output-estimation bounds but does not affect \(\kappa(L_{\rm ex})\).

For reaction--diffusion, the discrete Laplacian satisfies \(\|D_2\|=\Theta(h_s^{-2})\). As the spatial mesh is refined, \(\|H\|\) increases and the grid rule selects smaller time steps: \(m\) rises from \(466\) at \(n=16\) to \(1314\) at \(n=100\). The condition number initially decreases, reaching its smallest observed value of \(756.68\) at \(n=25\), before rising to \(1714.53\) at \(n=100\). Both the spatial mesh and the time step change in this experiment, so the observed increase in the condition number does not by itself show that the stability constant \(M\) grows with \(n\).

\begin{figure}[!h]
\centering
\includegraphics[width=0.70\linewidth]{%
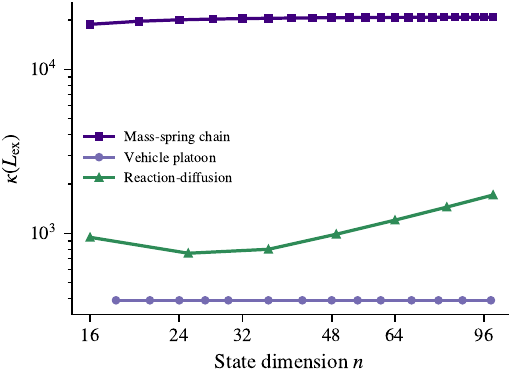}
\caption{State-dimension dependence of the exact shooting-system
condition number for the three benchmark families. All use
\(T=10\), with the time grid chosen separately for each
instance to satisfy \(h\|H\|\leq0.075\). Increasing \(n\)
extends the coupled chain, adds independent followers, or
refines the reaction--diffusion mesh, respectively.}
\label{fig:state_dimension_dependence}
\end{figure}

\paragraph{Actuation structure.}
We finally compare five actuation configurations for a chain
with \(n_m=16\). All five matrices satisfy \(\|\mathcal F\|=1\).
We use the same time grid, with \(T=10\) and \(m=1073\), giving
\(h\|H(\mathcal F)\|\leq0.07498\) in every case.
The three configurations with two control channels are
\[
\mathcal F_{\rm end}=[e_1,e_{16}],\qquad
\mathcal F_{\rm q}=[e_4,e_{13}],\qquad
\mathcal F_{\rm aq}=[\tfrac13e_4,e_{13}].
\]
The other two configurations apply controls to the masses
indexed by \(\mathcal I=\{1,3,\ldots,15\}\):
\[
\mathcal F_{\rm dist}=[e_i]_{i\in\mathcal I},\qquad
\mathcal F_{\rm het}=[g_j e_{2j-1}]_{j=1}^8,
\]
where
\[
(g_1,\ldots,g_8)=
\left(\frac13,\frac56,\frac12,1,
\frac25,\frac{11}{15},\frac35,\frac9{10}\right).
\] 
In each case, \(B=\begin{bmatrix}\mathcal F\\0\end{bmatrix}\) and \(R=\mu I_d\). All remaining parameters are those of the representative instance.

The condition numbers in Table~\ref{tab:mass_spring_actuation} are of the same order of magnitude for all five actuation configurations. They range from \(1.3867\times10^4\) to \(2.0420\times10^4\). The bounds on the stability parameters in Theorem~\ref{thm:Mc-alphac-Mp-bound} depend on \(\mathcal F\) only through \(\|\mathcal F\|\). With the remaining parameters fixed, \(\|H\|\) also admits an upper bound that depends on \(\mathcal F\) only through \(\|\mathcal F\|^2\). Since all five configurations satisfy \(\|\mathcal F\|=1\) and use the same time grid, these estimates give a common a priori upper bound on \(\kappa(L_{\rm ex})\). The measured differences arise from the dependence of the Hamiltonian on the full  matrix \(\mathcal F\mathcal F^\top\), whose structure is not determined by \(\|\mathcal F\|\).


\begin{table}[t]
\centering
\small
\setlength{\tabcolsep}{5pt}
\renewcommand{\arraystretch}{1.12}
\begin{tabular}{@{}lrcc@{}}
\toprule
Configuration & \(d\) & \(\sigma_{\min}(L_{\rm ex})\)
& \(\kappa(L_{\rm ex})\)\\
\midrule
Endpoints & 2 & \(2.5306\times10^{-4}\) & \(2.0420\times10^4\)\\
Quarter points & 2 & \(2.5628\times10^{-4}\) & \(2.0163\times10^4\)\\
Asymmetric quarter points & 2
& \(2.5594\times10^{-4}\) & \(2.0190\times10^4\)\\
Every second mass, uniform & 8
& \(3.7265\times10^{-4}\) & \(1.3867\times10^4\)\\
Every second mass, heterogeneous & 8
& \(3.6942\times10^{-4}\) & \(1.3988\times10^4\)\\
\bottomrule
\end{tabular}
\caption{Exact shooting-system conditioning for five actuation
configurations with \(\|\mathcal F\|=1\), at \(n_m=16\),
\(T=10\), and \(m=1073\). The control dimension \(d\) is
listed because equal spectral norms do not imply equal
numbers of control channels.}
\label{tab:mass_spring_actuation}
\end{table}


\section{Conclusion}
\label{sec:conclusion}
In this work, we develop quantum algorithms that return classical estimates of the optimal cost or the optimal control for continuous-time finite-horizon \gls{lqr}. We show that these algorithms have the potential to reduce the computational cost of optimal control for large systems and long time horizons. Our approach combines multiple shooting with a Gauss--Legendre Runge--Kutta method to formulate the state--costate boundary-value problem as a linear system. Although the underlying Hamiltonian state–costate dynamics contain both growing and decaying modes, we use a Riccati–Lyapunov decomposition to separate forward-stable and backward-stable propagation. Under stabilizability and detectability assumptions, this allows us to establish a condition-number bound for the multiple-shooting linear system that is independent of the horizon \(T\). The bound involves a problem-dependent stability coefficient \(M\), whose scaling with the state dimension \(n\) depends on the \gls{lqr} family. Under the stated input-access assumptions and at fixed target accuracy, the query and additional gate complexities are polylogarithmic in \(n\) and \(T\), provided that \(M\) and the remaining parameters entering the complexity bounds are bounded or grow at most polylogarithmically. This offers the prospect of quantum speedups over conventional classical solvers for suitable problem families. While these conditions are not guaranteed for general \gls{lqr} problems, we provide an explicit expression for \(M\) and establish the required uniform bounds for a concrete family of damped mass–spring chains. Numerical experiments further support our theoretical analysis. Since \gls{lqr} can arise as a subproblem in more general optimal control problems, our results may also contribute to the development of quantum algorithms for these broader problems.

\section{Acknowledgements}
NL acknowledges funding from the Science and Technology Commission of Shanghai Municipality (STCSM) grant no. 24LZ1401200 (21JC1402900) and the Shanghai Pilot Program for Basic Research. NL is also supported by NSFC grants No.12471411, No. W2632089 and No. 12341104, the Shanghai Jiao Tong University 2030 Initiative, the Shanghai Science and Technology Innovation Action Plan (24LZ1401200) and the Fundamental Research Funds for the Central Universities.
CAY and JGC acknowledge funding from IDAE-Nuevos modelos de negocio para la transición energética project number PR-NMN-01-2023-000373.
JGC acknowledges support from the Basque Government BasQ initiative under the
Q-STREAM project. They also acknowledge support from OpenSuperQ+100 (Grant No.~101113946) of the EU
Flagship on Quantum Technologies, from Project Grant No.~PID2024-156808NB-I00, from the Spanish Ministry for Digital Transformation
and of Civil Service of the Spanish Government through the QUANTUM ENIA project call Quantum Spain, and
by the EU through the Recovery, Transformation and Resilience Plan--Next Generation EU within the
framework of the Digital Spain 2026 Agenda, and from the Elkartek project KUBIKIT - kuantikaren
berrikuntzarako ibilbide teknologikoak (ELKARTEK25/79).

\section{Author contributions}
JGC and NL proposed the original idea of applying quantum computing to the \gls{lqr} problem and supervised the project. XH led the project, developed the theoretical framework and worked out the analysis results. CAY conducted the numerical experiments. NG assisted XH and CAY with the theoretical analysis and numerical experiments. All authors discussed the results and contributed to writing and revising the manuscript.

In this work, quantum algorithms, their complexity analysis, and the numerical experiments were designed by humans. Large language models (LLMs) were used to inspire some steps in the proofs, such as the definition and usage of the function $v_j$ in Lemma~\ref{lem:uniform-bou-pr}, the construction of the function $g_a$ in Lemma~\ref{lem:h_a-polynomial-approx}, the application of the diagonal Pad\'e remainder formula in Lemma~\ref{lem:rk-global-error-v2}, and the idea of studying the monotonicity of matrix functions on the set \(\mathcal{D}_W\) in Lemma~\ref{lem:chain-terminal-correction}. LLMs were also used to assist in the development of the code for the numerical experiments and to polish the language of the manuscript.

\bibliographystyle{unsrt}
\bibliography{references}

\appendix

\section{Notation and quantum algorithm primitives}
\label{app:quantum-algorithm-primitives}

We collect the notation and quantum primitives used in the algorithm
and its complexity analysis.

\subsection{Notation}
\label{app:notation}

For a positive integer \(m\), write
\([m]:=\{0,\ldots,m-1\}\) and \(l_m:=\lceil\log_2m\rceil\).
Unless otherwise stated, vector and matrix indices start from zero.
For \(j\in[m]\), let \(\ket{j}^{(m)}\in\mathbb C^m\) be the
\(j\)th standard basis vector and let
\(\bra{j}^{(m)}=(\ket{j}^{(m)})^\dagger\).
When \(m=2^r\), we also write \(\ket{j}_r:=\ket{j}^{(2^r)}\).
Thus a superscript specifies the dimension, whereas a subscript
specifies the number of qubits. We omit these indices when they are
clear from context, and write \(\ket{0}_a=\ket{0}^{\otimes a}\).

For a nonzero vector \(v\in\mathbb C^m\), \(\ket{v}\) denotes
the normalized state obtained by zero-padding \(v\) to dimension
\(2^{l_m}\) and dividing by \(\|v\|\).
We use \(\|v\|_1=\sum_{j\in[m]}|v_j|\).
For a matrix \(A\), \(A^\dagger\) and \(A^+\) denote its
conjugate transpose and Moore--Penrose inverse, respectively.
We write \(\mathbb I_m\) for the identity matrix and
\(\sigma(A)\) for the spectrum of a square matrix.
For a real symmetric matrix \(S\succeq0\), \(S^{1/2}\) denotes
its positive semidefinite square root; \(S^{-1/2}\) is defined
when \(S\succ0\).

We use standard asymptotic notation. The notation
\(\widetilde{\mathcal O}\) suppresses logarithmic factors in the
relevant parameters, while \(\operatorname{poly}\) and
\(\operatorname{polylog}\) denote polynomial and polylogarithmic
factors, respectively.

\subsection{Block-encoding}
\label{sec:block-encoding}

We use the following definition, including zero padding for
rectangular matrices \cite{Gilyen2019}.

\begin{definition}[Block-encoding]
\label{def:block-encoding}
Let \(A\in\mathbb C^{n\times m}\),
\(r=\lceil\log_2\max\{n,m\}\rceil\), and \(N=2^r\).
Define
\begin{equation}\label{eq:def-A-pad}
A_{\rm pad}=
\begin{bmatrix}A&0\\0&0\end{bmatrix}
\in\mathbb C^{N\times N}.
\end{equation}
A unitary \(U_A\) acting on \(a+r\) qubits is an
\((\alpha,a,\varepsilon)\)-block-encoding of \(A\) if
\[
\left\|A_{\rm pad}
-\alpha(\bra{0}_a\otimes\mathbb I_N)
U_A(\ket{0}_a\otimes\mathbb I_N)\right\|
\leq\varepsilon,
\]
where \(\alpha>0\), \(a\geq0\), and \(\varepsilon\geq0\).
When \(\varepsilon=0\), we call \(U_A\) an exact
\((\alpha,a)\)-block-encoding.
We also allow a larger power-of-two padding dimension when a
common system register is needed.
\end{definition}

Here \(a\) counts the block-encoding ancillas, not the system
qubits. If system qubits are subsequently initialized and projected
onto \(\ket{0}\) to obtain a smaller encoded matrix, they are
included in the new ancilla count.
Queries below include controlled calls and calls to inverse
oracles. Additional gate counts exclude the implementation of
the queried oracles.

\subsection{Generic block-encoding}
\label{app:generic-be}

The sparse-access construction gives a finite-precision
implementation of the matrix access used in the main text.

\begin{lemma}[Block-encoding of sparse-access matrices
\cite{Gilyen2019,Morales2026QLSSurvey}]
\label{lem:be-sparse}
Let \(A\in\mathbb C^{n\times m}\) be nonzero and \(s\)-sparse,
and set \(r=\lceil\log_2\max\{n,m\}\rceil\).
Suppose that sparse-access oracles provide the positions of
the nonzero entries in each row and column,
\[
O_r:\ket{i}\ket{k}\mapsto\ket{i}\ket{\rho_{ik}},
\qquad
O_c:\ket{l}\ket{j}\mapsto\ket{\gamma_{lj}}\ket{j},
\]
and a value oracle provides
\[
O_A:\ket{i}\ket{j}\ket{0}_b
\mapsto\ket{i}\ket{j}\ket{\widetilde A_{ij}},
\]
where \(\widetilde A_{ij}\) is a \(b\)-bit approximation of
\(A_{ij}/\|A\|_{\max}\).
The position registers have \(r+1\) qubits. Missing entries are
assigned distinct dummy indices with zero values, and the maps
are extended to unitaries.
With sufficient value precision, one can construct an
\((s\|A\|_{\max},r+3,\varepsilon)\)-block-encoding of \(A\)
using one call to each position oracle, two calls to the value
oracle or its inverse, and
\[
\mathcal O\!\left(
r+\log^{2.5}\!\left(2+\frac{s^2\|A\|_{\max}}{\varepsilon}\right)
\right)
\]
additional one-qubit and CNOT gates.
\end{lemma}

For the small dense auxiliary matrices, we use the following
construction.

\begin{lemma}[Block-encoding of dense matrices
\cite{CampsVanBeeumen2022FABLE}]
\label{lem:be-dense}
Let \(A\in\mathbb C^{n\times m}\) be nonzero, and set
\(r=\lceil\log_2\max\{n,m\}\rceil\) and \(N=2^r\).
There is an exact \((N\|A\|_{\max},r+1)\)-block-encoding of
\(A\) using \(\mathcal O(N^2)\) one-qubit and CNOT gates.
\end{lemma}

This gate count uses arbitrary one-qubit rotations. In the
algorithm, the construction is applied to the small Runge--Kutta
coefficient matrix.

\subsection{Basic operations on block-encodings}
\label{app:be-basic-operation}

We use the following rules for sums, products, and Kronecker
products. For sums and products, the input encodings are taken
on a common system register with compatible zero-padded embeddings.
Their ancilla counts refer to these embeddings.

\begin{lemma}[Basic operations]
\label{lem:manipulation-of-matrices}
Suppose that \(U_A\) and \(U_B\) are
\((\alpha_A,a_A,\varepsilon_A)\)- and
\((\alpha_B,a_B,\varepsilon_B)\)-block-encodings of \(A\) and
\(B\), respectively, with
\(\|A\|\leq\alpha_A\) and \(\|B\|\leq\alpha_B\).
Then the following constructions are available.
\begin{enumerate}
\item If \(A\) and \(B\) have the same dimensions, their sum
has a block-encoding with parameters
\[
\bigl(\alpha_A+\alpha_B,\,
\max\{a_A,a_B\}+1,\,
\varepsilon_A+\varepsilon_B\bigr).
\]
It uses one controlled call to each input encoding and two
additional one-qubit rotations.

\item If \(A\in\mathbb C^{n\times m}\) and
\(B\in\mathbb C^{m\times k}\), their product, on the common
system register, has a block-encoding with parameters
\[
\bigl(\alpha_A\alpha_B,\,a_A+a_B,\,
\alpha_A\varepsilon_B+\alpha_B\varepsilon_A\bigr),
\]
using one call to each input encoding.

\item If \(A\in\mathbb C^{n\times m}\) and
\(B\in\mathbb C^{l\times k}\), set
\[
\begin{aligned}
r_A&=\lceil\log_2\max\{n,m\}\rceil,\qquad
r_B=\lceil\log_2\max\{l,k\}\rceil,\\
r_{A\otimes B}&=\lceil\log_2\max\{nl,mk\}\rceil,\qquad
\Delta=r_A+r_B-r_{A\otimes B}.
\end{aligned}
\]
For input encodings on these minimal system registers,
\(A\otimes B\) has a block-encoding with parameters
\[
\bigl(\alpha_A\alpha_B,\,a_A+a_B+\Delta,\,
\alpha_A\varepsilon_B+\alpha_B\varepsilon_A
+\varepsilon_A\varepsilon_B\bigr).
\]
It uses one call to each input encoding and
\[
\mathcal O\!\left(
\log^2(2n)\log^2(2l)+\log^2(2m)\log^2(2k)
\right)
\]
additional one-qubit and CNOT gates.
\end{enumerate}
\end{lemma}

The sum and product rules follow from \cite{Gilyen2019}.
The Kronecker-product construction combines the tensor product
of the encodings with row and column permutations; see
\cite{DongLiXue2025ProductsBlockEncodings,CampsVanBeeumen2021CircuitSynthesis}.
Any system qubits projected out after a product must also be
included in its final ancilla count.
For later use, define
\begin{equation}\label{eq:def-ancilla-p,n,m}
\tau_{p,n}=l_p+l_n-\lceil\log_2(pn)\rceil,
\qquad
\tau_{m,n}=l_m+l_n-\lceil\log_2(mn)\rceil.
\end{equation}
Both quantities belong to \(\{0,1\}\).

\subsection{QSVT-based inverse block-encoding}
\label{app:qsvt-inversion-details}

The next lemma includes the effect of error in the input
block-encoding. Zero padding is allowed, so the proof uses a
cutoff approximation of the reciprocal function.

\begin{lemma}[Inverse block-encoding]
\label{lem:inverse-be-approx}
Let \(A\in\mathbb C^{n\times n}\) be nonsingular, and suppose
that \(U_A\) is an \((\alpha,a,\delta)\)-block-encoding of \(A\).
Let \(\kappa_0\geq\max\{2,\alpha\|A^{-1}\|\}\).
For a target accuracy \(\varepsilon>0\), suppose that
\[
\delta\leq\min\left\{
\frac{\alpha}{4\kappa_0},
\frac{3\alpha^2\varepsilon}{8(1+\sqrt5)\kappa_0^2}
\right\}.
\]
Then one can construct a
\((4\kappa_0/\alpha,a+1,\varepsilon)\)-block-encoding of
\(A^{-1}\) using
\[
\mathcal O\!\left(
\kappa_0\log\!\left(2+\frac{\kappa_0}{\alpha\varepsilon}\right)
\right)
\]
queries to \(U_A\) and \(U_A^\dagger\), and
\[
\mathcal O\!\left(
(a+1)\kappa_0\log\!\left(2+\frac{\kappa_0}{\alpha\varepsilon}\right)
\right)
\]
additional one-qubit and CNOT gates.
\end{lemma}

\begin{proof}
Let \(N\) be the padded dimension and put
\[
\widetilde A=
\alpha(\bra{0}_a\otimes\mathbb I_N)
U_A(\ket{0}_a\otimes\mathbb I_N).
\]
Then \(\|\widetilde A-A_{\rm pad}\|\leq\delta\) and
\(\|\widetilde A/\alpha\|\leq1\).
The nonzero singular values of \(A_{\rm pad}\) are those of \(A\).
Writing the singular values in decreasing order, Weyl's inequality gives
\begin{equation}\label{eq:A_pad'-singularvalue}
\begin{alignedat}{2}
\sigma_j(\widetilde A/\alpha)
&\geq \sigma_j(A_{\rm pad}/\alpha)
      -\frac{\|\widetilde A-A_{\rm pad}\|}{\alpha}
 \geq \frac{1}{\alpha\|A^{-1}\|}-\frac{\delta}{\alpha}
 \geq \frac{3}{4\kappa_0},
&\quad& 0\leq j<n,
\\[5pt]
\sigma_j(\widetilde A/\alpha)
&\leq \frac{\|\widetilde A-A_{\rm pad}\|}{\alpha}
 \leq \frac{\delta}{\alpha}
 \leq \frac{1}{4\kappa_0},
&\quad& n\leq j<N.
\end{alignedat}
\end{equation}
The second range is empty when \(N=n\).

Let \(\widehat A\) be the rank-\(n\) truncation of
\(\widetilde A\). For clarity, write
$s_j:=\sigma_j(\widetilde A/\alpha)$, with
$s_0\geq s_1\geq\cdots\geq s_{N-1}\geq0$.
Since $\widehat A$ is obtained by retaining the $n$ largest
singular values of $\widetilde A$, the discarded part satisfies,
when $n<N$,
\[
\|\widetilde A-\widehat A\|
=\alpha\max_{n\leq j<N}s_j
=\alpha s_n
\leq\delta,
\]
where the last inequality follows from
\eqref{eq:A_pad'-singularvalue}.
If $n=N$, then $\widehat A=\widetilde A$, so the same bound
holds with a zero left-hand side.
Consequently, the triangle inequality gives
\[
\begin{aligned}
\|\widehat A-A_{\rm pad}\|
\leq
\|\widehat A-\widetilde A\|
+\|\widetilde A-A_{\rm pad}\|\leq\delta+\delta
=2\delta.
\end{aligned}
\]

Moreover, the nonzero singular values of $\widehat A$ are
$\alpha s_0,\ldots,\alpha s_{n-1}$, while those of
$A_{\rm pad}$ coincide with the singular values of $A$.
Recall that the spectral norm of the Moore--Penrose inverse
is the reciprocal of the smallest nonzero singular value.
Thus, denoting the latter by $\sigma_{\min}^{+}$, we obtain
\[
\begin{aligned}
\|\widehat A^+\|
&=\frac{1}{\sigma_{\min}^{+}(\widehat A)}
 =\frac{1}{\alpha s_{n-1}}
 \leq\frac{4\kappa_0}{3\alpha},\\[4pt]
\|A_{\rm pad}^+\|
&=\frac{1}{\sigma_{\min}^{+}(A_{\rm pad})}
 =\frac{1}{\sigma_{\min}(A)}
 =\|A^{-1}\|
 \leq\frac{\kappa_0}{\alpha},
\end{aligned}
\]
where we used \eqref{eq:A_pad'-singularvalue} and
$\kappa_0\geq\alpha\|A^{-1}\|$.
Since \(\widehat A\) and \(A_{\rm pad}\) have the same rank,
Wedin's perturbation bound \cite[Eq.~(1.3)]{XU2020124920} gives
\[
\begin{aligned}
\|\widehat A^+-A_{\rm pad}^+\|
&\leq\frac{1+\sqrt5}{2}
\|\widehat A^+\|\|A_{\rm pad}^+\|
\|\widehat A-A_{\rm pad}\|\\
&\leq\frac{4(1+\sqrt5)\kappa_0^2\delta}{3\alpha^2}
\leq\frac{\varepsilon}{2}.
\end{aligned}
\]

For \(0<\eta\leq1/4\), the bounded polynomial approximation
results in \cite{Gilyen2019} give an odd real polynomial \(P\)
of degree \(\mathcal O(\kappa_0\log\eta^{-1})\) such that
\[
\begin{aligned}
|P(x)|&\leq1, &&x\in[-1,1],\\
\left|P(x)-\frac{1}{4\kappa_0x}\right|
&\leq\eta, &&x\in[1/(2\kappa_0),1],\\
|P(x)|&\leq\eta, &&x\in[0,1/(4\kappa_0)].
\end{aligned}
\]
Write \(\widetilde A/\alpha=W\Sigma V^\dagger\).
Applying real \gls{qsvt} to \(U_A^\dagger\) gives an encoded
block \(B_{\rm out}=VP(\Sigma)W^\dagger\), with one additional
ancilla. To make the approximation error explicit, write the
singular-value decomposition as
\[
\widetilde A/\alpha
=W\operatorname{diag}(s_0,\ldots,s_{N-1})V^\dagger.
\]
By the definition of $\widehat A$, its Moore--Penrose inverse is
\[
\widehat A^+
=\frac{1}{\alpha}\,
V\operatorname{diag}
\left(s_0^{-1},\ldots,s_{n-1}^{-1},0,\ldots,0\right)
W^\dagger,
\]
whereas the encoded output block is
\[
B_{\rm out}
=V\operatorname{diag}
\left(P(s_0),\ldots,P(s_{N-1})\right)W^\dagger.
\]

Since $\widetilde A/\alpha$ is a block of a unitary,
all its singular values are at most $1$.
Together with \eqref{eq:A_pad'-singularvalue}, this gives
\[
\begin{alignedat}{2}
s_j
&\in\left[\frac{3}{4\kappa_0},1\right]
 \subseteq\left[\frac{1}{2\kappa_0},1\right],
&\qquad& 0\leq j<n,\\[4pt]
s_j
&\in\left[0,\frac{1}{4\kappa_0}\right],
&\qquad& n\leq j<N.
\end{alignedat}
\]
The approximation properties of $P$ therefore imply
\[
\begin{alignedat}{2}
\left|P(s_j)-\frac{1}{4\kappa_0s_j}\right|
&\leq\eta,
&\qquad& 0\leq j<n,\\[4pt]
|P(s_j)|
&\leq\eta,
&\qquad& n\leq j<N.
\end{alignedat}
\]
By unitary invariance of the spectral norm, we consequently have
\[
\begin{aligned}
\left\|
\frac{\alpha}{4\kappa_0}\widehat A^+-B_{\rm out}
\right\|
&=
\max\left\{
\max_{0\leq j<n}
\left|\frac{1}{4\kappa_0s_j}-P(s_j)\right|,
\max_{n\leq j<N}|P(s_j)|
\right\}\leq\eta,
\end{aligned}
\]
where the second maximum is omitted if $n=N$.
Multiplying by $4\kappa_0/\alpha$ yields
\[
\begin{aligned}
\left\|
\widehat A^+-\frac{4\kappa_0}{\alpha}B_{\rm out}
\right\|
&=
\frac{4\kappa_0}{\alpha}
\left\|
\frac{\alpha}{4\kappa_0}\widehat A^+-B_{\rm out}
\right\|\leq\frac{4\kappa_0\eta}{\alpha}.
\end{aligned}
\]
Choose \(\eta=\min\{1/4,\alpha\varepsilon/(8\kappa_0)\}\).
It follows that
\[
\left\|A_{\rm pad}^+
-\frac{4\kappa_0}{\alpha}B_{\rm out}\right\|
\leq\varepsilon.
\]
Since \(A_{\rm pad}^+=(A^{-1})_{\rm pad}\), this is the
required inverse block-encoding. The complexity follows from
the degree of \(P\) and the \gls{qsvt} construction.
\end{proof}

\begin{remark}\label{rem:QSVT-inversion}
One may use any available bound
\(\kappa_0=\Theta(\max\{2,\alpha\|A^{-1}\|\})\).
If \(\alpha\geq\|A\|\), then \(\alpha\|A^{-1}\|\geq1\),
and the query bound becomes
\[
\mathcal O\!\left(
\alpha\|A^{-1}\|
\log\!\left(2+\frac{\|A^{-1}\|}{\varepsilon}\right)
\right).
\]
The additional gate count has a further factor \(a+1\).
The inequality \(\alpha\geq\|A\|\) holds automatically for
an exact block-encoding.
\end{remark}

\subsection{Matrix-element estimation}
\label{app:expectation-estimation}

The algorithm reads scalar quantities from block-encoded
operators using a Hadamard test followed by amplitude estimation.

\begin{lemma}[Hadamard test and amplitude estimation]
\label{lem:Hadamard-test}
Let \(O\in\mathbb C^{n\times n}\), and let \(U_O\) be an
\((\alpha,a,\varepsilon')\)-block-encoding of \(O\).
Suppose that \(U_\psi\) and \(U_\varphi\) prepare normalized
states \(\ket{\psi}\) and \(\ket{\varphi}\) in \(\mathbb C^n\),
with zero padding to \(r=\lceil\log_2n\rceil\) qubits.
Controlled versions and inverses of these unitaries are available.
Given \(\varepsilon>0\) and \(0<\delta<1\), if
\(\varepsilon'\leq\varepsilon/2\), one can output an estimate
\(v\) satisfying
\[
\Pr\!\left(
|v-\langle\psi|O|\varphi\rangle|\leq\varepsilon
\right)\geq1-\delta
\]
using
\[
Q_{\rm est}=\mathcal O\!\left(
\left(1+\frac{\alpha}{\varepsilon}\right)
\log\frac{2}{\delta}\right)
\]
queries to these oracles and their inverses, and
\(\mathcal O((a+r+1)Q_{\rm est})\) additional one-qubit and
CNOT gates.
\end{lemma}

\begin{proof}
Suppressing the zero padding, set
\[
z=(\bra{0}_a\otimes\bra{\psi})
U_O(\ket{0}_a\otimes\ket{\varphi}).
\]
The block-encoding error gives
\begin{equation}\label{eq:block-encoded-transition-error}
\left|\alpha z-\langle\psi|O|\varphi\rangle\right|
\leq\varepsilon'.
\end{equation}
The Hadamard-test circuit in
Figure~\ref{fig:transition-hadamard-test} prepares
\[
\frac{1}{\sqrt2}\left(
\ket{0}\ket{0}_a\ket{\psi}
+\ket{1}U_O(\ket{0}_a\ket{\varphi})
\right).
\]
After applying
\(P_\theta=\ket{0}\!\bra{0}+e^{i\theta}\ket{1}\!\bra{1}\)
and a Hadamard gate to the first qubit, the probability of
measuring zero is
\begin{equation}\label{eq:transition-hadamard-output}
p_\theta=\frac{1+\operatorname{Re}(e^{i\theta}z)}{2}.
\end{equation}
Thus \(2p_0-1=\operatorname{Re}z\) and
\(2p_{-\pi/2}-1=\operatorname{Im}z\).

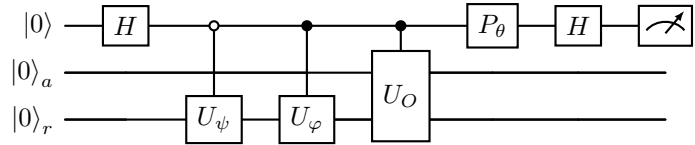
\begin{figure}[!h]
\centering
\begin{quantikz}[
row sep={0.62cm,between origins},
column sep=0.50cm
]
\lstick{\ensuremath{\ket{0}}}
& \gate{H} & \octrl{2} & \ctrl{2} & \ctrl{1}
& \gate{P_\theta} & \gate{H} & \meter{} \\
\lstick{\ensuremath{\ket{0}_a}}
& \qw & \qw & \qw & \gate[2]{U_O}
& \qw & \qw & \qw \\
\lstick{\ensuremath{\ket{0}_r}}
& \qw & \gate{U_\psi} & \gate{U_\varphi} & \qw
& \qw & \qw & \qw
\end{quantikz}
\caption{Hadamard test for a block-encoded matrix element.}
\label{fig:transition-hadamard-test}
\end{figure}

Apply amplitude estimation to the two circuits with
\(\theta=0\) and \(\theta=-\pi/2\)
\cite{Brassard2002}.
Estimate each probability to error at most
\(\min\{1/4,\varepsilon/(8\alpha)\}\), with failure probability
at most \(\delta/2\). The estimates \(\widehat p_0\) and
\(\widehat p_{-\pi/2}\) then give
\[
v=\alpha\left[
(2\widehat p_0-1)+i(2\widehat p_{-\pi/2}-1)
\right].
\]
With probability at least \(1-\delta\), both probability
estimates satisfy their error bounds, and hence
\[
\begin{aligned}
|v-\langle\psi|O|\varphi\rangle|
&\leq\varepsilon'
+2\alpha\bigl(
|\widehat p_0-p_0|
+|\widehat p_{-\pi/2}-p_{-\pi/2}|
\bigr)\\
&\leq\frac{\varepsilon}{2}+\frac{\varepsilon}{2}
=\varepsilon.
\end{aligned}
\]
Amplitude estimation and repetition to reduce the failure
probability give the stated query bound. Each iteration uses
\(\mathcal O(a+r+1)\) additional gates for the required
reflections and control operations.
\end{proof}

When the matrix element is known to be real, only the circuit
with \(\theta=0\) is needed, with the same asymptotic complexity.

\section{Conditioning and perturbation analysis}
\label{app:conditioning-proofs}

We prove the estimates used in Sections~\ref{sec:main-results}
and~\ref{sec:algorithm-proof-overview}. Throughout, \(T>0\),
\(m\geq2\), \(h=T/(m-1)\), and \(t_k=kh\). We write
\(D:=BR^{-1}B^\top\) and use the notation of
Section~\ref{sec:problem-formulation}.

\subsection{Riccati--Lyapunov transformation}
\label{subsec:riccati-lyapunov-decoupling}
\label{app:riccati-lyapunov-proof}

Let \(P(t)\), \(W(t)\), and \(\mathcal S(t)\) be defined by
\eqref{eq:finite-horizon-riccati},
\eqref{eq:lyapunov-W}, and
\eqref{eq:riccati-lyapunov-transform}, respectively.
The assumptions \(Q,G\succeq0\) and \(R\succ0\) ensure that
\(P(t)=P(t)^\top\succeq0\) exists on \([0,T]\)
\cite{anderson1990,liberzon2012}.

\begin{theorem}[Riccati--Lyapunov decoupling]
\label{thm:riccati-lyapunov-decoupling}
Let \(y=(x^\top,\lambda^\top)^\top\) solve
\(\dot y=Hy\), and set
\[
z=\lambda-Px,\qquad v=x+Wz.
\]
Then
\[
\dot v=A_c(t)v,\qquad
\dot z=-A_c(t)^\top z,
\qquad A_c(t)=A-DP(t).
\]
The transformation
\((v^\top,z^\top)^\top=\mathcal S(t)y\) is invertible.
If \(x(0)=x^0\) and \(\lambda(T)=Gx(T)\), then
\(v(0)=x^0\) and \(z(T)=0\).
\end{theorem}

\begin{proof}
Since \(\lambda=Px+z\), the state equation becomes
\(\dot x=A_cx-Dz\). Differentiating \(z\) gives
\[
\begin{aligned}
\dot z
&=-Qx-A^\top(Px+z)-\dot Px-P(A_cx-Dz)\\
&=\bigl(-Q-A^\top P-\dot P-PA+PDP\bigr)x
  +(-A^\top+PD)z.
\end{aligned}
\]
The coefficient of \(x\) vanishes by
\eqref{eq:finite-horizon-riccati}. As \(D=D^\top\) and
\(P=P^\top\), the remaining coefficient is \(-A_c^\top\).
Consequently,
\[
\begin{aligned}
\dot v
&=\dot x+\dot Wz+W\dot z\\
&=A_cv+\bigl(\dot W-A_cW-WA_c^\top-D\bigr)z
 =A_cv,
\end{aligned}
\]
where the last equality follows from \eqref{eq:lyapunov-W}.

The factorization
\[
\mathcal S(t)=
\begin{bmatrix}I_n&W(t)\\0&I_n\end{bmatrix}
\begin{bmatrix}I_n&0\\-P(t)&I_n\end{bmatrix}
\]
shows that \(\mathcal S(t)\) is invertible, with
\[
\mathcal S(t)^{-1}=
\begin{bmatrix}I_n&0\\P(t)&I_n\end{bmatrix}
\begin{bmatrix}I_n&-W(t)\\0&I_n\end{bmatrix}
=
\begin{bmatrix}
I_n&-W(t)\\
P(t)&I_n-P(t)W(t)
\end{bmatrix}.
\]
Finally, \(W(0)=0\) and \(P(T)=G\) imply
\[
v(0)=x(0)=x^0,\qquad
z(T)=\lambda(T)-Gx(T)=0.
\]
\end{proof}

For the LQR boundary conditions, uniqueness of the homogeneous
terminal-value problem gives \(z(t)=0\) on \([0,T]\).
Thus \(v=x\), \(\lambda=Px\), and
\begin{equation}\label{eq:close-loop-IVP}
\dot x=A_c(t)x,\qquad x(0)=x^0.
\end{equation}
The optimal feedback is therefore
\begin{equation}\label{eq:optimal-control-feedback}
u^*(t)=-R^{-1}B^\top P(t)x(t).
\end{equation}
The full transformation will still be needed below, since a bound
on \(L_{\rm ex}^{-1}\) must hold for arbitrary right-hand sides.

\subsection{Horizon-uniform closed-loop and Riccati bounds}
\label{app:horizon-uniform-bounds}

Under Assumption~\ref{ass:sta&det}, the algebraic Riccati equation
\[
A^\top P_++P_+A-P_+DP_++Q=0
\]
has a unique symmetric positive semidefinite solution \(P_+\),
and \(A_+:=A-DP_+\) is Hurwitz
\cite{Kucera1972Contribution,Molinari1977TimeInvariant}.
Let \(W_1\succ0\) and \(W_2\succeq0\) \cite[Theorem 56]{CallierDesoer1991LinearSystemTheory} solve
\begin{equation}\label{eq:Lyapunov-equation-I}
A_+^\top W_1+W_1A_+=-I_n,
\end{equation}
and
\[
A_+W_2+W_2A_+^\top=-D,
\]
respectively. We also set
\[
\widetilde G
:=(G-P_+)\bigl(I_n+W_2(G-P_+)\bigr)^{-1}.
\]
The inverse exists under Assumption~\ref{ass:sta&det} and
\(G\succeq0\); see
\cite[Theorem~2 and Corollary~1]{CallierWinkinWillems1994RDE}.

\subsubsection{Closed-loop stability}
\label{app:finite-horizon-closed-loop-stability-proof}

\begin{proposition}[Horizon-uniform closed-loop stability]
\label{prop:finite-horizon-closed-loop-stability}
Under Assumption~\ref{ass:sta&det}, there are constants
\(M_c\geq1\) and \(\alpha_c>0\), depending only on
\(A,B,Q,R,G\), such that the propagator of
\(\dot v=A_c(t)v\) satisfies
\begin{equation}\label{eq:horizon-uniform-closed-loop-stability}
\|\Psi(t,s)\|\leq M_ce^{-\alpha_c(t-s)},
\qquad 0\leq s\leq t\leq T,
\end{equation}
for every \(T>0\).
\end{proposition}

\begin{proof}
The substitution \(B_R=BR^{-1/2}\) and \(u'=R^{1/2}u\)
reduces the control weight in the cost functional to the identity, so the results of
\cite{CallierWinkinWillems1994RDE} apply with \(B_RB_R^\top=D\).
Since \(A_+\) is Hurwitz, by \cite{CallierDesoer1991LinearSystemTheory} and \cite[p. 993-994]{CallierWinkinWillems1994RDE} there are \(M_1\geq1\) and
\(\alpha_c>0\) such that
\begin{equation}\label{eq:exp-Aplus-bound}
\|e^{A_+t}\|\leq M_1e^{-\alpha_ct},
\qquad t\geq0.
\end{equation}

To keep track of the horizon, let \(P_{T'}\) solve
\begin{equation}\label{eq:DRE-Tprime}
-\dot P_{T'}
=A^\top P_{T'}+P_{T'}A-P_{T'}DP_{T'}+Q,
\qquad P_{T'}(T')=G,
\end{equation}
and let \(\Psi_{T'}\) be the propagator of
\(A_{c,T'}(t):=A-DP_{T'}(t)\).
Taking the supremum over unit initial states in
\cite[Theorem~4]{CallierWinkinWillems1994RDE} gives
\(M_2\geq0\), independent of \(T'\), such that
\begin{equation}\label{eq:Psi-difference-bound}
\|\Psi_{T'}(t,0)-e^{A_+t}\|
\leq M_2e^{-\alpha_cT'}e^{-\alpha_c(T'-t)}
\leq M_2e^{-\alpha_ct},
\qquad 0\leq t\leq T'.
\end{equation}
The last inequality uses \(2T'-t\geq t\). Hence
\begin{equation}\label{eq:Psi-zero-bound}
\|\Psi_{T'}(t,0)\|\leq \|\Psi_{T'}(t,0)-e^{A_+t}\|+\|e^{A_+t}\|
\leq (M_1+M_2)e^{-\alpha_ct}.
\end{equation}

For \(0\leq s<T'\), the function
\(\widehat P(\tau):=P_{T'}(s+\tau)\) solves
\eqref{eq:DRE-Tprime} on \([0,T'-s]\) with terminal value \(G\).
Uniqueness therefore gives
\[
P_{T'-s}(\tau)=P_{T'}(s+\tau),
\qquad 0\leq\tau\leq T'-s.
\]
The corresponding propagators satisfy
\[
\Psi_{T'}(t,s)=\Psi_{T'-s}(t-s,0).
\]
Applying \eqref{eq:Psi-zero-bound} with horizon \(T'-s\) gives
\[
\|\Psi_{T'}(t,s)\| = \|\Psi_{T'-s}(t-s,0)\|
\leq(M_1+M_2)e^{-\alpha_c(t-s)}.
\]
Taking \(T'=T\) and \(M_c=M_1+M_2\) proves the claim.
The case \(s=t=T\) follows from \(\Psi(T,T)=I_n\).
\end{proof}

\subsubsection{Explicit stability constants}
\label{app:ad-M_c-alpha_c}

\begin{remark}\label{re:ad-M_c-alpha_c}
An admissible choice in
Proposition~\ref{prop:finite-horizon-closed-loop-stability} is
\[
\begin{aligned}
M_c
&=\sqrt{\kappa(W_1)}
+\|W_2\|\kappa(W_1)
  \max\{\|G-P_+\|,\|\widetilde G\|\},\\
\alpha_c&=\frac{1}{2\|W_1\|}.
\end{aligned}
\]
These constants are independent of \(T\).
\end{remark}

\begin{proof}
Let \(\dot\xi=A_+\xi\) and \(V=\xi^\top W_1\xi\).
Equation~\eqref{eq:Lyapunov-equation-I} gives
\[
\dot V=-\|\xi\|^2\leq-\frac{V}{\|W_1\|}.
\]
By Gr\"{o}nwall's inequality
\[
V(t)\leq e^{-t/\|W_1\|}V(0).
\]
Using
\(\lambda_{\min}(W_1)\|\xi\|^2\leq V
\leq\lambda_{\max}(W_1)\|\xi\|^2\) and $\|W_1\|=\lambda_{\max}(W_1)$, we obtain
\[
\|\xi(t)\|^2 \leq \frac{\lambda_{\max}(W_1)}{\lambda_{\min}(W_1)} e^{-t/\|W_1\|}\|\xi(0)\|^2
=\kappa(W_1)e^{-t/\|W_1\|}\|\xi(0)\|^2.
\]
Thus \eqref{eq:exp-Aplus-bound} holds with
\[
M_1=\sqrt{\kappa(W_1)},\qquad 
\alpha_c=\frac{1}{2\|W_1\|}.
\]
For these values, the estimate in
\cite[Theorem~4]{CallierWinkinWillems1994RDE} allows
\[
M_2=\|W_2\|M_1^2
\max\{\|G-P_+\|,\|\widetilde G\|\}.
\]
Substituting \(M_c=M_1+M_2\) gives the stated formula.
\end{proof}

\subsubsection{A uniform bound on the Riccati solution}
\label{app::horizon-uniform-riccati-bound}

\begin{proposition}[Horizon-uniform Riccati bound]
\label{prop:horizon-uniform-riccati-bound}
Under Assumption~\ref{ass:sta&det}, the solution of
\eqref{eq:finite-horizon-riccati} satisfies
\[
\sup_{t\in[0,T]}\|P(t)\|\leq M_P,
\]
where one may take
\[
M_P=\|P_+\|+
\kappa(W_1)\max\{\|G-P_+\|,\|\widetilde G\|\}.
\]
In particular, \(M_P\geq0\) is independent of \(T\).
\end{proposition}

\begin{proof}
With the constants chosen in Remark~\ref{re:ad-M_c-alpha_c},
\cite[Theorem~3]{CallierWinkinWillems1994RDE} gives
\begin{equation}\label{eq:error-bound-are-dre}
\|P(t)-P_+\|
\leq
\kappa(W_1)\max\{\|G-P_+\|,\|\widetilde G\|\}
e^{-2\alpha_c(T-t)},
\qquad 0\leq t\leq T.
\end{equation}
Since the exponential factor is at most one,
\[
\|P(t)\|\leq\|P_+\|+\|P(t)-P_+\|\leq M_P.
\]
\end{proof}

\subsection{Decoupling of the multiple-shooting matrix}
\label{subsec:linear-system-decoupling}
\label{app:transformed-one-step-map}

\begin{lemma}[Adjoint propagator]
\label{lem:propagation-of-z}
Let \(\Psi(t,s)\) satisfy
\begin{equation}\label{eq:Psi-propagator-definition}
\partial_t\Psi(t,s)=A_c(t)\Psi(t,s),
\qquad \Psi(s,s)=I_n,
\end{equation}
and let \(\Xi(t,s)\) be the propagator of
\(\dot z=-A_c(t)^\top z\). Then
\[
\Xi(t,s)=\Psi(t,s)^{-\top},
\qquad s,t\in[0,T].
\]
In particular, \(\Xi(s,t)=\Psi(t,s)^\top\).
\end{lemma}

\begin{proof}
Differentiation gives
\[
\frac{\mathrm d}{\mathrm dt}
\bigl(\Psi(t,s)^\top\Xi(t,s)\bigr)
=
\Psi(t,s)^\top A_c(t)^\top\Xi(t,s)
-\Psi(t,s)^\top A_c(t)^\top\Xi(t,s)=0.
\]
The product equals \(I_n\) at \(t=s\), so
\(\Xi(t,s)=\Psi(t,s)^{-\top}\).
The last identity follows from
\(\Psi(s,t)=\Psi(t,s)^{-1}\).
\end{proof}

\begin{lemma}[Transformed one-step propagator]
\label{lem:transformed-one-step-map}
Set \(\mathcal S_k=\mathcal S(t_k)\),
\(\Psi_k=\Psi(t_k,t_{k-1})\), and
\(\Phi_{\rm ex}=e^{hH}\). Then
\begin{equation}\label{eq:transformed-one-step-map}
\mathcal S_k\Phi_{\rm ex}\mathcal S_{k-1}^{-1}
=
\begin{bmatrix}
\Psi_k&0\\
0&\Psi_k^{-\top}
\end{bmatrix},
\qquad k=1,\ldots,m-1.
\end{equation}
\end{lemma}

\begin{proof}
For arbitrary \(w_{k-1}\in\mathbb R^{2n}\), let \(y(t)\)
solve \(\dot y=Hy\) with
\(y(t_{k-1})=\mathcal S_{k-1}^{-1}w_{k-1}\).
Exact propagation gives
\[
\mathcal S_k y(t_k)
=\mathcal S_k\Phi_{\rm ex}\mathcal S_{k-1}^{-1}w_{k-1}.
\]
On the other hand, Theorem~\ref{thm:riccati-lyapunov-decoupling}
and Lemma~\ref{lem:propagation-of-z} give
\[
\mathcal S_k y(t_k)
=\operatorname{diag}(\Psi_k,\Psi_k^{-\top})w_{k-1}.
\]
The two matrices agree on every \(w_{k-1}\), proving the identity.
\end{proof}

We now apply this transformation to
\(L_{\rm ex}\) in \eqref{eq:exact-shooting-operator-block-form}.
For arbitrary nodal vectors, write
\(w_k=(v_k^\top,z_k^\top)^\top=\mathcal S_k y_k\).
The endpoint values \(W(0)=0\) and \(P(T)=G\) give
\begin{equation}\label{eq:transformed-BC}
B_0\mathcal S_0^{-1}w_0
+B_T\mathcal S_{m-1}^{-1}w_{m-1}
=
\begin{bmatrix}v_0\\z_{m-1}\end{bmatrix}.
\end{equation}
For a matching row,
\begin{equation}\label{eq:transformed-matrix-action}
\mathcal S_k(y_k-\Phi_{\rm ex}y_{k-1})
=
w_k-\mathcal S_k\Phi_{\rm ex}
       \mathcal S_{k-1}^{-1}w_{k-1}.
\end{equation}

Define the variable and row transformations by
\begin{equation}\label{eq:defition-right-transformation}
\mathcal S_{\rm grid}
:=\operatorname{diag}(\mathcal S_0,\ldots,\mathcal S_{m-1}),
\end{equation}
and
\begin{equation}\label{eq:defition-left-transformation}
\mathcal R
:=\operatorname{diag}(I_{2n},\mathcal S_1,\ldots,\mathcal S_{m-1}).
\end{equation}
Thus \(\mathcal R L_{\rm ex}\mathcal S_{\rm grid}^{-1}\)
has boundary row \((v_0^\top,z_{m-1}^\top)^\top\) and matching
rows
\[
\begin{bmatrix}
v_k-\Psi_kv_{k-1}\\
z_k-\Psi_k^{-\top}z_{k-1}
\end{bmatrix}.
\]
To put the second component in backward form, set
\[
\mathcal D:=\operatorname{diag}(I_{2n},D_1,\ldots,D_{m-1}),
\qquad
D_k:=\begin{bmatrix}I_n&0\\0&\Psi_k^\top\end{bmatrix}.
\]
The matching rows of
\(\mathcal D\mathcal R L_{\rm ex}\mathcal S_{\rm grid}^{-1}\)
are then
\[
\begin{bmatrix}
v_k-\Psi_kv_{k-1}\\
\Psi_k^\top z_k-z_{k-1}
\end{bmatrix}.
\]

Let $\Pi_c$ be the permutation that maps the grouped variables
to the nodewise ordering, namely,
\[
(v_0,\ldots,v_{m-1},z_0,\ldots,z_{m-1})
\longmapsto
(v_0,z_0,\ldots,v_{m-1},z_{m-1}).
\]
Let $\Pi_r$ group the equations in the order
\[
\left(
v_0,\,
(v_k-\Psi_k v_{k-1})_{k=1}^{m-1},\,
z_{m-1},\,
(\Psi_k^\top z_k-z_{k-1})_{k=1}^{m-1}
\right).
\]
With \(E_{ij} = \ket{i}^{(m)}\bra{j}^{(m)}\) denoting the \(m\times m\) matrix units, this gives
\begin{equation}\label{eq:L-diag-relationship}
\Pi_r\mathcal D\mathcal R L_{\rm ex}
\mathcal S_{\rm grid}^{-1}\Pi_c
=
\begin{bmatrix}L_v&0\\0&L_z\end{bmatrix},
\end{equation}
where
\begin{equation}\label{eq:matrix Lv}
L_v
=\sum_{k=0}^{m-1}E_{kk}\otimes I_n
-\sum_{k=1}^{m-1}E_{k,k-1}\otimes\Psi_k,
\end{equation}
and
\begin{equation}\label{eq:matrix Lz}
L_z
=E_{0,m-1}\otimes I_n
+\sum_{k=1}^{m-1}E_{kk}\otimes\Psi_k^\top
-\sum_{k=1}^{m-1}E_{k,k-1}\otimes I_n.
\end{equation}
Thus \(L_v\) is solved forward from \(v_0\), while \(L_z\)
is solved backward from \(z_{m-1}\).

\subsection{Bounds for the exact linear system}
\label{subsec:condition-number-bound}

For brevity, set
\[
\gamma_h:=\sum_{j=0}^{m-1}e^{-\alpha_cjh}
=\frac{1-e^{-\alpha_chm}}{1-e^{-\alpha_ch}}.
\]

\subsubsection{Inverse bounds for the decoupled matrices}
\label{app:triangular-inverse-bounds-proof}

\begin{lemma}[Inverse bounds for the decoupled matrices]
\label{lem:Lv-Lz-inverse-bounds}
Under Assumption~\ref{ass:sta&det}, the matrices \(L_v\) and
\(L_z\) are nonsingular and
\[
\max\{\|L_v^{-1}\|,\|L_z^{-1}\|\}\leq M_c\gamma_h.
\]
This is the estimate stated in \eqref{eq:Lv-Lz-inverse-bound}.
\end{lemma}

\begin{proof}
Let \(L_vv=f\), with
\(v=(v_0^\top,\ldots,v_{m-1}^\top)^\top\) and
\(f=(f_0^\top,\ldots,f_{m-1}^\top)^\top\).
Equation~\eqref{eq:matrix Lv} gives
\[
v_0=f_0,\qquad v_k=\Psi_kv_{k-1}+f_k.
\]
This recursion has a unique solution. By the composition property
of the propagator,
\[
v_k=\sum_{j=0}^k\Psi(t_k,t_j)f_j.
\]
Proposition~\ref{prop:finite-horizon-closed-loop-stability}
and weighted Cauchy--Schwarz imply
\[
\begin{aligned}
\|v_k\|^2
&\leq M_c^2
\left(\sum_{j=0}^k e^{-\alpha_ch(k-j)}\|f_j\|\right)^2\\
&\leq M_c^2
\left(\sum_{j=0}^k e^{-\alpha_ch(k-j)}\right)
\left(\sum_{j=0}^k e^{-\alpha_ch(k-j)}\|f_j\|^2\right)\\
&\leq M_c^2\gamma_h
\sum_{j=0}^k e^{-\alpha_ch(k-j)}\|f_j\|^2.
\end{aligned}
\]
Summing over \(k\) and reversing the order of summation yields
\[
\begin{aligned}
\|v\|^2
&\leq M_c^2\gamma_h
\sum_{j=0}^{m-1}
\left(\sum_{k=j}^{m-1}e^{-\alpha_ch(k-j)}\right)\|f_j\|^2\\
&\leq M_c^2\gamma_h^2\|f\|^2.
\end{aligned}
\]
Hence
\begin{equation}\label{eq: upper bound for L_v inverse}
\|L_v^{-1}\|\leq M_c\gamma_h.
\end{equation}

For the backward solve, take arbitrary
\(g=(g_0^\top,\ldots,g_{m-1}^\top)^\top\) and write the
right-hand side in the row ordering of \(L_z\) as
\[
\widehat g
=(g_{m-1}^\top,-g_0^\top,\ldots,-g_{m-2}^\top)^\top.
\]
This signed permutation preserves the Euclidean norm.
Then \(L_zz=\widehat g\) is equivalent to
\[
z_{m-1}=g_{m-1},\qquad
z_k=\Psi_{k+1}^\top z_{k+1}+g_k,
\quad k=m-2,\ldots,0.
\]
The solution is unique and satisfies
\[
z_k=\sum_{j=k}^{m-1}\Psi(t_j,t_k)^\top g_j.
\]
As above,
\[
\|z_k\|^2
\leq M_c^2\gamma_h
\sum_{j=k}^{m-1}e^{-\alpha_ch(j-k)}\|g_j\|^2.
\]
Therefore,
\[
\begin{aligned}
\|z\|^2
&\leq M_c^2\gamma_h
\sum_{j=0}^{m-1}
\left(\sum_{k=0}^{j}e^{-\alpha_ch(j-k)}\right)\|g_j\|^2\\
&\leq M_c^2\gamma_h^2\|g\|^2
=M_c^2\gamma_h^2\|\widehat g\|^2.
\end{aligned}
\]
This proves \(\|L_z^{-1}\|\leq M_c\gamma_h\).
\end{proof}

\subsubsection{The inverse of the exact shooting matrix}
\label{app:inverse-bound-L}

\begin{theorem}[Inverse bound for the exact shooting matrix]
\label{thm:inverse-bound-L}
Under Assumption~\ref{ass:sta&det}, \(L_{\rm ex}\) is
nonsingular and
\[
\|L_{\rm ex}^{-1}\|\leq M\gamma_h,
\]
where
\[
M=M_c^2
\left(1+\frac{M_c^2\|B\|^2\|R^{-1}\|}{2\alpha_c}\right)^2
(1+M_P)^2
\]
is the constant defined in Section~\ref{sec:main-results}.
\end{theorem}

\begin{proof}
By \cite[Corollary 1.1.6]{aboukandil2003basic} the solution to the Lyapunov equation in \eqref{eq:lyapunov-W} is
\[
W(t)=\int_0^t\Psi(t,s)D\Psi(t,s)^\top\,\mathrm ds.
\]
Thus
\[
\begin{aligned}
\|W(t)\|
&\leq M_c^2\|B\|^2\|R^{-1}\|
       \int_0^t e^{-2\alpha_c(t-s)}\,\mathrm ds\\
&\leq\frac{M_c^2\|B\|^2\|R^{-1}\|}{2\alpha_c}.
\end{aligned}
\]
The factorization in
Appendix~\ref{app:riccati-lyapunov-proof} therefore gives
\[
\max\{\|\mathcal S(t)\|,\|\mathcal S(t)^{-1}\|\}
\leq(1+\|P(t)\|)(1+\|W(t)\|)
\leq a_S,
\]
where
\[
a_S:=(1+M_P)
\left(1+\frac{M_c^2\|B\|^2\|R^{-1}\|}{2\alpha_c}\right).
\]
By the block-diagonal definitions,
\[
\|\mathcal S_{\rm grid}^{-1}\|\leq a_S,\qquad
\|\mathcal R\|\leq a_S,
\]
and
\[
\|\mathcal D\|
=\max\left\{1,\max_{1\leq k\leq m-1}\|\Psi_k\|\right\}
\leq M_c.
\]

All transformations in \eqref{eq:L-diag-relationship} are
invertible, and Lemma~\ref{lem:Lv-Lz-inverse-bounds} gives
invertibility of both diagonal blocks. Consequently,
\[
L_{\rm ex}^{-1}
=\mathcal S_{\rm grid}^{-1}\Pi_c
\begin{bmatrix}L_v^{-1}&0\\0&L_z^{-1}\end{bmatrix}
\Pi_r\mathcal D\mathcal R.
\]
Since permutation matrices have norm one,
\[
\begin{aligned}
\|L_{\rm ex}^{-1}\|
&\leq
\|\mathcal S_{\rm grid}^{-1}\|
\max\{\|L_v^{-1}\|,\|L_z^{-1}\|\}
\|\mathcal D\|\|\mathcal R\|\\
&\leq a_S^2M_c^2\gamma_h
=M\gamma_h.
\end{aligned}
\]
This proves \eqref{eq:L-inverse-bound-explicit}.
\end{proof}

For every \(h>0\),
\[
\gamma_h
\leq\frac{1}{1-e^{-\alpha_ch}}
=1+\frac{1}{e^{\alpha_ch}-1}
\leq1+\frac{1}{\alpha_ch}.
\]
In particular, the inverse bound is uniform in \(T\) at fixed
\(h\), and is \(\mathcal O(h^{-1})\) as \(h\downarrow0\)
for fixed system matrices.

\subsubsection{The norm and condition number of the exact matrix}
\label{app:norm-bounds-L}

\begin{theorem}[Norm bounds for the exact shooting matrix]
\label{thm:norm-bounds-L}
For \(m\geq2\), the exact shooting matrix satisfies
\begin{equation}\label{eq:L-norm-two-sided-bound}
\begin{aligned}
\|L_{\rm ex}\|
&\geq\max\{\|\Phi_{\rm ex}\|,\sqrt{2+\|G\|^2}\},\\
\|L_{\rm ex}\|
&\leq1+\max\{\|\Phi_{\rm ex}\|,\sqrt{1+\|G\|^2}\}.
\end{aligned}
\end{equation}
\end{theorem}

\begin{proof}
Decompose \(L_{\rm ex}=L_{\rm diag}+L_{\rm off}\), where
\[
\begin{aligned}
L_{\rm diag}
&=E_{00}\otimes B_0
+\sum_{k=1}^{m-1}E_{kk}\otimes I_{2n},\\
L_{\rm off}
&=E_{0,m-1}\otimes B_T
-\sum_{k=1}^{m-1}E_{k,k-1}\otimes\Phi_{\rm ex}.
\end{aligned}
\]
We have \(\|L_{\rm diag}\|=1\). For an arbitrary vector
\(a=(a_0^\top,\ldots,a_{m-1}^\top)^\top\),
\[
\|L_{\rm off}a\|^2
=\|B_Ta_{m-1}\|^2+
\sum_{k=1}^{m-1}\|\Phi_{\rm ex}a_{k-1}\|^2.
\]
Since \(B_T=\begin{bmatrix}0&0\\-G&I_n\end{bmatrix}\),
\[
\|B_T\|^2=\|GG^\top+I_n\|=1+\|G\|^2.
\]
It follows that
\[
\|L_{\rm off}a\|^2
\leq
\max\{\|\Phi_{\rm ex}\|^2,1+\|G\|^2\}\|a\|^2.
\]
The triangle inequality now proves the upper bound.

For the first lower bound, take \(a\) with only its last block
nonzero. The boundary row and the last matching row give
\[
\|L_{\rm ex}a\|^2
=\|B_Ta_{m-1}\|^2+\|a_{m-1}\|^2.
\]
Choosing \(a_{m-1}\) as a unit right singular vector for
\(\|B_T\|\) yields
\(\|L_{\rm ex}\|\geq\sqrt{2+\|G\|^2}\).
Next, take \(a\) with only its first block nonzero.
The first matching row is \(-\Phi_{\rm ex}a_0\), so
\[
\|L_{\rm ex}a\|\geq\|\Phi_{\rm ex}a_0\|.
\]
Taking the supremum over unit \(a_0\) gives
\(\|L_{\rm ex}\|\geq\|\Phi_{\rm ex}\|\).
\end{proof}

Combining Theorems~\ref{thm:inverse-bound-L}
and~\ref{thm:norm-bounds-L}, with
\(\Phi_{\rm ex}=e^{hH}\), proves
\eqref{eq:condition-number-scaling} in
Theorem~\ref{thm:main-conditioning}.
For fixed system matrices,
\begin{equation}\label{eq:condition-number-order}
\kappa(L_{\rm ex})=\mathcal O(h^{-1}),
\qquad h\downarrow0,
\end{equation}
uniformly in \(T\). At a fixed step size \(h>0\), the same
estimate \eqref{eq:condition-number-scaling} is independent
of \(T\).

\subsection{Propagator perturbation and solution error}
\label{app:propagator-perturbation}

Let \(L\) be obtained from \(L_{\rm ex}\) by replacing
\(\Phi_{\rm ex}\) with \(\Phi\), and set
\[
\Delta_h:=\|\Phi-\Phi_{\rm ex}\|,
\qquad E:=L-L_{\rm ex}.
\]
The boundary row is unchanged, so
\[
E=-S_m\otimes(\Phi-\Phi_{\rm ex}),
\qquad
S_m=\sum_{k=1}^{m-1}E_{k,k-1}.
\]
Since \(\|S_m\|=1\) for \(m\geq2\),
\(\|E\|=\Delta_h\), as stated in
\eqref{eq:shooting-matrix-perturbation}. In this subsection,
write \(\theta:=\Delta_h\|L_{\rm ex}^{-1}\|\).

\subsubsection{Conditioning under propagator perturbation}
\label{app:approx-conditioning-proof}

\begin{theorem}[Conditioning under propagator perturbation]
\label{thm:approx-shooting-conditioning}
Suppose that \(L_{\rm ex}\) is nonsingular and \(\theta<1\).
Then \(L\) is nonsingular and
\begin{equation}\label{eq:approx-inverse-bound}
\|L^{-1}\|
\leq\frac{\|L_{\rm ex}^{-1}\|}{1-\theta}.
\end{equation}
Moreover,
\begin{equation}\label{eq:approx-kappa-bound}
\kappa(L)
\leq
\frac{(\|L_{\rm ex}\|+\Delta_h)\|L_{\rm ex}^{-1}\|}
{1-\theta}.
\end{equation}
If \(\theta\leq1/2\), then
\[
\|L^{-1}\|\leq2\|L_{\rm ex}^{-1}\|,
\qquad
\kappa(L)\leq3\kappa(L_{\rm ex}).
\]
\end{theorem}

\begin{proof}
Factor \(L\) as
\begin{equation}\label{eq:tilde-L-another-form}
L=L_{\rm ex}(I+L_{\rm ex}^{-1}E).
\end{equation}
Since \(\|L_{\rm ex}^{-1}E\|\leq\theta<1\), the Neumann series \cite[Chapter 17]{Lax2002FunctionalAnalysis}
\begin{equation}\label{eq:Neumann-series}
(I+L_{\rm ex}^{-1}E)^{-1}
=\sum_{j=0}^{\infty}(-L_{\rm ex}^{-1}E)^j
\end{equation}
converges in operator norm and satisfies
\[
\|(I+L_{\rm ex}^{-1}E)^{-1}\|
\leq\sum_{j=0}^{\infty}\theta^j
=\frac{1}{1-\theta}.
\]
Thus \(L\) is nonsingular, and
\[
L^{-1}=(I+L_{\rm ex}^{-1}E)^{-1}L_{\rm ex}^{-1}
\]
proves \eqref{eq:approx-inverse-bound}.
Combining this with
\(\|L\|\leq\|L_{\rm ex}\|+\Delta_h\)
gives \eqref{eq:approx-kappa-bound}.

When \(\theta\leq1/2\),
\[
\kappa(L)
\leq2\bigl(\kappa(L_{\rm ex})+\theta\bigr)
\leq2\kappa(L_{\rm ex})+1
\leq3\kappa(L_{\rm ex}),
\]
where the last step uses \(\kappa(L_{\rm ex})\geq1\).
\end{proof}

\subsubsection{Error in the nodal solution}
\label{app:global-error}

\begin{theorem}[Global nodal error]
\label{thm:global-error}
Under Assumption~\ref{ass:sta&det}, suppose that \(\theta<1\),
\(L_{\rm ex}q_{\rm ex}=b\), and \(Lq=b\), with
\(b\) defined in \eqref{eq:def-b}. Then
\[
\|q_{\rm ex}-q\|
\leq
\frac{\theta}{1-\theta}\,
M_c
\sqrt{(1+M_P^2)
\frac{1-e^{-2\alpha_chm}}{1-e^{-2\alpha_ch}}}\,
\|x^0\|.
\]
If \(\theta\leq1/2\), this becomes
\[
\|q_{\rm ex}-q\|
\leq
2\Delta_h\|L_{\rm ex}^{-1}\|M_c
\sqrt{(1+M_P^2)
\frac{1-e^{-2\alpha_chm}}{1-e^{-2\alpha_ch}}}\,
\|x^0\|.
\]
\end{theorem}

\begin{proof}
Theorem~\ref{thm:approx-shooting-conditioning} ensures that
\(L\) is invertible. Subtracting the two systems gives
\[
q_{\rm ex}-q=L^{-1}(L-L_{\rm ex})q_{\rm ex}.
\]
Consequently,
\begin{equation}\label{eq:appendix-q-qhat-first-bound}
\|q_{\rm ex}-q\|
\leq\|L^{-1}\|\Delta_h\|q_{\rm ex}\|
\leq\frac{\theta}{1-\theta}\|q_{\rm ex}\|.
\end{equation}

For the LQR boundary data, the decoupling theorem gives
\(z(t)=0\), and hence
\[
x(t)=\Psi(t,0)x^0,\qquad
\lambda(t)=P(t)\Psi(t,0)x^0.
\]
At each node,
\[
\begin{aligned}
\|y_k\|^2
&=\|x(t_k)\|^2+\|\lambda(t_k)\|^2\\
&\leq(1+M_P^2)M_c^2e^{-2\alpha_ct_k}\|x^0\|^2.
\end{aligned}
\]
Summing over all nodes yields
\[
\begin{aligned}
\|q_{\rm ex}\|^2
&\leq(1+M_P^2)M_c^2
\sum_{k=0}^{m-1}e^{-2\alpha_ckh}\|x^0\|^2\\
&=(1+M_P^2)M_c^2
\frac{1-e^{-2\alpha_chm}}{1-e^{-2\alpha_ch}}
\|x^0\|^2.
\end{aligned}
\]
Substitution into \eqref{eq:appendix-q-qhat-first-bound}
proves the first estimate. The second follows from
\((1-\theta)^{-1}\leq2\) when \(\theta\leq1/2\).
\end{proof}

For the error allocation in Section~\ref{subsec:gauss-perturbation},
the simpler estimate
\[
\|q_{\rm ex}\|
\leq\|L_{\rm ex}^{-1}\|\|b\|
=\|L_{\rm ex}^{-1}\|\|x^0\|
\]
is sufficient. Inserting it into
\eqref{eq:appendix-q-qhat-first-bound} gives, for
\(\theta\leq1/2\),
\[
\|q_{\rm ex}-q\|
\leq2\Delta_h\|L_{\rm ex}^{-1}\|^2\|x^0\|.
\]

\section{Gauss--Legendre discretization and one-step error}
\label{app:gauss-discretization}
\label{app:runge-kutta approximation}

This appendix gives the coefficient bounds and one-step error estimate
used in Section~\ref{subsec:gauss-perturbation}. Throughout,
\(\|\cdot\|\) denotes the Euclidean vector norm or the induced matrix
norm, and \(\|\cdot\|_{\max}\) denotes the largest absolute matrix entry.

\subsection{Quadrature weights}
\label{app:gauss-weights}

Let \(0<c_0<\cdots<c_{p-1}<1\) be the \(p\) Gauss--Legendre nodes
on \([0,1]\), and let
\begin{equation}\label{eq:def-Lagrange-basis}
\ell_j(s)
=
\prod_{\substack{0\leq k<p\\k\ne j}}
\frac{s-c_k}{c_j-c_k},
\qquad 0\leq j<p.
\end{equation}
Recall from \eqref{eq:def-A-r} that
\[
\mathcal A_{ij}=\int_0^{c_i}\ell_j(s)\,\mathrm ds,
\qquad
r_j=\int_0^1\ell_j(s)\,\mathrm ds.
\]
Write \(\mathbb P_d\) for the space of real polynomials of degree
at most \(d\). The Gauss quadrature rule is exact on
\(\mathbb P_{2p-1}\):
\begin{equation}\label{eq:gauss-quadrature-exactness}
\int_0^1v(s)\,\mathrm ds
=
\sum_{j=0}^{p-1}r_jv(c_j),
\qquad v\in\mathbb P_{2p-1};
\end{equation}
see \cite[\S3.5(v)]{NIST:DLMF}.

\begin{lemma}[Uniform bounds for the quadrature weights]
\label{lem:uniform-bou-pr}
For every \(p\geq1\), the weights satisfy
\[
r_j>0,\qquad
\sum_{j=0}^{p-1}r_j=1,\qquad
pr_j\leq\frac{4\pi}{3},
\qquad 0\leq j<p.
\]
Consequently,
\[
\frac{1}{\sqrt p}
\leq\|r\|
\leq\sqrt{\frac{4\pi}{3p}}.
\]
\end{lemma}

\begin{proof}
Applying \eqref{eq:gauss-quadrature-exactness} to \(\ell_j^2\) gives
\begin{equation}\label{eq:r_j-another-expression2}
r_j=\int_0^1\ell_j(s)^2\,\mathrm ds>0.
\end{equation}
The identity \(\sum_jr_j=1\) follows by applying the same rule
to the constant polynomial. For any \(v\in\mathbb P_{p-1}\)
with \(v(c_j)=1\), exactness and positivity give
\begin{equation}\label{eq:r_j-inequality}
\int_0^1v(s)^2\,\mathrm ds
=
r_j+\sum_{i\ne j}r_iv(c_i)^2
\geq r_j.
\end{equation}
Equality holds for \(v=\ell_j\), so
\begin{equation}\label{eq:r_j-another-expression}
r_j
=
\min_{\substack{v\in\mathbb P_{p-1}\\v(c_j)=1}}
\int_0^1v(s)^2\,\mathrm ds.
\end{equation}

Suppose first that \(p\geq2\). Set
\[
\theta_j=\arccos(2c_j-1)\in(0,\pi),\qquad
S_j=1+2\sum_{k=1}^{p-1}\cos^2(k\theta_j).
\]
Using the Chebyshev polynomials \(T_k(\cos\theta)=\cos(k\theta)\),
define
\[
v_j(s)
=
\frac{1+2\sum_{k=1}^{p-1}\cos(k\theta_j)T_k(2s-1)}{S_j}.
\]
Then \(v_j\in\mathbb P_{p-1}\) and \(v_j(c_j)=1\). Changing
variables by \(s=(1+\cos\theta)/2\) in
\eqref{eq:r_j-another-expression}, we obtain
\[
\begin{aligned}
r_j
&\leq \int_0^1v_j(s)^2\,\mathrm ds\\
&=\frac{1}{2S_j^2}\int_0^\pi
\left(1+2\sum_{k=1}^{p-1}
\cos(k\theta_j)\cos(k\theta)\right)^2
\sin\theta\,\mathrm d\theta\\
&\leq\frac{1}{2S_j^2}\int_0^\pi
\left(1+2\sum_{k=1}^{p-1}
\cos(k\theta_j)\cos(k\theta)\right)^2
\,\mathrm d\theta
=\frac{\pi}{2S_j}.
\end{aligned}
\]
The last equality uses the orthogonality of the cosine functions.

The finite cosine sum gives
\[
S_j
=
p+\sum_{k=1}^{p-1}\cos(2k\theta_j)
=
\frac{2p-1}{2}
+\frac{\sin((2p-1)\theta_j)}{2\sin\theta_j}.
\]
If \(\sin((2p-1)\theta_j)\geq0\), then
\(S_j\geq(2p-1)/2\). Otherwise, let
\(\eta_j=\min\{\theta_j,\pi-\theta_j\}\). Since \(2p-1\) is odd,
\[
\sin((2p-1)\eta_j)
=
\sin((2p-1)\theta_j)<0,
\]
which implies \(\eta_j>\pi/(2p-1)\). As
\(0<\eta_j\leq\pi/2\), we have
\[
\sin\theta_j=\sin\eta_j
\geq\sin\frac{\pi}{2p-1}
\geq\frac{2}{2p-1}.
\]
It follows that
\[
S_j\geq\frac{2p-1}{2}-\frac{2p-1}{4}
=\frac{2p-1}{4}.
\]
This bound holds in both cases. Therefore
\[
pr_j\leq\frac{2\pi p}{2p-1}\leq\frac{4\pi}{3}.
\]
For \(p=1\), we have \(r_0=1\), so the same bound holds.

Finally, Cauchy--Schwarz gives
\(1=\sum_jr_j\leq\sqrt p\,\|r\|\), while
\[
\|r\|^2
=
\sum_{j=0}^{p-1}r_j^2
\leq\left(\max_jr_j\right)\sum_{j=0}^{p-1}r_j
\leq\frac{4\pi}{3p}.
\]
\end{proof}

\subsection{Uniform bounds for the coefficient matrix}
\label{app:gauss-coefficient-bounds}

We first construct a bounded polynomial approximation of a step
function. This will be used to estimate the entries of \(\mathcal A\).

\begin{lemma}\label{lem:h_a-polynomial-approx}
Let \(a\in[0,1]\) and \(h_a(s)=\mathbf 1_{[0,a]}(s)\) on \([0,1]\).
For every \(p\geq2\), there exists \(g_a\in\mathbb P_{p-1}\) such that
\[
0\leq g_a(s)\leq1\quad(0\leq s\leq1),
\qquad
\|h_a-g_a\|_{L^2[0,1]}^2\leq\frac{3}{\pi p}.
\]
\end{lemma}

\begin{proof}
Consider the Fej\'er kernel
\begin{equation}\label{eq:Fejer-kernel}
F_{p-1}(\theta)
=
\frac1p\left(\frac{\sin(p\theta/2)}{\sin(\theta/2)}\right)^2
=
\sum_{|k|<p}\left(1-\frac{|k|}{p}\right)e^{ik\theta},
\end{equation}
with its values at multiples of \(2\pi\) defined by continuity.
This kernel is even, nonnegative, and satisfies
\((2\pi)^{-1}\int_{-\pi}^{\pi}F_{p-1}(\theta)\,\mathrm d\theta=1\);
see \cite[\S1.15(iii)]{NIST:DLMF}.

Set \(\alpha=\arccos(2a-1)\in(0,\pi)\), and define the even \(2\pi\)-periodic
function
\[
f_\alpha(\theta)=\mathbf 1_{[\alpha,\pi]}(|\theta|),
\qquad -\pi\leq\theta\leq\pi.
\]
Then \(h_a(s)=f_\alpha(\arccos(2s-1))\). Let
\begin{equation}\label{eq:def-g-alpha}
J_\alpha(\theta)
=
\frac1{2\pi}\int_{-\pi}^{\pi}
F_{p-1}(\theta-t)f_\alpha(t)\,\mathrm dt.
\end{equation}
Since \(F_{p-1}/(2\pi)\) is nonnegative and has unit integral, it follows that \(0\leq J_\alpha\leq 1\).
The Fourier coefficients of \(f_\alpha\) are
\[
\widehat f_\alpha(k)
=
\frac1{2\pi}\int_{-\pi}^{\pi}f_\alpha(t)e^{-ikt}\,\mathrm dt
=
\begin{cases}
1-\alpha/\pi,&k=0,\\[2pt]
-\sin(k\alpha)/(\pi k),&k\ne0.
\end{cases}
\]
Substituting \eqref{eq:Fejer-kernel} into \eqref{eq:def-g-alpha} then gives
\begin{equation}\label{eq:J-alpha-Fourier-series}
\begin{aligned}
J_\alpha(\theta)
&=\sum_{|k|<p}
\left(1-\frac{|k|}{p}\right)\widehat f_\alpha(k)e^{ik\theta}\\
&=\widehat f_\alpha(0)
+2\sum_{k=1}^{p-1}
\left(1-\frac{k}{p}\right)\widehat f_\alpha(k)\cos(k\theta).
\end{aligned}
\end{equation}
Hence
\[
g_a(s):=J_\alpha(\arccos(2s-1))
=
\widehat f_\alpha(0)
+2\sum_{k=1}^{p-1}
\left(1-\frac{k}{p}\right)\widehat f_\alpha(k)T_k(2s-1)
\]
belongs to \(\mathbb P_{p-1}\) and takes values in \([0,1]\).

For \(1\leq|k|<p\), the \(k\)th Fourier coefficient of
\(f_\alpha-J_\alpha\) is
\((|k|/p)\widehat f_\alpha(k)\). For \(|k|\geq p\), it is
\(\widehat f_\alpha(k)\), and the zeroth coefficient vanishes.
Parseval's identity therefore gives
\[
\begin{aligned}
\frac1{2\pi}\int_{-\pi}^{\pi}
|f_\alpha(\theta)-J_\alpha(\theta)|^2\,\mathrm d\theta
&=
2\sum_{k=1}^{p-1}\frac{k^2}{p^2}|\widehat f_\alpha(k)|^2
+2\sum_{k=p}^{\infty}|\widehat f_\alpha(k)|^2\\
&\leq\frac{2}{\pi^2}
\left(\frac{p-1}{p^2}+\sum_{k=p}^{\infty}\frac1{k^2}\right)
\leq\frac{6}{\pi^2p}.
\end{aligned}
\]
Here we used \(|\widehat f_\alpha(k)|\leq1/(\pi|k|)\) and
\[
\sum_{k=p}^{\infty}\frac1{k^2}
\leq\int_{p-1}^{\infty}\frac{\mathrm dt}{t^2}
=\frac1{p-1}\leq\frac2p.
\]
Changing variables back to \(s\), and using evenness, yields
\[
\begin{aligned}
\|h_a-g_a\|_{L^2[0,1]}^2
&=\frac12\int_0^\pi
|f_\alpha(\theta)-J_\alpha(\theta)|^2
\sin\theta\,\mathrm d\theta\\
&\leq\frac14\int_{-\pi}^{\pi}
|f_\alpha(\theta)-J_\alpha(\theta)|^2\,\mathrm d\theta
\leq\frac{3}{\pi p}.
\end{aligned}
\]
\end{proof}

\begin{lemma}\label{lem:mathcal-A-norm-bound}
For every \(p\geq1\),
\[
p\|\mathcal A\|_{\max}\geq\frac12,
\qquad
\|\mathcal A\|
\leq p\|\mathcal A\|_{\max}
\leq\frac{4\pi}{3}+2
\leq2\pi.
\]
\end{lemma}

\begin{proof}
Since \(\sum_j\ell_j(s)=1\) and \(c_{p-1}\geq1/2\),
\[
p\|\mathcal A\|_{\max}
\geq\left|\sum_{j=0}^{p-1}\mathcal A_{p-1,j}\right|
=\int_0^{c_{p-1}}1\,\mathrm ds
=c_{p-1}\geq\frac12.
\]

For the upper bound, suppose first that \(p\geq2\). Let \(P_k\)
be the Legendre polynomial normalized by \(P_k(1)=1\), and set
\[
\varphi_k(s)=\sqrt{2k+1}\,P_k(2s-1),
\qquad
K_{p-1}(s,t)=\sum_{k=0}^{p-1}\varphi_k(s)\varphi_k(t).
\]
The polynomials \(\varphi_k\) are orthonormal in \(L^2[0,1]\).
Thus \(K_{p-1}\) is the reproducing kernel for \(\mathbb P_{p-1}\):
\[
v(s)
=
\sum_{k=0}^{p-1}\varphi_k(s)
\int_0^1\varphi_k(t)v(t)\,\mathrm dt
=
\int_0^1K_{p-1}(s,t)v(t)\,\mathrm dt,
\qquad v\in\mathbb P_{p-1}.
\]
See also \cite[\S18.2(v)]{NIST:DLMF}. Since
\(K_{p-1}(s,\cdot)\ell_j(\cdot)\in\mathbb P_{2p-2}\), exact
Gauss quadrature gives
\begin{equation}\label{eq:ell_j-kernel}
\begin{aligned}
\ell_j(s)
&=\int_0^1K_{p-1}(s,t)\ell_j(t)\,\mathrm dt\\
&=\sum_{i=0}^{p-1}r_iK_{p-1}(s,c_i)\ell_j(c_i)
=r_jK_{p-1}(s,c_j).
\end{aligned}
\end{equation}
Taking \(s=c_j\) shows that \(r_jK_{p-1}(c_j,c_j)=1\).
Orthonormality also gives
\begin{equation}\label{eq:kernel-integral}
\int_0^1K_{p-1}(s,c_j)^2\,\mathrm ds
=
\sum_{k=0}^{p-1}\varphi_k(c_j)^2
=
K_{p-1}(c_j,c_j)
=\frac1{r_j}.
\end{equation}

With \(h_a\) as in Lemma~\ref{lem:h_a-polynomial-approx} and \eqref{eq:ell_j-kernel},
\begin{equation}\label{eq:mathcal_A-equation}
\mathcal A_{ij}
=
\int_0^{c_i} \ell_j(s) \mathrm ds
=
r_j\int_0^1h_{c_i}(s)K_{p-1}(s,c_j)\,\mathrm ds.
\end{equation}
Choose the polynomial \(g_{c_i}\) from that lemma. The reproducing
property gives
\[
\mathcal A_{ij}
=
r_jg_{c_i}(c_j)
+
r_j\int_0^1
\bigl(h_{c_i}(s)-g_{c_i}(s)\bigr)K_{p-1}(s,c_j)\,\mathrm ds.
\]
By Cauchy--Schwarz and \eqref{eq:kernel-integral},
\[
|\mathcal A_{ij}|
\leq r_j
+r_j\|h_{c_i}-g_{c_i}\|_{L^2[0,1]}
\|K_{p-1}(\cdot,c_j)\|_{L^2[0,1]}
\leq r_j+\sqrt{\frac{3r_j}{\pi p}}.
\]
Lemma~\ref{lem:uniform-bou-pr} now yields
\[
p|\mathcal A_{ij}|
\leq pr_j+\sqrt{\frac{3pr_j}{\pi}}
\leq\frac{4\pi}{3}+2.
\]
For \(p=1\), \(\mathcal A_{00}=1/2\), so this bound also holds.
Finally,
\[
\|\mathcal A\|
\leq\|\mathcal A\|_{\mathrm F}
\leq p\|\mathcal A\|_{\max},
\]
which proves the result.
\end{proof}

In particular, the normalization used for the dense block-encoding,
\[
\alpha_{\mathcal A}
=
2^{\lceil\log_2p\rceil}\|\mathcal A\|_{\max},
\]
satisfies
\begin{equation}\label{eq:gauss-coefficient-normalization}
\frac12
\leq p\|\mathcal A\|_{\max}
\leq\alpha_{\mathcal A}
\leq2p\|\mathcal A\|_{\max}
\leq4\pi.
\end{equation}
Thus \(\alpha_{\mathcal A}=\Theta(1)\), uniformly in \(p\).

\subsection{The stage system and symplecticity}
\label{app:gauss-stage-system}

For \(y'=Hy\), write the stage derivatives as \(K_i y\), where
\(K_i\in\mathbb R^{2n\times2n}\) satisfy
\[
K_i
=
H\left(I_{2n}+h\sum_{j=0}^{p-1}\mathcal A_{ij}K_j\right),
\qquad 0\leq i<p.
\]
With \(\mathcal K=(K_0^\top,\ldots,K_{p-1}^\top)^\top\), these
equations become
\[
\mathcal G_h\mathcal K=\mathbf1_p\otimes H,
\qquad
\mathcal G_h=I_p\otimes I_{2n}-\mathcal A\otimes hH.
\]
If \(\mathcal G_h\) is invertible, the update
\(y\mapsto y+h\sum_i r_iK_i y\) gives exactly the propagator
\(\Phi\) in \eqref{eq:rk-propagator-definition}.

\begin{lemma}\label{lem:gauss-rk-propagator-v2}
If \(h\|H\|<1\), then \(\mathcal G_h\) is invertible.
If, in addition, \(h\|H\|<1/(4\pi)\), then
\[
\frac12<\|\mathcal G_h\|<\frac32,\qquad
\frac23<\|\mathcal G_h^{-1}\|<2,\qquad
\kappa(\mathcal G_h)<3,\qquad
\|\Phi\|\leq2.
\]
These bounds are independent of \(p\).
\end{lemma}

\begin{proof}
The Gauss coefficient matrix satisfies \(\rho(\mathcal A)<1\),
where \(\rho\) denotes the spectral radius; see the proof of
\cite[Theorem~3.11]{wu2024structure}. Hence, with
\(\xi=h\|H\|<1\),
\[
\rho(\mathcal A\otimes hH)
=
\rho(\mathcal A)\rho(hH)
\leq\rho(\mathcal A)\xi<1.
\]
Thus \(\mathcal G_h\) is invertible.

Under the stronger condition \(\xi<1/(4\pi)\),
Lemma~\ref{lem:mathcal-A-norm-bound} gives
\[
\eta:=\|\mathcal A\otimes hH\|
=\|\mathcal A\|\xi<\frac12.
\]
The triangle inequality and the Neumann series imply
\[
1-\eta\leq\|\mathcal G_h\|\leq1+\eta,
\qquad
\|\mathcal G_h^{-1}\|\leq\frac1{1-\eta}<2.
\]
Also,
\[
\|\mathcal G_h^{-1}\|
\geq\frac1{\|\mathcal G_h\|}>\frac23,
\qquad
\kappa(\mathcal G_h)
\leq\frac{1+\eta}{1-\eta}<3.
\]
Finally, by \eqref{eq:rk-propagator-definition} and
Lemma~\ref{lem:uniform-bou-pr},
\[
\begin{aligned}
\|\Phi\|
&\leq
1+h\|r\|\,\|\mathcal G_h^{-1}\|
\|\mathbf1_p\|\,\|H\|\\
&\leq
1+2\sqrt{\frac{4\pi}{3}}\,\xi
<1+\frac1{\sqrt{3\pi}}<2.
\end{aligned}
\]
\end{proof}

The Hamiltonian matrix of the LQR system satisfies
\[
H^\top\mathbb J_{2n}+\mathbb J_{2n}H=0,
\qquad
\mathbb J_{2n}
=
\begin{bmatrix}
0&I_n\\
-I_n&0
\end{bmatrix}.
\]
The Gauss method preserves the corresponding symplectic form.

\begin{lemma}[Symplecticity]\label{lem:gauss-symplecticity}
If \(\mathcal G_h\) is invertible and
\(H^\top\mathbb J_{2n}+\mathbb J_{2n}H=0\), then
\[
\Phi^\top\mathbb J_{2n}\Phi=\mathbb J_{2n}.
\]
\end{lemma}

\begin{proof}
Let \(L_j(s)=\int_0^s\ell_j(t)\,\mathrm dt\), and define
\[
\mathcal Y(s)
=
I_{2n}+h\sum_{j=0}^{p-1}L_j(s)K_j.
\]
Then \(\mathcal Y(0)=I_{2n}\), \(\mathcal Y(1)=\Phi\), and the
stage equations imply
\[
\mathcal Y'(c_i)=hK_i=hH\mathcal Y(c_i).
\]
Each entry of
\(\frac{\mathrm d}{\mathrm ds}
(\mathcal Y(s)^\top\mathbb J_{2n}\mathcal Y(s))\)
has degree at most \(2p-1\). Applying exact Gauss quadrature
entrywise gives
\[
\begin{aligned}
\Phi^\top\mathbb J_{2n}\Phi-\mathbb J_{2n}
&=
\int_0^1
\bigl(\mathcal Y'(s)^\top\mathbb J_{2n}\mathcal Y(s)
+\mathcal Y(s)^\top\mathbb J_{2n}\mathcal Y'(s)\bigr)
\,\mathrm ds\\
&=
h\sum_{i=0}^{p-1}r_i\mathcal Y(c_i)^\top
\bigl(H^\top\mathbb J_{2n}+\mathbb J_{2n}H\bigr)
\mathcal Y(c_i)
=0.
\end{aligned}
\]
\end{proof}

\subsection{One-step approximation error}
\label{app:rk-global-error-proof-v2}

We now prove the one-step bound \eqref{eq:rk-one-step-error}.
This estimate only requires \(h\|H\|<1\).

\begin{lemma}[One-step propagator approximation error]
\label{lem:rk-global-error-v2}
Let \(\xi=h\|H\|<1\) and \(\Phi_{\rm ex}=e^{hH}\). Then
\[
\Delta_h:=\|\Phi-\Phi_{\rm ex}\|
\leq
\frac{e^\xi}{2-e^{\xi/2}}
\frac{(p!)^2}{(2p)!(2p+1)!}\xi^{2p+1}
\leq
\frac{8(p!)^2}{(2p)!(2p+1)!}\xi^{2p+1}.
\]
In particular,
\begin{equation}\label{eq:one-step-error-scaling}
\Delta_h
=
\mathcal O\!\left(
\frac{p^{1/2}}{2^{2p}(2p+1)!}
\right).
\end{equation}
\end{lemma}

\begin{proof}
The scalar stability function of the Gauss method is
\[
\mathscr R_p(z)
=
1+z\,r^\top(I_p-z\mathcal A)^{-1}\mathbf1_p.
\]
By Lemma~\ref{lem:gauss-rk-propagator-v2}, the stage system is
invertible for \(\xi<1\), and
\(\Phi=\mathscr R_p(hH)\); see
\cite[Eq.~(350a)]{butcher2008numerical} for the matrix form of
the Runge--Kutta update.

The stability function is the \([p/p]\) Pad\'e approximant of
\(e^z\); see \cite[Table~3, p.~1285]{hairer2015runge}. Thus, by \cite[Method 2]{moler1978nineteen}
\[
\mathscr R_p(z)=Q_p(z)^{-1}N_p(z),
\qquad
N_p(z)=\sum_{k=0}^p a_{p,k}z^k,
\qquad
Q_p(z)=N_p(-z),
\]
where
\[
a_{p,k}
=
\frac{(2p-k)!\,p!}{(2p)!\,(p-k)!\,k!},
\qquad 0\leq k\leq p.
\]
Set \(X=hH\). We first bound \(Q_p(X)^{-1}\).
For \(1\leq k\leq p\),
\[
a_{p,k}
=
\frac1{k!}\prod_{\ell=0}^{k-1}
\frac{p-\ell}{2p-\ell}
\leq\frac1{2^k k!}.
\]
It follows that
\[
\|Q_p(X)-I_{2n}\|
\leq\sum_{k=1}^p a_{p,k}\xi^k
\leq e^{\xi/2}-1<1.
\]
The Neumann series therefore gives
\begin{equation}\label{eq:Qp-inverse-bound}
\|Q_p(X)^{-1}\|
\leq
\frac1{1-\|Q_p(X)-I_{2n}\|}
\leq\frac1{2-e^{\xi/2}}.
\end{equation}
In particular,
\(\Phi=\mathscr R_p(X)=Q_p(X)^{-1}N_p(X)\) is well defined.

The diagonal Pad\'e remainder formula yields
\begin{equation}\label{eq:pade-remainder-matrix}
e^X-\mathscr R_p(X)
=
(-1)^p Q_p(X)^{-1}
\frac{X^{2p+1}}{(2p)!}
\int_0^1
e^{(1-\theta)X}\theta^p(1-\theta)^p\,\mathrm d\theta;
\end{equation}
see \cite[Appendix~1, proof of Lemma~3]{moler1978nineteen}.
To obtain the matrix identity, multiply the scalar remainder
formula by \(Q_p(z)\) and compare its entire power series.
Substitution of \(X\) is valid term by term, with uniform
convergence in \(\theta\in[0,1]\). No normality assumption on
\(H\) is needed.

Taking norms in \eqref{eq:pade-remainder-matrix}, using
\(\|e^{(1-\theta)X}\|\leq e^{(1-\theta)\xi}\leq e^\xi\), and
applying \eqref{eq:Qp-inverse-bound}, we obtain
\[
\begin{aligned}
\|e^X-\mathscr R_p(X)\|
&\leq
\|Q_p(X)^{-1}\|
\frac{\xi^{2p+1}}{(2p)!}
\int_0^1
\|e^{(1-\theta)X}\|\theta^p(1-\theta)^p\,\mathrm d\theta\\
&\leq
\frac{e^\xi}{2-e^{\xi/2}}
\frac{\xi^{2p+1}}{(2p)!}
\int_0^1\theta^p(1-\theta)^p\,\mathrm d\theta\\
&=
\frac{e^\xi}{2-e^{\xi/2}}
\frac{(p!)^2}{(2p)!(2p+1)!}\xi^{2p+1}.
\end{aligned}
\]
The last equality uses
\(\int_0^1\theta^p(1-\theta)^p\,\mathrm d\theta
=(p!)^2/(2p+1)!\).
Since \(0\leq\xi<1\),
\[
\frac{e^\xi}{2-e^{\xi/2}}
\leq\frac{e}{2-\sqrt e}<8,
\]
which proves \eqref{eq:rk-one-step-error}. Finally, Stirling's
formula gives
\[
\frac{(p!)^2}{(2p)!}
=
\Theta\!\left(\frac{p^{1/2}}{2^{2p}}\right),
\]
so
\[
\Delta_h
=
\mathcal O\!\left(
\frac{p^{1/2}}{2^{2p}(2p+1)!}\xi^{2p+1}
\right).
\]
Using \(\xi<1\) proves \eqref{eq:one-step-error-scaling}.
\end{proof}

For completeness, the logarithmic choice of \(p\) stated in
Section~\ref{subsec:gauss-perturbation} follows directly from
this bound. Define
\[
C_p=\frac{8(p!)^2}{(2p)!(2p+1)!}.
\]
Then
\[
C_1=\frac23,\qquad
\frac{C_{p+1}}{C_p}
=
\frac1{4(2p+1)(2p+3)}
\leq\frac1{60},
\]
and hence \(\Delta_h\leq(2/3)60^{-(p-1)}\) for \(\xi<1\).
Thus a one-step tolerance \(\delta\in(0,1)\) can be attained with
\(p=\mathcal O(1+\log(1/\delta))\). The resulting perturbation of
the nodal system is estimated in
Appendix~\ref{app:approx-conditioning-proof}.

\section{Auxiliary block-encoding constructions}
\label{app:auxiliary-block-encodings}

We give the projections, shifts, and rectangular factors used in
Appendix~\ref{app:quantum-algorithm-proofs}. All embeddings follow
Definition~\ref{def:block-encoding}. In this section, a dirty ancilla is a workspace qubit whose initial state need not be known but which must be restored to that state at the end of the computation. We exclude such workspace qubits from the ancilla count of the block-encoding. 
Gate counts below refer to one-qubit and CNOT gates.

\subsection{Coordinate projections}

\begin{lemma}\label{lem:projection-block-encoding}
Let \(1\leq a\leq n\). The matrices
\[
\Pi_{n,a}=\operatorname{diag}(I_a,0_{n-a}),
\qquad
\Pi_{n,a}'=\operatorname{diag}(0_{n-a},I_a)
\]
admit exact \((1,1)\)-block-encodings using
\(\mathcal O(\log(2n))\) gates.
\end{lemma}

\begin{proof}
Set \(N=2^{l_n}\). For \(0\leq t<N\), consider the comparator
\begin{equation}\label{eq:def-U_t}
U_t\ket z_1\ket j_{l_n}
=
\ket{z\oplus\mathbf1_{\{j>t\}}}_1\ket j_{l_n}.
\end{equation}
The classical--quantum comparator in
\cite[Section~4.2]{vandaele2026comparators} implements this map
with \(\mathcal O(l_n+1)\) gates over the \{$X$,CNOT, Toffoli\} gate set and one dirty workspace qubit.
Each Toffoli gate has a constant-size decomposition into the
gate set used here. For \(0\leq j\leq N\), we have
\[
U_{a-1}|j\rangle_{l_n}|0\rangle
=
\begin{cases}
|j\rangle_{l_n}|0\rangle,
& j\leq a-1,\\[2mm]
|j\rangle_{l_n}|1\rangle,
& j > a-1.
\end{cases}
\]
Consequently,
\[
(\bra0\otimes I_N)U_{a-1}(\ket0\otimes I_N)
=
\sum_{j=0}^{a-1}\ket j\bra j
=
(\Pi_{n,a})_{\rm pad}.
\]
\(U_{a-1}\) is an exact \((1,1)\)-block-encoding of \(\Pi_{n,a}\).

If \(a=n\), the same construction with \(U_{n-1}\) encodes
\(\Pi_{n,n}'=I_n\). If \(a<n\), set
\[
\widetilde U_a
=
U_{n-1}U_{n-a-1}(X\otimes I_N).
\]
Starting with a zero signal \(\ket{0}\), the final value of the signal qubit is
\[
1\oplus\mathbf1_{\{j>n-a-1\}}
 \oplus\mathbf1_{\{j>n-1\}}.
\]
It is zero exactly for \(n-a\leq j<n\). Consequently,
\[
(\bra0\otimes I_N)\widetilde U_a(\ket0\otimes I_N)
=
\sum_{j=n-a}^{n-1}\ket j\bra j
=
(\Pi_{n,a}')_{\rm pad}.
\]
\(\widetilde U_a\) is an exact \((1,1)\)-block-encoding of \(\Pi'_{n,a}\). Only a constant number of comparators is used.
\end{proof}

\subsection{Shifts and matrix units}

\begin{lemma}\label{lem:shift-matrix-block-encoding}
For \(n\geq2\), the shift
\begin{equation}
\label{eq:shift}
S_n=\sum_{j=1}^{n-1}\ket j^{(n)}\bra{j-1}^{(n)}
\end{equation}
and every matrix unit
\(E_{ij}=\ket i^{(n)}\bra j^{(n)}\), \(i,j\in[n]\),
admit exact \((1,1)\)-block-encodings using
\(\mathcal O(\log(2n))\) gates. In particular, this applies to
\(E_{0,n-1}\).
\end{lemma}

\begin{proof}
Let \(N=2^{l_n}\), and define
\[
O_{\rm shift}\ket j_{l_n}
=
\ket{(j+1)\bmod N}_{l_n}.
\]
By \cite[Theorem~4]{vandaele2026comparators}, this increment
uses \(\mathcal O(l_n+1)\) gates and one dirty workspace qubit.
The padded trailing projection satisfies
\[
(\Pi_{n,n-1}')_{\rm pad}O_{\rm shift}
=
\sum_{j=1}^{n-1}\ket j_{l_n}\bra{j-1}_{l_n}
=
(S_n)_{\rm pad}.
\]
Thus composing the increment with the projection encoding from
Lemma~\ref{lem:projection-block-encoding} gives the required
encoding of \(S_n\).

For a fixed \(j\in[n]\), let \(V_j\) apply a Pauli-\(X\) gate
at each nonzero bit of \(j\). Then
\(V_j\ket k=\ket{k\oplus j}\), \(V_j^\dagger=V_j\), and
\[
V_i\ket0\bra0V_j=\ket i\bra j.
\]
If \(U_{\Pi_{n,1}}\) encodes \(\ket0\bra0\), then
\[
W_{ij} = (I_2\otimes V_i)U_{\Pi_{n,1}}(I_2\otimes V_j)
\]
encodes \(E_{ij}\), because
\[
\begin{aligned}
(\bra0\otimes I_N)W_{ij}(\ket0\otimes I_N) 
&= 
V_i \left[(\bra0\otimes I_N)U_{\Pi_{n,1}}(\ket0\otimes I_N)\right]V_j
&=
V_i\ket0\bra0V_j=\ket i\bra j.
\end{aligned}
\]
The two bit-flip permutations \(V_i\), \(V_j\) require at most
\(2l_n\) gates.
\end{proof}

For the type register, \(\ket0\bra1\) has the particularly simple
encoding shown in Figure~\ref{fig:block-encoding1}. Indeed,
applying \(X\) to the system qubit and then a CNOT from the system
to the signal gives
\[
(\bra0_{\rm sig}\otimes I_2)
\,\mathrm{CNOT}_{{\rm sys}\to{\rm sig}}(I_2\otimes X)
(\ket0_{\rm sig}\otimes I_2)
=
\ket0\bra1.
\]

\begin{figure}[htbp]
\centering
\begin{quantikz}[
row sep={0.5cm,between origins},
column sep=0.6cm
]
\lstick{\(\ket0_{\rm sig}\)}
& \qw & \targ{} & \qw \\
\lstick{\(\ket\psi_{\rm sys}\)}
& \gate{X} & \ctrl{-1} & \qw
\end{quantikz}
\caption{An exact \((1,1)\)-block-encoding of \(\ket0\bra1\).}
\label{fig:block-encoding1}
\end{figure}
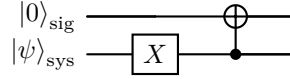

\subsection{Extension and compression of system registers}
\label{app:register-padding}

The rectangular products below require compatible embeddings.
For integers \(u\leq v\), write
\begin{equation}
\label{eq:def-e()}
e(u,v)=
\begin{cases}
0,&u=v,\\
1,&u<v.
\end{cases}
\end{equation}

\begin{lemma}\label{lem:register-extension-compression}
Suppose that \(U_A\) is an
\((\alpha,a,\varepsilon)\)-block-encoding on an \(r\)-qubit
system register.
\begin{enumerate}
\item Its zero extension to an \(s\)-qubit system register,
\(s\geq r\), has parameters
\((\alpha,a+e(r,s),\varepsilon)\). It uses one call to \(U_A\)
and \(\mathcal O(s+1)\) additional gates.
\item Suppose that the target matrix is supported on the first
\(2^t\) coordinates, where \(t\leq r\). Moving the other \(r-t\)
system qubits into the projected register gives an
\((\alpha,a+r-t,\varepsilon)\)-block-encoding on \(t\) system
qubits.
\end{enumerate}
\end{lemma}

\begin{proof}
Let \(\widetilde A\) be the unnormalized encoded block of \(U_A\).
For \(s>r\), apply \(U_A\) to the low \(r\) system qubits and
combine it with an exact encoding of
\(P=\ket0_{s-r}\bra0_{s-r}\) on the high qubits. The encoded block
is \(P\otimes\widetilde A\), and
\[
\|P\otimes\widetilde A-P\otimes A_{\rm pad}\|
=
\|\widetilde A-A_{\rm pad}\|
\leq\varepsilon.
\]
Lemma~\ref{lem:projection-block-encoding} supplies the additional
signal qubit and gate bound. No extension is needed when \(s=r\).

For compression, put
\(V=\ket0_{r-t}\otimes I_{2^t}\). The new encoded block is
\(V^\dagger\widetilde A V\), and
\[
\|V^\dagger(\widetilde A-A_{\rm pad})V\|
\leq\varepsilon.
\]
The \(r-t\) high system qubits are now initialized and projected
onto zero, so they count as block-encoding ancillas.
\end{proof}

\subsection{Gauss factors and readout operators}
\label{app:auxiliary-factors}

Let \(N_p=2^{l_p}\), \(P_0=\ket0_{l_p}\bra0_{l_p}\), and
\(P_p=\sum_{j=0}^{p-1}\ket j_{l_p}\bra j_{l_p}\).
The Gauss weights satisfy
\begin{equation}\label{eq:r-bound}
\|r\|\leq\|r\|_1=\sum_{j=0}^{p-1}r_j=1,
\end{equation}
by Lemma~\ref{lem:uniform-bou-pr}. Let \(U_r\ket0=\ket r\).
Then
\[
P_0U_r^\dagger=\ket0\bra r,
\qquad
P_p H_{\rm Had}^{\otimes l_p}P_0
=
2^{-l_p/2}\sum_{j=0}^{p-1}\ket j\bra0,
\]
where \(H_{\rm Had}\) is the Hadamard gate. These are the padded
row and column factors, respectively. The projection and product
constructions give exact encodings
\begingroup
\renewcommand{\arraystretch}{1.35}
\[
\begin{array}{c|c}
\text{matrix} & (\text{normalization},\text{ancillas})\\ \hline
r^\top       & (\|r\|,1)\\
\mathbf{1}_p & (2^{l_p/2},2)
\end{array}
\]
\endgroup
The first uses one call to \(U_r^\dagger\); the second uses
\(\mathcal O(l_p+1)\) gates. All products use separate signal
ancillas.

The coefficient matrix \(\mathcal A\) has an exact encoding with
\[
\alpha_{\mathcal A}
=
2^{l_p}\|\mathcal A\|_{\max},
\qquad a_{\mathcal A}=l_p+1,
\qquad
\frac12\leq\alpha_{\mathcal A}\leq4\pi,
\]
by Lemma~\ref{lem:be-dense} and
\eqref{eq:gauss-coefficient-normalization}. Its implementation
uses \(\mathcal O(p^2)\) gates. A recursive preparation by
controlled rotations gives \(U_r\) using \(\mathcal O(p)\) gates.
These costs are included when the auxiliary oracles are expanded
in Appendix~\ref{app:be-L-construction}.

For the output operators, let \(E_{ij}\) now denote the
\(m\times m\) matrix units and define
\begin{equation}\label{eq:cost-transfer-operator-v2}
E_0=E_{00}\otimes\ket0\bra1\otimes I_n.
\end{equation}
Recall also
\[
S_{\lambda,k}
=
\bra k^{(m)}\otimes\bra1^{(2)}\otimes I_n
\in\mathbb R^{n\times2nm}.
\]
The matrix \(E_0\) moves the initial costate into the initial-state
slot, while \(S_{\lambda,k}\) extracts the costate at node \(k\).
Both have norm one.

\begin{lemma}\label{lem:readout-operator-encoding}
The operators \(E_0\) and \(S_{\lambda,k}\) admit exact
\((1,3+\tau_{m,n})\)-block-encodings. Each construction uses
\[
\mathcal O(\log^2(2m)\log^2(2n))
\]
gates.
\end{lemma}

\begin{proof}
The factors \(E_{00}\), \(\ket0\bra1\), and \(I_n\) each have
exact \((1,1)\)-block-encodings. For the selector, the padded
embedding of \(\bra k^{(m)}\) is \(\ket0\bra k\), obtained from
Lemma~\ref{lem:shift-matrix-block-encoding}; the padded embedding
of \(\bra1^{(2)}\) is \(\ket0\bra1\).
Apply the Kronecker-product rule in
Lemma~\ref{lem:manipulation-of-matrices}.
The first product has no padding gap, and the second contributes
\[
(l_m+1)+l_n-\lceil\log_2(2mn)\rceil=\tau_{m,n}.
\]
This gives the stated ancilla and gate counts.
\end{proof}

\section{Quantum implementation and output guarantees}
\label{app:quantum-algorithm-proofs}

We use the primitives in Appendix~\ref{app:quantum-algorithm-primitives}
and the auxiliary constructions in
Appendix~\ref{app:auxiliary-block-encodings}. Put
\[
\rho_R=\|R^{-1}\|,\qquad
\kappa_R=\kappa(R),\qquad c_{\rm inv}=1+\sqrt5.
\]
Queries include controlled calls and calls to inverse oracles.
Zero summands are omitted. Inversion parameters may use available
upper bounds on the corresponding inverse norms.

\subsection{Block-encoding of the Hamiltonian matrix}
\label{app:be-H-construction}

\begin{lemma}[Block-encoding of \(H\)]
\label{lem:be-H-construction}
Suppose that Assumption~\ref{ass:be} holds and \(H\ne0\).
For every \(\varepsilon_H>0\), there is an
\((\alpha_H,a_H,\varepsilon_H)\)-block-encoding of \(H\), where
\begin{equation}\label{eq:alpha-H}
\alpha_H
=
2\alpha_A+\alpha_Q+8\alpha_B^2\rho_R
=
2s\|A\|_{\max}+s\|Q\|_{\max}
+8s^2\|B\|_{\max}^2\rho_R.
\end{equation}
Let \(l_*=\max\{l_n,l_d\}\) and
\[
b_H=2a_B+a_R+1+e(l_d,l_*)+l_*-l_n,
\]
the definition of \(e(\cdot, \cdot)\) is given in \eqref{eq:def-e()}. A sufficient ancilla count is
\begin{equation}\label{eq:a_H}
a_H=3+\max\{a_A,a_Q,b_H\}
=\mathcal O(\log(2n)+\log(2d)).
\end{equation}
The construction uses a constant number of calls to
\(U_A,U_B,U_Q\) and their inverses, and
\[
\mathcal O\!\left(
s\kappa_R\log\!\left(
2+\frac{s^2\kappa_R\|H\|}{\varepsilon_H}
\right)\right)
\]
calls to \(U_R\) and \(U_R^\dagger\). The additional gate count is
\[
\mathcal O\!\left(
s\kappa_R\log\!\left(
2+\frac{s^2\kappa_R\|H\|}{\varepsilon_H}
\right)\log(2d)
+\log^2(2n)+\log(2d)
\right).
\]
\end{lemma}

\begin{proof}
Since \(\alpha_R\geq\|R\|\), we may apply
Lemma~\ref{lem:inverse-be-approx} with
\(\kappa_0=2\alpha_R\rho_R\geq2\). For any
\(\varepsilon_{R^{-1}}>0\), this gives an
\((8\rho_R,a_R+1,\varepsilon_{R^{-1}})\)-block-encoding of
\(R^{-1}\), using
\[
\mathcal O\!\left(
\alpha_R\rho_R
\log\!\left(2+\frac{\rho_R}{\varepsilon_{R^{-1}}}\right)
\right)
\]
queries to \(U_R\) and its inverse. The additional gate count has
a further factor \(a_R+1\).

Assume first that \(B\ne0\), and choose
\[
\varepsilon_{R^{-1}}=\frac{\varepsilon_H}{\alpha_B^2}.
\]
Extend the encoding of \(R^{-1}\) to the \(l_*\)-qubit system
register of \(U_B\), using
Lemma~\ref{lem:register-extension-compression}. Multiplication by
\(U_B\) and \(U_B^\dagger\) then encodes \(BR^{-1}B^\top\)
with normalization \(8\alpha_B^2\rho_R\), error at most
\(\alpha_B^2\varepsilon_{R^{-1}}=\varepsilon_H\), and
\(2a_B+a_R+1+e(l_d,l_*)\) ancillas.
The target is supported on its first \(n\) coordinates.
Compressing the system register from \(l_*\) to \(l_n\) therefore
adds \(l_*-l_n\) projected qubits, giving \(b_H\) ancillas.

Write
\[
H
=
\ket0\bra0\otimes A
-\ket0\bra1\otimes BR^{-1}B^\top
-\ket1\bra0\otimes Q
-\ket1\bra1\otimes A^\top.
\]
The four type-register matrix units have exact \((1,1)\)
encodings. The Kronecker-product rule gives ancilla counts
\(a_A+1,b_H+1,a_Q+1,a_A+1\), respectively.
A balanced two-level linear combination of these four terms
adds two selection qubits. Its normalization is
\eqref{eq:alpha-H}, its error is at most \(\varepsilon_H\),
and its ancilla count is \eqref{eq:a_H}.
If \(B=0\), the corresponding term and inverse construction
are simply omitted.

We next bound the normalization. Since
\(B=(BR^{-1/2})R^{1/2}\),
\begin{equation}\label{eq:alpha-H-inequality0}
\begin{aligned}
\|B\|_{\max}^2\rho_R
&\leq\|B\|^2\rho_R\\
&\leq\|BR^{-1/2}\|^2\|R\|\rho_R
=\kappa_R\|BR^{-1}B^\top\|
\leq\kappa_R\|H\|.
\end{aligned}
\end{equation}
Each block of \(H\) has norm at most \(\|H\|\). Thus
\begin{equation}\label{eq:alpha-H-inequality}
\alpha_H
\leq(3s+8s^2\kappa_R)\|H\|
\leq11s^2\kappa_R\|H\|.
\end{equation}
Conversely, the triangle inequality gives
\begin{equation}\label{eq:alpha_H-geq-H}
\|H\|
\leq2\|A\|+\|Q\|+\|BR^{-1}B^\top\|
\leq\alpha_H.
\end{equation}
The inverse construction therefore uses
\begin{equation}\label{eq:BE-R-inverse-scaling}
\begin{aligned}
\mathcal O\!\left(
\alpha_R\rho_R
\log\!\left(2+\frac{\alpha_B^2\rho_R}{\varepsilon_H}\right)
\right)
=
\mathcal O\!\left(
s\kappa_R
\log\!\left(2+\frac{s^2\kappa_R\|H\|}{\varepsilon_H}\right)
\right)
\end{aligned}
\end{equation}
queries. Register extension, compression, and the type-register
products give the stated additional gate count.
\end{proof}

\subsection{Block-encoding of the one-step propagator}
\label{app:be-Phi-construction}

\begin{lemma}[Block-encoding of \(\Phi\)]
\label{lem:be-Phi-construction}
Suppose that \(H\ne0\), \(h\|H\|<1/(4\pi)\), and let
\(\chi=\|\mathcal G_h^{-1}\|\). The known upper bound
\(\chi=2\) may instead be used in the construction.
For any target \(\varepsilon_\Phi>0\), choose
\begin{equation}\label{eq:eplison_G-Phi-relation}
\varepsilon_{\mathcal G_h^{-1}}
=
\min\left\{
\frac12,\,
\frac{4c_{\rm inv}\alpha_{\mathcal A}\chi}
{2^{l_p/2}(4c_{\rm inv}\alpha_{\mathcal A}
h\alpha_H\chi+3)}
\varepsilon_\Phi
\right\},
\end{equation}
and
\begin{equation}\label{eq:vareplison-H}
\varepsilon_H
=
\frac{3\varepsilon_{\mathcal G_h^{-1}}}
{32c_{\rm inv}h\alpha_{\mathcal A}\chi^2}.
\end{equation}
Then \(\Phi\) has a block-encoding with parameters
\begin{equation}\label{eq:alpha_Phi-and-a_Phi}
\begin{aligned}
\alpha_\Phi
&=1+2^{l_p/2+3}h\alpha_H\chi\|r\|,\\
a_\Phi
&=a_{\mathcal A}+2a_H+l_p+7+2\tau_{p,n},
\end{aligned}
\end{equation}
and error at most \(\varepsilon_\Phi\).
It uses one call to \(U_r^\dagger\) and
\[
\mathcal O\!\left(s^2\kappa_R\Lambda_\Phi\right),
\qquad
\Lambda_\Phi=
\log\!\left(2+\frac{\sqrt p\,s^2\kappa_R}{\varepsilon_\Phi}\right),
\]
calls to \(U_{\mathcal A},U_H\) and their inverses.
Excluding the implementations of these oracles and \(U_r\),
the additional gate count is
\[
\mathcal O\!\left(
s^2\kappa_R\Lambda_\Phi
\log(2d)\log^2(2n)\log^2(2p)
\right).
\]
\end{lemma}

\begin{proof}
Combine \(U_{\mathcal A}\) and \(U_H\) to encode
\(\mathcal A\otimes hH\). Adding the padded identity gives an
encoding of \(\mathcal G_h\) with
\[
\begin{aligned}
\alpha_{\mathcal G_h}&=1+h\alpha_{\mathcal A}\alpha_H,\\
a_{\mathcal G_h}&=a_{\mathcal A}+a_H+\tau_{p,n}+1,\\
\varepsilon_{\mathcal G_h}&=h\alpha_{\mathcal A}\varepsilon_H.
\end{aligned}
\]
Since \(\alpha_{\mathcal A}\geq\|\mathcal A\|\) and
\(\alpha_H\geq\|H\|\), we have
\(\alpha_{\mathcal G_h}\geq\|\mathcal G_h\|\).
Apply Lemma~\ref{lem:inverse-be-approx} with
\(\kappa_0=2\alpha_{\mathcal G_h}\chi\).
Its input error condition is
\[
\varepsilon_{\mathcal G_h}
\leq
\min\left\{
\frac1{8\chi},
\frac{3\varepsilon_{\mathcal G_h^{-1}}}
{32c_{\rm inv}\chi^2}
\right\}.
\]
Our choices satisfy this condition: the second term is the
chosen error, and it does not exceed the first because
\(\varepsilon_{\mathcal G_h^{-1}}\leq1/2\) and
\(\chi>2/3\).
Thus the inverse encoding has normalization \(8\chi\),
\(a_{\mathcal A}+a_H+\tau_{p,n}+2\) ancillas, and error
\(\varepsilon_{\mathcal G_h^{-1}}\).

The two rectangular factors in
\eqref{eq:rk-propagator-definition} have parameters
\[
\begin{array}{c|c|c|c}
\text{matrix}
&\text{normalization}
&\text{ancillas}
&\text{error}\\
\hline
\rule{0pt}{2.8ex}r^\top\otimes I_{2n}
&\|r\|
&2+\tau_{p,n}
&0\\
\mathbf1_p\otimes hH
&2^{l_p/2}h\alpha_H
&a_H+2+\tau_{p,n}
&2^{l_p/2}h\varepsilon_H
\end{array}
\]
These follow from Appendix~\ref{app:auxiliary-factors} and the
Kronecker-product rule. All three factors act on the common
system register of size
\[
r_{\mathcal G}
=
\lceil\log_2(2pn)\rceil
=
l_p+l_n+1-\tau_{p,n}.
\]
Their product has
\(a_{\mathcal A}+2a_H+6+3\tau_{p,n}\) ancillas.
Its target is supported on the first \(2n\) coordinates.
Compression to \(l_n+1\) system qubits adds
\(l_p-\tau_{p,n}\) projected qubits. The final sum with \(I_{2n}\)
adds one selection qubit. This proves both counts in
\eqref{eq:alpha_Phi-and-a_Phi}.

The product rule gives the following error bound:
\begin{equation}\label{eq:vareplison-Phi}
\begin{aligned}
\varepsilon_{\rm out}
&\leq
2^{l_p/2}h\alpha_H\|r\|
\varepsilon_{\mathcal G_h^{-1}}
+
2^{l_p/2+3}h\chi\|r\|\varepsilon_H\\
&\leq
2^{l_p/2}
\left(h\alpha_H+
\frac{3}{4c_{\rm inv}\alpha_{\mathcal A}\chi}\right)
\varepsilon_{\mathcal G_h^{-1}}
\leq\varepsilon_\Phi.
\end{aligned}
\end{equation}
Here we used \(\|r\|\leq1\) and \eqref{eq:vareplison-H}.
This argument retains the minimum in
\eqref{eq:eplison_G-Phi-relation}, so it applies to every target
accuracy.

The inverse construction uses
\[
\mathcal O\!\left(
(1+h\alpha_{\mathcal A}\alpha_H)\chi
\log\!\left(2+\frac{\chi}{\varepsilon_{\mathcal G_h^{-1}}}\right)
\right)
\]
calls to \(U_{\mathcal G_h}\) and its inverse.
Now \(\alpha_{\mathcal A}=\Theta(1)\), \(\chi=\Theta(1)\), and
\(h\alpha_H=\mathcal O(s^2\kappa_R)\). Moreover,
\[
\varepsilon_{\mathcal G_h^{-1}}^{-1}
=
\mathcal O\!\left(
1+\frac{\sqrt p(1+h\alpha_H)}{\varepsilon_\Phi}
\right).
\]
These estimates give the query bound. The register permutations
and QSVT reflections give the stated additional gate count.
\end{proof}

\subsection{Block-encoding of the nodal matrix and its inverse}
\label{app:be-L-construction}

\begin{theorem}[Block-encoding of \(L\) and \(L^{-1}\)]
\label{thm:be-L-construction}
Under the assumptions of Lemma~\ref{lem:be-Phi-construction},
there is an \((\alpha_L,a_L,\varepsilon_\Phi)\)-block-encoding
of \(L\), where
\[
\alpha_L=\alpha_\Phi+\alpha_G+3,\qquad
a_L=\max\{a_\Phi+1,a_G+3\}+\tau_{m,n}+2.
\]
It uses one call to \(U_\Phi\), one call to \(U_G\), and
\(\mathcal O(\log^2(2m)\log^2(2n))\) additional gates.

Suppose that \(L\) is nonsingular. Let
\(\nu\geq\|L^{-1}\|\) be an available bound and
\(0<\varepsilon_{L^{-1}}\leq\nu\). Choose
\begin{equation}\label{eq:inverse-L-input-accuracy}
\varepsilon_\Phi=\varepsilon_L
=
\min\left\{
\frac1{8\nu},
\frac{3\varepsilon_{L^{-1}}}{32c_{\rm inv}\nu^2}
\right\}.
\end{equation}
Then there is an
\((8\nu,a_L+1,\varepsilon_{L^{-1}})\)-block-encoding of
\(L^{-1}\). Define
\[
\begin{aligned}
\Gamma_L&=s^2\kappa_R+s\|G\|_{\max}+1,\\
\Lambda_L&=\log\!\left(2+\frac{\nu}{\varepsilon_{L^{-1}}}\right),\\
\Lambda_\Phi&=\log\!\left(
2+\frac{\sqrt p\,s^2\kappa_R\nu^2}{\varepsilon_{L^{-1}}}
\right).
\end{aligned}
\]
One implementation uses
\begin{equation}\label{eq:inverse-L-total-query-cost}
Q_{\rm inv}
=
\mathcal O\!\left(
\Gamma_Ls^3\kappa_R^2\nu\Lambda_L\Lambda_\Phi^2
\right)
\end{equation}
queries to the input matrix oracles.
\end{theorem}

\begin{proof}
Decompose \(L=L_1+L_2+L_3+L_4\), where
\[
\begin{aligned}
L_1&=\sum_{k=1}^{m-1}E_{kk}\otimes I_{2n},
&L_2&=-S_m\otimes\Phi,\\
L_3&=E_{00}\otimes B_0,
&L_4&=E_{0,m-1}\otimes B_T.
\end{aligned}
\]
Here the shift \(S_m\) is defined in \eqref{eq:shift}. The matrices \(L_1\) and \(L_3\) are coordinate projections:
\[
L_1=\operatorname{diag}(0_{2n},I_{2n(m-1)}),
\qquad
L_3=\operatorname{diag}(I_n,0_{2nm-n}).
\]
Each therefore has an exact \((1,1)\)-encoding by Lemma~\ref{lem:projection-block-encoding}.
For \(L_2\), the shift construction gives parameters
\((\alpha_\Phi,a_\Phi+1+\tau_{m,n},\varepsilon_\Phi)\).

The terminal block is
\[
B_T=-\ket1\bra0\otimes G
+\operatorname{diag}(0_n,I_n).
\]
Its two terms have exact encodings with parameters
\((\alpha_G,a_G+1)\) and \((1,1)\), respectively.
Their sum has parameters \((\alpha_G+1,a_G+2)\).
Combining it with \(E_{0,m-1}\) gives an exact
\((\alpha_G+1,a_G+3+\tau_{m,n})\)-encoding of \(L_4\).
A balanced sum of the four \(L_i\) yields the stated
\(\alpha_L,a_L,\varepsilon_L\). Only \(L_2\) contributes error.

Since \(\alpha_\Phi\geq\|\Phi\|\) and
\(\alpha_G\geq\|G\|\), this construction also gives
\(\alpha_L\geq\|L\|\).
Lemma~\ref{lem:inverse-be-approx}, with
\(\kappa_0=2\alpha_L\nu\), therefore applies under
\eqref{eq:inverse-L-input-accuracy}.
It gives output normalization \(8\nu\) and uses
\(\mathcal O(\alpha_L\nu\Lambda_L)\) calls to \(U_L\)
and its inverse.

Since \(\varepsilon_{L^{-1}}\leq\nu\), the second term in
\eqref{eq:inverse-L-input-accuracy} is the smaller one.
Thus \(\varepsilon_\Phi=\Theta(\varepsilon_{L^{-1}}/\nu^2)\).
The coefficient bounds in Appendix~\ref{app:gauss-discretization}
give
\[
2^{l_p/2}\|r\|\leq\sqrt{\frac{8\pi}{3}},
\qquad
\alpha_\Phi=\mathcal O(1+h\alpha_H)
=\mathcal O(s^2\kappa_R+1).
\]
Consequently, \(\alpha_L=\mathcal O(\Gamma_L)\).
Also \(\|L^{-1}\|\geq1\): for a right-hand side supported on
the initial-state boundary data, the solution contains that same
initial state. Thus \(\nu\geq1\).

Let \(Q_0=\mathcal O(\Gamma_L\nu\Lambda_L)\) be the outer
inversion query count. Each call to \(U_L\) uses one \(U_\Phi\),
and each \(U_\Phi\) uses
\(\mathcal O(s^2\kappa_R\Lambda_\Phi)\) calls to \(U_H\)
and \(U_{\mathcal A}\).
For the precision \eqref{eq:vareplison-H},
\[
\log\!\left(2+\frac{s^2\kappa_R\|H\|}{\varepsilon_H}\right)
=
\mathcal O(\Lambda_\Phi).
\]
To see this, use
\(\varepsilon_H^{-1}
=\mathcal O(h/\varepsilon_{\mathcal G_h^{-1}})\),
the bound for \(\varepsilon_{\mathcal G_h^{-1}}^{-1}\)
above, and \(h\|H\|=\mathcal O(1)\).
The resulting query counts are
\[
\begin{array}{c|c}
\text{oracle, including its inverse}&\text{number of calls}\\ \hline
U_G,\ U_r
&\mathcal O(\Gamma_L\nu\Lambda_L)\\
U_{\mathcal A},\ U_A,\ U_B,\ U_Q
&\mathcal O(\Gamma_Ls^2\kappa_R\nu\Lambda_L\Lambda_\Phi)\\
U_R
&\mathcal O(\Gamma_Ls^3\kappa_R^2\nu
\Lambda_L\Lambda_\Phi^2).
\end{array}
\]
This proves \eqref{eq:inverse-L-total-query-cost}.

For clarity, the auxiliary oracles also contribute gates.
Let \(Q_{\mathcal A}\) and \(Q_r\) denote their query bounds in
the table, and put
\[
\mathcal L_{\rm reg}
=
\log(2d)\log^2(2n)\log^2(2m)\log^2(2p).
\]
The fully expanded additional gate count is bounded by
\begin{equation}\label{eq:inverse-L-additional-gates}
\mathcal O\!\left(
Q_{\rm inv}\mathcal L_{\rm reg}
+Q_{\mathcal A}p^2+Q_rp
\right).
\end{equation}
The first term includes the nested QSVT procedures, register
permutations, and projections. The other two terms implement
\(U_{\mathcal A}\) and \(U_r\). The ancilla formulas above give
\(a_L=\mathcal O(\log(2n)+\log(2d)+\log(2p))\).
\end{proof}

\subsection{Grid choice and discretization bounds}
\label{app:algorithm-discretization}

Use the grid specified in Section~\ref{sec:main-results}:
\begin{equation}\label{eq:m-scaling}
m=1+\max\{1,\lceil8\pi T\|H\|\rceil\},
\qquad h=\frac{T}{m-1}.
\end{equation}
Then \(h\|H\|\leq1/(8\pi)<1/(4\pi)\). Define the available bound
\begin{equation}\label{eq:algorithm-exact-inverse-bound}
K_{\rm ex}
=
M\sum_{j=0}^{m-1}e^{-\alpha_cjh}
=
M\frac{1-e^{-\alpha_chm}}{1-e^{-\alpha_ch}}.
\end{equation}
By Theorem~\ref{thm:inverse-bound-L},
\(\|L_{\rm ex}^{-1}\|\leq K_{\rm ex}\).

This bound remains uniform also for short horizons.
If \(8\pi T\|H\|\leq1\), then \(m=2\), so \(K_{\rm ex}\leq2M\).
Otherwise \(m-1\leq16\pi T\|H\|\), and hence
\[
\sum_{j=0}^{m-1}e^{-\alpha_cjh}
\leq\frac1{1-e^{-\alpha_ch}}
\leq1+\frac1{\alpha_ch}
\leq1+\frac{16\pi\|H\|}{\alpha_c}.
\]
Thus, in both cases,
\begin{equation}\label{eq:algorithm-inverse-scaling}
K_{\rm ex}
=
\mathcal O\!\left(M(1+\alpha_c^{-1}\|H\|)\right).
\end{equation}

Lemma~\ref{lem:rk-global-error-v2} gives the convenient bound
\begin{equation}\label{eq:Delta_h-ineq}
\Delta_h
\leq\frac{8(p!)^2}{(2p)!(2p+1)!}
\leq\frac{16\sqrt p}{4^p p^{p+1}}
\leq\frac{16}{(4p)^p}.
\end{equation}
Here we used
\(\binom{2p}{p}\geq4^p/(2\sqrt p)\) and
\((2p+1)!\geq p^{p+1}\).
If \(\Delta_hK_{\rm ex}\leq1/2\), the perturbation estimates in
Appendix~\ref{app:approx-conditioning-proof} imply
\begin{equation}\label{eq:algorithm-nodal-error}
\|L^{-1}\|\leq2K_{\rm ex},
\qquad
\|q-q_{\rm ex}\|
\leq2\Delta_hK_{\rm ex}^2\|x^0\|.
\end{equation}
Indeed,
\(q-q_{\rm ex}=L^{-1}(L_{\rm ex}-L)q_{\rm ex}\) and
\(\|q_{\rm ex}\|\leq K_{\rm ex}\|b\|
=K_{\rm ex}\|x^0\|\).
We use the known inversion bound \(\nu=2K_{\rm ex}\) below.

The terminal matrix can also be bounded by the stability
parameters:
\begin{equation}\label{eq:G-M-relationship}
\|G\|_{\max}\leq\|G\|\leq M_P\leq M.
\end{equation}
The second inequality follows from \(P(T)=G\), and the last
from the definition of \(M\). We retain \(\|G\|_{\max}\) explicitly
in \(\Gamma_L\).

\subsection{The optimal-cost identity}
\label{app:optimal-cost-identity-proof-v2}

\begin{lemma}[Optimal-cost identity]
\label{lem:optimal-cost-identity-v2}
The optimal state and costate satisfy
\begin{equation}\label{eq:optimal-cost-initial-costate}
J^*=\frac12(x^0)^\top\lambda(0).
\end{equation}
\end{lemma}

\begin{proof}
The state and costate equations give
\[
\frac{\mathrm d}{\mathrm dt}(x^\top\lambda)
=
(Ax-BR^{-1}B^\top\lambda)^\top\lambda
+x^\top(-Qx-A^\top\lambda)
=
-x^\top Qx-\lambda^\top BR^{-1}B^\top\lambda.
\]
Since \(u^*=-R^{-1}B^\top\lambda\), the last quadratic term is
\((u^*)^\top Ru^*\). Integrating and using
\(\lambda(T)=Gx(T)\), we obtain
\[
\begin{aligned}
J^*
&=\frac12\bigl(x(0)^\top\lambda(0)
-x(T)^\top\lambda(T)\bigr)
+\frac12x(T)^\top Gx(T)\\
&=\frac12(x^0)^\top\lambda(0).
\end{aligned}
\]
\end{proof}

\subsection{Optimal-cost estimation}
\label{app:optimal-cost-estimation-proof-v2}

\begin{theorem}[Optimal-cost estimation]
\label{thm:optimal-cost-estimation-v2}
Under Assumptions~\ref{ass:be} and~\ref{ass:sta&det},
Algorithm~\ref{alg:optimal-cost}, with the parameter choices and
the preliminary check below, returns an estimate satisfying
\[
\Pr\bigl(|J^*_{\rm num}-J^*|\leq\varepsilon\bigr)\geq1-\delta.
\]
Its query complexity is
\[
\widetilde{\mathcal O}\!\left(
\frac{\Gamma_Ls^3\kappa_R^2K_{\rm ex}^2\|x^0\|^2}
{\varepsilon}\log\frac2\delta
\right).
\]
The additional gate count is bounded by this expression times
\(\log(2d)\log^2(2n)\log^2(2m)\), up to the suppressed
logarithmic factors.
\end{theorem}

\begin{proof}
If \(x^0=0\), return zero. More generally,
\[
|J^*|\leq\frac12\|x^0\|\|\lambda(0)\|
\leq\frac12K_{\rm ex}\|x^0\|^2.
\]
Thus zero is already an admissible output when
\(\varepsilon\geq K_{\rm ex}\|x^0\|^2/2\).
Assume henceforth that
\(0<\varepsilon<K_{\rm ex}\|x^0\|^2/2\).

Choose \(p\) so that
\begin{equation}\label{eq:p-scaling-equation1}
(4p)^p\geq
\frac{48K_{\rm ex}^2\|x^0\|^2}{\varepsilon},
\end{equation}
and
\begin{equation}\label{eq:p-scaling-equation2}
(4p)^p\geq32K_{\rm ex}.
\end{equation}
The second condition gives \(\Delta_hK_{\rm ex}\leq1/2\).
Using the first condition and \eqref{eq:algorithm-nodal-error},
\[
|J_h-J^*|
\leq\frac{\|x^0\|}{2}\|q-q_{\rm ex}\|
\leq\Delta_hK_{\rm ex}^2\|x^0\|^2
\leq\frac{\varepsilon}{3}.
\]
With
\[
K_J=\max\left\{
32K_{\rm ex},
\frac{48K_{\rm ex}^2\|x^0\|^2}{\varepsilon}
\right\},
\]
the least sufficient stage number satisfies
\(p=\mathcal O(\log K_J/\log\log K_J)\).
In particular, it grows at most logarithmically in the inverse
required accuracy.

Let \(c=S_{\lambda,0}^\top x^0\). Since
\(\|b\|=\|c\|=\|x^0\|\), the cost-transfer operator in
\eqref{eq:cost-transfer-operator-v2} gives
\[
J_h
=
\frac12b^\top E_0L^{-1}b
=
\frac12c^\top L^{-1}b
=
\frac{\|x^0\|^2}{2}\bra cL^{-1}\ket b.
\]
The states \(\ket b\) and \(\ket c\) each use one call to
\(U_{x^0}\). The former is obtained by adding zero registers.
For the latter, shift the populated indices \(j\in[n]\) to
\(n+j\), using reversible addition of the known constant \(n\).
This works also when \(n\) is not a power of two and uses
\(\mathcal O(\log^2(2nm))\) additional gates.

Apply Theorem~\ref{thm:be-L-construction} with \(\nu=2K_{\rm ex}\)
and
\[
\varepsilon_{L^{-1}}
=
\min\left\{
\nu,\frac{2\varepsilon}{3\|x^0\|^2}
\right\},
\qquad
\alpha_{L^{-1}}=8\nu.
\]
Write
\[
\widetilde L^{-1}
=
\alpha_{L^{-1}}\mathcal B(U_{L^{-1}}),
\qquad
\widetilde J_h
=
\frac{\|x^0\|^2}{2}
\operatorname{Re}\bra c\widetilde L^{-1}\ket b.
\]
The encoding error satisfies
\[
|\widetilde J_h-J_h|
\leq\frac{\|x^0\|^2}{2}\varepsilon_{L^{-1}}
\leq\frac{\varepsilon}{3}.
\]

Since \(U_{L^{-1}}\) is an exact block-encoding of
\(\mathcal B(U_{L^{-1}})\) with normalization factor \(1\),
we use the estimation procedure in
Lemma~\ref{lem:Hadamard-test} to obtain an estimate
\(\widehat\mu_J\) of
\[
\mu_J=
\operatorname{Re}\bra c\mathcal B(U_{L^{-1}})\ket b
\]
with additive error at most
\[
\eta_J=
\frac{2\varepsilon}
{3\alpha_{L^{-1}}\|x^0\|^2}
\]
and failure probability at most \(\delta\).
We then return
\[
J^*_{\rm num}
=
\frac{\alpha_{L^{-1}}\|x^0\|^2}{2}\,
\widehat\mu_J.
\]
On the success event,
\[
|J^*_{\rm num}-\widetilde J_h|
\leq\frac{\alpha_{L^{-1}}\|x^0\|^2}{2}\eta_J
=\frac{\varepsilon}{3}.
\]
Adding the three errors proves the output guarantee.

The estimation procedure requires
\[
\mathcal O\!\left(
\left(1+\frac{\alpha_{L^{-1}}\|x^0\|^2}{\varepsilon}\right)
\log\frac2\delta
\right)
\]
calls to the inverse encoding and state-preparation oracles.
In the nontrivial regime considered here,
\(\alpha_{L^{-1}}\|x^0\|^2/\varepsilon=\Omega(1)\), multiplying by
\eqref{eq:inverse-L-total-query-cost} gives
\[
\widetilde{\mathcal O}\!\left(
\frac{\Gamma_Ls^3\kappa_R^2\nu^2\|x^0\|^2}
{\varepsilon}\log\frac2\delta
\right).
\]
Using \(\nu=2K_{\rm ex}\) proves the stated query bound.
Equation~\eqref{eq:inverse-L-additional-gates}, together with the
stage-number bound, gives the gate count: the factors polynomial
in \(p\) are logarithmic in the precision and problem parameters.
Finally, \eqref{eq:algorithm-inverse-scaling} gives
Theorem~\ref{thm:main-cost}.
\end{proof}

\subsection{Optimal-control estimation}
\label{app:optimal-control-estimation-proof-v2}

\begin{theorem}[Optimal-control estimation]
\label{thm:optimal-control-estimation-v2}
Under Assumptions~\ref{ass:be} and~\ref{ass:sta&det},
Algorithm~\ref{alg:optimal-control}, with the parameter choices
and the preliminary check below, returns a vector satisfying
\[
\Pr\bigl(\|u_{\rm num}(t_k)-u^*(t_k)\|
\leq\varepsilon\bigr)\geq1-\delta.
\]
Its query complexity is
\[
\widetilde{\mathcal O}\!\left(
\frac{
d^{3/2}\Gamma_Ls^4\kappa_R^2
K_{\rm ex}^2\|B\|_{\max}\rho_R\|x^0\|
}{\varepsilon}\log\frac{2d}{\delta}
\right).
\]
The additional gate count has the same factor
\(\log(2d)\log^2(2n)\log^2(2m)\) as in the cost estimate,
up to suppressed logarithmic factors.
\end{theorem}

\begin{proof}
If \(x^0=0\) or \(B=0\), return zero.
Otherwise, put \(\beta=\alpha_B\rho_R\), so that
\(\|R^{-1}B^\top\|\leq\beta\). The exact control is
\begin{equation}\label{eq:OC-}
u^*(t_k)=-R^{-1}B^\top S_{\lambda,k}q_{\rm ex},
\end{equation}
and
\(\|u^*(t_k)\|\leq\beta K_{\rm ex}\|x^0\|\).
Hence zero is admissible if
\(\varepsilon\geq\beta K_{\rm ex}\|x^0\|\).
Assume below that
\(0<\varepsilon<\beta K_{\rm ex}\|x^0\|\).

For the discretization error, choose \(p\) so that
\begin{equation}\label{eq:control-p-condition-accuracy}
(4p)^p
\geq\frac{96\beta K_{\rm ex}^2\|x^0\|}{\varepsilon},
\end{equation}
and
\begin{equation}\label{eq:control-p-condition-stability}
(4p)^p\geq32K_{\rm ex}.
\end{equation}
Then \eqref{eq:Delta_h-ineq} and \eqref{eq:algorithm-nodal-error} give
\begin{equation}\label{eq:control-discretization-error}
\begin{aligned}
\|u_h(t_k)-u^*(t_k)\|
\leq\beta\|q-q_{\rm ex}\|\leq2\beta\Delta_hK_{\rm ex}^2\|x^0\|
\leq\frac{\varepsilon}{3}.
\end{aligned}
\end{equation}
With
\[
K_u=\max\left\{
32K_{\rm ex},
\frac{96\beta K_{\rm ex}^2\|x^0\|}{\varepsilon}
\right\},
\]
the least sufficient stage number satisfies
\begin{equation}\label{eq:control-p-scaling}
p=\mathcal O\!\left(\frac{\log K_u}{\log\log K_u}\right).
\end{equation}

We next construct the operator used in
Algorithm~\ref{alg:optimal-control},
\[
C_k=R^{-1}B^\top S_{\lambda,k}L^{-1},
\qquad
u_h(t_k)=-C_kb.
\]
Use \(\nu=2K_{\rm ex}\) and
\(\alpha_{L^{-1}}=8\nu\). A separate inverse encoding of \(R\),
for this final product, has normalization \(8\rho_R\) and
error \(\varepsilon_{R^{-1}}^{(u)}\).
Together with the exact encodings of \(B^\top\) and
\(S_{\lambda,k}\), the product rule gives
\begin{equation}\label{eq:control-encoding-parameters}
\begin{aligned}
\alpha_{C_k}
&=8\alpha_B\rho_R\alpha_{L^{-1}}
=64\beta\nu,\\
\varepsilon_{C_k}
&\leq
\alpha_B\alpha_{L^{-1}}\varepsilon_{R^{-1}}^{(u)}
+8\alpha_B\rho_R\varepsilon_{L^{-1}}.
\end{aligned}
\end{equation}
The inverse of \(R\) used inside \(U_H\) keeps the precision
specified in Appendix~\ref{app:be-H-construction}; the readout
precision is chosen separately.

To make the rectangular product precise, set
\[
r_L=\lceil\log_2(2nm)\rceil,\qquad
r_{\rm out}=\max\{r_L,l_d\}.
\]
Extend all four encodings to the common \(r_{\rm out}\)-qubit
system register. Lemma~\ref{lem:register-extension-compression}
gives the sufficient ancilla count
\[
\begin{aligned}
a_{C_k}
={}&a_R+1+e(l_d,r_{\rm out})
+a_B+e(l_*,r_{\rm out})\\
&+3+\tau_{m,n}+e(r_L,r_{\rm out})
+a_{L^{-1}}+e(r_L,r_{\rm out}).
\end{aligned}
\]
The zero extensions do not change the normalization or error.
They require only logarithmically many additional gates.

Choose the internal errors so that
\begin{equation}\label{eq:control-internal-errors}
\begin{aligned}
\varepsilon_{R^{-1}}^{(u)}
&=
\min\left\{
\rho_R,\,
\frac{\varepsilon}
{6\alpha_B\alpha_{L^{-1}}\|x^0\|}
\right\},\\
\varepsilon_{L^{-1}}
&=
\min\left\{
\nu,\,
\frac{\varepsilon}{48\beta\|x^0\|}
\right\}.
\end{aligned}
\end{equation}
Each term in the error bound
\eqref{eq:control-encoding-parameters} is then at most
\(\varepsilon/(6\|x^0\|)\).
Consequently,
\(\varepsilon_{C_k}\leq\varepsilon/(3\|x^0\|)\).
Let \(\widetilde C_k\) be the leading \(d\times2nm\) block of
\(\alpha_{C_k}\mathcal B(U_{C_k})\), and set
\[
\widetilde u_h=-\operatorname{Re}(\widetilde C_kb).
\]
we obtain the vector bound
\[
\|\widetilde u_h-u_h(t_k)\|
\leq\|\widetilde C_k-C_k\|\,\|b\|
\leq\frac{\varepsilon}{3}.
\]

Pad \(\ket b\) and the basis states \(\ket j\) to the common
\(r_{\rm out}\)-qubit system register.
For \(j\in[d]\), estimate the real matrix element
\[
\mu_{k,j}
=
\operatorname{Re}
\bra j\mathcal B(U_{C_k})\ket b
\]
to accuracy
\[
\eta_u=
\frac{\varepsilon}
{3\sqrt d\,\alpha_{C_k}\|x^0\|},
\]
with failure probability \(\delta/d\).
Set
\[
(u_{\rm num}(t_k))_j
=
-\alpha_{C_k}\|x^0\|\,\widehat\mu_{k,j}.
\]
By the union bound, all \(d\) estimates satisfy their error bounds
with probability at least \(1-\delta\). On that event,
\[
\begin{aligned}
\|u_{\rm num}(t_k)-\widetilde u_h\|^2
&\leq
\alpha_{C_k}^2\|x^0\|^2
\sum_{j=0}^{d-1}
|\widehat\mu_{k,j}-\mu_{k,j}|^2\\
&\leq
d\alpha_{C_k}^2\|x^0\|^2\eta_u^2
=\frac{\varepsilon^2}{9}.
\end{aligned}
\]
Combining this bound with the encoding and discretization errors
proves the output guarantee.

The total number of readout calls is
\begin{equation}\label{eq:control-total-readout-calls}
\begin{aligned}
\mathcal O\!\left(
d\left(1+
\frac{\sqrt d\,\alpha_{C_k}\|x^0\|}{\varepsilon}\right)
\log\frac{2d}{\delta}
\right)=
\mathcal O\!\left(
\frac{d^{3/2}\beta\nu\|x^0\|}{\varepsilon}
\log\frac{2d}{\delta}
\right).
\end{aligned}
\end{equation}
The equality uses the nontrivial precision regime and
\(\nu=2K_{\rm ex}\).
One call to \(U_{C_k}\) uses one \(U_{L^{-1}}\), one \(U_B^\dagger\),
and the readout encoding of \(R^{-1}\).
The last requires
\[
\mathcal O\!\left(
s\kappa_R\log\!\left(
2+\frac{\rho_R}{\varepsilon_{R^{-1}}^{(u)}}\right)
\right)
\]
queries. By \eqref{eq:control-internal-errors}, this is dominated
by \eqref{eq:inverse-L-total-query-cost}.
Therefore the total input query bound is
\begin{equation}\label{eq:control-query-before-substitution}
\widetilde{\mathcal O}\!\left(
\frac{
d^{3/2}\Gamma_Ls^4\kappa_R^2
\|B\|_{\max}\rho_R\nu^2\|x^0\|
}{\varepsilon}
\log\frac{2d}{\delta}
\right).
\end{equation}
Substituting \(\nu=2K_{\rm ex}\) proves the stated bound.
The selector, basis-state preparation, and register extensions
fit within the gate count in
\eqref{eq:inverse-L-additional-gates}.
Since \(p\) satisfies \eqref{eq:control-p-scaling}, its polynomial
cost is absorbed into the logarithmic factors.
Finally, \eqref{eq:algorithm-inverse-scaling} gives
Theorem~\ref{thm:main-control}.
\end{proof}

\section{Dimension-uniform analysis of the mass--spring chain}
\label{app:Proofs for the benchmark}

We prove Corollary~\ref{cor:main-mass-spring} for
Benchmark~\ref{bench:mass_spring_chain}. Throughout this appendix,
\(n=2n_m\), and the positive parameters
\(m_-,m_+,\delta,k_0,k_1,\mu,\gamma\) are fixed independently of
\(n_m\) and \(T\). Set
\[
k_+=k_0+4k_1,\qquad
q_-=\min\{m_+^{-1},k_0\},\qquad
q_+=\max\{m_-^{-1},k_+\}.
\]
The eigenvalues of
\(k_1\operatorname{tridiag}(-1,2,-1)\in\mathbb R^{n_m\times n_m}\) are
\[
4k_1\sin^2\!\left(\frac{j\pi}{2(n_m+1)}\right),
\qquad j=1,\ldots,n_m.
\]
Hence
\begin{equation}\label{eq:chain-coefficient-bounds}
\begin{gathered}
m_-I_{n_m}\preceq\mathcal M\preceq m_+I_{n_m},\qquad
k_0I_{n_m}\preceq\mathcal K\preceq k_+I_{n_m},\\
q_-I_n\preceq Q\preceq q_+I_n,\qquad
G=\gamma Q,\qquad
\|B\|=\|\mathcal F\|,\qquad \|R^{-1}\|=\mu^{-1}.
\end{gathered}
\end{equation}
We write \(D=BR^{-1}B^\top\), so \(\|D\|=\mu^{-1}\|\mathcal F\|^2\).

\subsection{Uniform open-loop stability}
\label{app:open-loop-stable}

\begin{lemma}[Open-loop stability]
\label{lem:open-loop-stable}
Define
\[
\eta=\frac12\min\left\{
\frac{\delta}{m_+},\sqrt{\frac{k_0}{m_+}}\right\},
\qquad
c_V=1+\eta\sqrt{\frac{m_+}{k_0}}
+\frac{\eta\delta m_+}{m_-k_0}.
\]
Then
\[
\|e^{At}\|\leq C_0e^{-\alpha_0t},\qquad t\geq0,
\qquad
C_0=\left(\frac{2q_+c_V}{q_-}\right)^{1/2},\qquad
\alpha_0=\frac{\eta}{c_V}.
\]
In particular, \(C_0\) and \(\alpha_0^{-1}\) are independent
of \(n_m\).
\end{lemma}

\begin{proof}
For the uncontrolled system, introduce
\[
v=\mathcal M^{-1/2}\rho,\qquad
z=\mathcal M^{1/2}\omega,\qquad
S=\mathcal M^{-1/2}\mathcal K\mathcal M^{-1/2}.
\]
Equation~\eqref{eq:ode-for-oscillator} then becomes
\[
\dot z=v,\qquad
\dot v=-\delta\mathcal M^{-1}v-Sz.
\]
Since \(\mathcal K\succeq k_0I_{n_m}\) and
\(\mathcal M\preceq m_+I_{n_m}\), we have
\[
S\succeq k_0\mathcal M^{-1}
\succeq \frac{k_0}{m_+}I_{n_m}.
\]
Define
\[
E(v,z)=\|v\|^2+z^\top Sz,
\qquad
V(v,z)=E(v,z)+2\eta z^\top v
+\eta\delta z^\top\mathcal M^{-1}z.
\]
By Young's inequality,
\[
\begin{aligned}
2|z^\top v|
\leq 2\|z\|\|v\|
\leq
\sqrt{\frac{m_+}{k_0}}\|v\|^2
+\sqrt{\frac{k_0}{m_+}}\|z\|^2
\leq
\sqrt{\frac{m_+}{k_0}}E(v,z).
\end{aligned}
\]
Moreover, \(S\succeq k_0\mathcal M^{-1}\) implies
\[
z^\top\mathcal M^{-1}z
\leq \frac{1}{k_0}z^\top Sz
\leq \frac{1}{k_0}E(v,z).
\]
Combining these inequalities gives
\begin{equation}\label{eq:inequality-for-V1}
V(v,z)
\leq
\left(
1+\eta\sqrt{\frac{m_+}{k_0}}
+\frac{\eta\delta}{k_0}
\right)E(v,z)
=c_VE(v,z).
\end{equation}
For the lower bound, the last term in \(V\) is nonnegative.
Since \(\eta\sqrt{m_+/k_0}\leq 1/2\), we obtain
\begin{equation}\label{eq:inequality-for-V2}
V(v,z)
\geq E(v,z)-2\eta|z^\top v|
\geq \frac12E(v,z).
\end{equation}

Differentiating \(V\) along solutions and cancelling the
mixed terms gives
\[
\begin{aligned}
\dot V
&=-2\delta v^\top\mathcal M^{-1}v
  +2\eta\|v\|^2-2\eta z^\top Sz\\
&=-2v^\top(\delta\mathcal M^{-1}-\eta I_{n_m})v
  -2\eta z^\top Sz.
\end{aligned}
\]
The condition \(\eta\leq\delta/(2m_+)\) ensures that
\[
\delta\mathcal M^{-1}-\eta I_{n_m}
\succeq
\left(\frac{\delta}{m_+}-\eta\right)I_{n_m}
\succeq \eta I_{n_m}.
\]
Consequently, using~\eqref{eq:inequality-for-V1},
\[
\dot V
\leq -2\eta E(v,z)
\leq -\frac{2\eta}{c_V}V(v,z)
=-2\alpha_0V(v,z).
\]
Writing \(V(t)=V(v(t),z(t))\), Gr\"onwall's inequality yields
\begin{equation}\label{eq:inequality-for-V3}
V(t)\leq e^{-2\alpha_0t}V(0),
\qquad t\geq 0.
\end{equation}

In the original coordinates,
\[
E(v,z)
=\rho^\top\mathcal M^{-1}\rho
+\omega^\top\mathcal K\omega
=x^\top Qx.
\]
Thus, the bounds on \(Q\) give
\begin{equation}\label{eq:x-inequality}
q_-\|x\|^2
\leq E(v,z)
\leq q_+\|x\|^2.
\end{equation}
Combining
\eqref{eq:inequality-for-V1}--\eqref{eq:x-inequality},
we obtain
\[
\|x(t)\|^2
\leq \frac{2}{q_-}V(t)
\leq \frac{2c_Vq_+}{q_-}
e^{-2\alpha_0t}\|x(0)\|^2.
\]
Since \(x(t)=e^{tA}x(0)\), taking square roots and then
the supremum over unit initial states gives
\[
\|e^{tA}\|
\leq
\sqrt{\frac{2c_Vq_+}{q_-}}\,
e^{-\alpha_0t},
\qquad t\geq 0.
\]
This proves the claim.
\end{proof}

Since \(A\) is Hurwitz, both pairs in
Assumption~\ref{ass:sta&det} satisfy the required conditions.
No controllability assumption on the endpoint actuators is needed
for this conclusion.

\subsection{Bounds for the algebraic Riccati solution}
\label{app:are-bound}

\begin{lemma}[Algebraic Riccati bounds]
\label{lem:are-bound}
The stabilizing solution \(P_+\) of
\begin{equation}\label{eq:are-P-+}
A^\top P_++P_+A-P_+DP_++Q=0
\end{equation}
satisfies
\[
\underline p I_n\preceq P_+\preceq\bar p I_n,
\]
where
\[
\bar p=\frac{q_+C_0^2}{2\alpha_0},\qquad
\alpha_1=\frac{1+\delta}{m_-}+k_++\frac{\bar p}{\mu}\|\mathcal F\|^2,
\qquad
\underline p=\frac{q_-}{2\alpha_1}.
\]
These constants are positive and independent of \(n_m\).
\end{lemma}

\begin{proof}
Assumption~\ref{ass:sta&det} gives a unique symmetric
positive semidefinite stabilizing solution; see
Appendix~\ref{app:horizon-uniform-bounds}.
Write \(A_+=A-DP_+\), which is Hurwitz.

For the upper bound, take \(x(t)=e^{At}x(0)\). The Riccati
equation gives
\[
\frac{\mathrm d}{\mathrm dt}(x^\top P_+x)
=x^\top P_+DP_+x-x^\top Qx.
\]
Since \(P_+DP_+\succeq0\), integration over \([0,\tau]\) yields
\begin{equation}\label{eq:x-integrating-inequality}
x(0)^\top P_+x(0)
\leq x(\tau)^\top P_+x(\tau)
+\int_0^\tau x(t)^\top Qx(t)\,\mathrm dt.
\end{equation}
The first term tends to zero as
\(\tau\to\infty\), by Lemma~\ref{lem:open-loop-stable}.
Consequently,
\[
x(0)^\top P_+x(0)
\leq q_+C_0^2\int_0^\infty e^{-2\alpha_0t}\,\mathrm dt
\,\|x(0)\|^2
=\bar p\|x(0)\|^2.
\]
Thus \(P_+\preceq\bar p I_n\).

For the lower bound, rewrite \eqref{eq:are-P-+} as
\begin{equation}\label{eq:ARE-for-P+2}
A_+^\top P_++P_+A_+=-Q-P_+DP_+.
\end{equation}
Since \(A_+\) is Hurwitz, by Lemma 3.18 in \cite{ZhouDoyleGlover1996} the solution to \eqref{eq:ARE-for-P+2} can be written as
\[
P_+=\int_0^\infty
e^{A_+^\top t}(Q+P_+DP_+)e^{A_+t}\,\mathrm dt.
\]
The block form of \(A\) in \eqref{eq:ode-for-oscillator} and the upper bound just proved imply
\[
\|A_+\|
\leq\|A\|+\|D\|\|P_+\|
\leq\frac{1+\delta}{m_-}+k_++\frac{\bar p}{\mu} \|\mathcal F\|^2
=\alpha_1.
\]
For every \(w\in\mathbb R^n\),
\begin{equation}\label{eq:chain-semigroup-lower-bound}
\|e^{A_+t}w\|
\geq\|e^{-A_+t}\|^{-1}\|w\|
\geq e^{-\alpha_1t}\|w\|.
\end{equation}
Using \(Q\succeq q_-I_n\) and \(P_+DP_+\succeq0\), we obtain
\[
w^\top P_+w
\geq q_-\int_0^\infty\|e^{A_+t}w\|^2\,\mathrm dt
\geq\frac{q_-}{2\alpha_1}\|w\|^2.
\]
This proves the lower bound and, in particular, \(P_+\succ0\).
\end{proof}

\subsection{Auxiliary Lyapunov bounds}
\label{app:lyapunov-bound}

\begin{lemma}[Lyapunov bounds]
\label{lem:lyapunov-bound}
Let \(W_1\) and \(W_2\) be the solutions of
\[
A_+^\top W_1+W_1A_+=-I_n,
\qquad
A_+W_2+W_2A_+^\top=-D.
\]
Then
\[
\frac1{2\alpha_1}I_n
\preceq W_1
\preceq\frac{\bar p^2}{\underline p q_-}I_n,
\qquad
0\preceq W_2,
\qquad
\|W_2\|\leq\frac{\bar p^2 \|\mathcal F\|^2}{\mu\underline p q_-}.
\]
\end{lemma}

\begin{proof}
For \(\xi(t)=e^{A_+t}\xi(0)\),
equation~\eqref{eq:ARE-for-P+2} gives
\[
\frac{\mathrm d}{\mathrm dt}(\xi^\top P_+\xi)
=-\xi^\top(Q+P_+DP_+)\xi
\leq -q_-\|\xi\|^2
\leq -\frac{q_-}{\bar p}\,\xi^\top P_+\xi.
\]
Gr\"onwall's inequality and
\(\underline p I_n\preceq P_+\preceq\bar p I_n\) imply
\[
\|\xi(t)\|^2 
\leq \frac{1}{\underline p} \xi(t)^\top P_+\xi(t) \leq \frac{1}{\underline p} e^{-q_- t/\bar p}\xi(0)^\top P_+\xi(0)
\leq \frac{\bar p}{\underline p}e^{-q_- t/\bar p}\|\xi(0)\|^2,
\]
Taking square roots and the supremum over unit initial states
yields
\begin{equation}\label{eq:chain-algebraic-closed-loop-bound}
\|e^{A_+t}\|
\leq \sqrt{\frac{\bar p}{\underline p}}\,
e^{-q_-t/(2\bar p)},
\qquad t\geq0.
\end{equation}
In particular, \(A_+\) is Hurwitz. The Lyapunov equations
therefore have unique solutions with the integral
representations~\cite[Lemma 3.18]{ZhouDoyleGlover1996}
\[
W_1=\int_0^\infty
e^{A_+^\top t}e^{A_+t}\,\mathrm dt,
\qquad
W_2=\int_0^\infty
e^{A_+t}De^{A_+^\top t}\,\mathrm dt.
\]
Since \(D\succeq0\), the second representation gives
\(W_2\succeq0\).
Using~\eqref{eq:chain-algebraic-closed-loop-bound} and
\(\|D\|=\mu^{-1}\|\mathcal F\|^2\), we obtain
\[
\begin{aligned}
W_1
&\preceq
\frac{\bar p}{\underline p}
\int_0^\infty e^{-q_-t/\bar p}\,\mathrm dt\,I_n
=\frac{\bar p^2}{\underline p q_-}I_n,\\
\|W_2\|
&\leq
\frac{\bar p\|D\|}{\underline p}
\int_0^\infty e^{-q_-t/\bar p}\,\mathrm dt
=\frac{\bar p^2\|\mathcal F\|^2}
{\mu\underline p q_-}.
\end{aligned}
\]
Finally, for every \(w\in\mathbb R^n\),
equation~\eqref{eq:chain-semigroup-lower-bound} gives
\[
w^\top W_1w
=\int_0^\infty\|e^{A_+t}w\|^2\,\mathrm dt
\geq
\int_0^\infty e^{-2\alpha_1t}\|w\|^2\,\mathrm dt
=\frac{1}{2\alpha_1}\|w\|^2.
\]
This proves the lower bound on \(W_1\).
\end{proof}

\subsection{The terminal correction}
\label{app:chain-terminal-correction}

The constants in Appendix~\ref{app:horizon-uniform-bounds}
also involve
\[
\widetilde G=(G-P_+)(I_n+W_2(G-P_+))^{-1}.
\]
We next derive a bound on \(\|\widetilde G\|\) that is
independent of \(n_m\). Define
\begin{equation}\label{eq:chain-terminal-gap}
\theta=\frac{q_-\underline p}{2\alpha_1\bar p^2}>0.
\end{equation}

\begin{lemma}[Terminal-correction bound]
\label{lem:chain-terminal-correction}
The matrix \(I_n+W_2(G-P_+)\) is invertible, and
\[
\|G-P_+\|\leq\gamma q_++\bar p,
\qquad
\|\widetilde G\|
\leq\max\left\{\frac{\bar p}{\theta},\gamma q_+\right\}.
\]
\end{lemma}

\begin{proof}
The first bound follows from \(G=\gamma Q\) and
Lemma~\ref{lem:are-bound}. To control the inverse, multiply
\eqref{eq:ARE-for-P+2} on both sides by \(P_+^{-1}\):
\[
A_+P_+^{-1}+P_+^{-1}A_+^\top
=-D-P_+^{-1}QP_+^{-1}.
\]
Subtracting the equation for \(W_2\) gives 
\[
A_+(P_+^{-1}-W_2)+(P_+^{-1}-W_2)A_+^\top=-P_+^{-1}QP_+^{-1},
\] 
therefore
\[
P_+^{-1}-W_2
=\int_0^\infty
e^{A_+t}P_+^{-1}QP_+^{-1}e^{A_+^\top t}\,\mathrm dt.
\]
Since
\(P_+^{-1}QP_+^{-1}\succeq(q_-/\bar p^2)I_n\),
the argument in \eqref{eq:chain-semigroup-lower-bound}, applied
to \(A_+^\top\), yields
\[
\begin{aligned}
w^\top(P_+^{-1}-W_2)w
\geq\frac{q_-}{\bar p^2}
\int_0^\infty\|e^{A_+^\top t}w\|^2\,\mathrm dt
\geq\frac{q_-}{2\alpha_1\bar p^2}\|w\|^2.
\end{aligned}
\]
Consequently,
\[
I_n-P_+^{1/2}W_2P_+^{1/2}
=P_+^{1/2}(P_+^{-1}-W_2)P_+^{1/2}
\succeq\theta I_n.
\]
The matrices \(P_+^{1/2}W_2P_+^{1/2}\) and
\(W_2^{1/2}P_+W_2^{1/2}\) have the same eigenvalues:
they are \(C^\top C\) and \(CC^\top\), respectively, for
\(C=W_2^{1/2}P_+^{1/2}\). Hence
\begin{equation}\label{eq:P_+-in-set}
I_n-W_2^{1/2}P_+W_2^{1/2}\succeq\theta I_n.
\end{equation}
Since \(G\succeq0\),
\begin{equation}\label{eq:G-P_+-in-set}
I_n+W_2^{1/2}(G-P_+)W_2^{1/2}
\succeq I_n-W_2^{1/2}P_+W_2^{1/2}
\succeq\theta I_n.
\end{equation}
Sylvester's determinant identity now gives
\[
\det(I_n+W_2(G-P_+))
=\det(I_n+W_2^{1/2}(G-P_+)W_2^{1/2})>0.
\]
This argument also applies when \(W_2\) is singular. Therefore \(I_n+W_2(G-P_+)\) is invertible.

To bound \(\widetilde G\), consider
\[
f(X)=X(I_n+W_2X)^{-1}
\]
on the convex set
\[
\mathcal D_W=
\{X=X^\top:I_n+W_2^{1/2}XW_2^{1/2}\succ0\}.
\]
Similarly, Sylvester's determinant identity gives
\[
\det(I_n+W_2X)
=\det(I_n+W_2^{1/2}XW_2^{1/2})>0, \qquad X\in\mathcal D_W
\]
which shows that \(f\) is well-defined on this set. Moreover,
\(f(X)=(I_n+XW_2)^{-1}X\), so \(f(X)\) is symmetric.
For \(X,Y\in\mathcal D_W\) with \(X\preceq Y\), put
\(X_s=X+s(Y-X)\). Differentiation gives
\[
\begin{aligned}
\frac{\mathrm d}{\mathrm ds}f(X_s)
&=(Y-X)(I_n+W_2X_s)^{-1}-X_s(I_n+W_2X_s)^{-1}W_2
(Y-X)(I_n+W_2X_s)^{-1}\\
&=(I_n-X_s(I_n+W_2X_s)^{-1}W_2)(Y-X)(I_n+W_2X_s)^{-1}\\
&=(I_n+X_sW_2)^{-1}(Y-X)(I_n+W_2X_s)^{-1}
\succeq0.
\end{aligned}
\]
The last equality comes from Woodbury matrix identity. In the last line, the outer factors are transposes of each other.
Integrating over \(0\leq s\leq1\) proves
\(f(X)\preceq f(Y)\).

Equations~\eqref{eq:P_+-in-set}--\eqref{eq:G-P_+-in-set}
show that \(-P_+,G-P_+,G\in\mathcal D_W\). Since
\(-P_+\preceq G-P_+\preceq G\),
\[
f(-P_+)\preceq\widetilde G\preceq f(G).
\]
The two endpoints satisfy
\[
\begin{aligned}
f(G)
&=G^{1/2}G^{1/2}(I_n+W_2G)^{-1}G^{-1/2}G^{1/2}=G^{1/2}(I_n+G^{1/2}W_2G^{1/2})^{-1}G^{1/2},\\
f(-P_+)
&=-P_+^{1/2}P_+^{1/2}(I_n-W_2P_+)^{-1}P_+^{-1/2}P_+^{1/2}=-P_+^{1/2}(I_n-P_+^{1/2}W_2P_+^{1/2})^{-1}P_+^{1/2}.
\end{aligned}
\]
Therefore
\[
0\preceq f(G)\preceq G\preceq\gamma q_+I_n,
\qquad
-\frac{\bar p}{\theta}I_n\preceq f(-P_+)\preceq0.
\]
It follows that
\(-\bar p\theta^{-1}I_n\preceq\widetilde G
\preceq\gamma q_+I_n\), which proves the norm bound.
\end{proof}

\subsection{Uniform finite-horizon bounds}
\label{app:Mc-alphac-Mp-bound}

\begin{theorem}[Dimension-uniform closed-loop and Riccati bounds]
\label{thm:Mc-alphac-Mp-bound}
For Benchmark~\ref{bench:mass_spring_chain}, the constants chosen
in Remark~\ref{re:ad-M_c-alpha_c} and
Proposition~\ref{prop:horizon-uniform-riccati-bound} satisfy
\[
\alpha_*:=\frac{\underline p q_-}{2\bar p^2}
\leq\alpha_c\leq\alpha_1.
\]
Let
\[
C_\star=\max\left\{
\frac{\bar p}{\theta},\gamma q_++\bar p\right\}.
\]
Then
\begin{equation}\label{eq:chain-uniform-stability-constants}
\begin{aligned}
M_c&\leq\overline M_c
:=\theta^{-1/2}
+\frac{\bar p^2\|\mathcal F\|^2C_\star}{\mu\underline p q_-\theta},\\
M_P&\leq\overline M_P
:=\bar p+\frac{C_\star}{\theta}.
\end{aligned}
\end{equation}
In particular, for every \(n_m\geq2\) and \(T>0\),
\[
\|\Psi(t,s)\|\leq\overline M_c e^{-\alpha_*(t-s)},
\qquad 0\leq s\leq t\leq T,
\qquad
\sup_{t\in[0,T]}\|P(t)\|\leq\overline M_P,
\]
with the same constants \(\overline M_c,\alpha_*,\overline M_P\).
\end{theorem}

\begin{proof}
Lemma~\ref{lem:lyapunov-bound} gives
\[
\kappa(W_1)
\leq\frac{2\alpha_1\bar p^2}{\underline p q_-}
=\theta^{-1},
\qquad
\frac{\underline p q_-}{2\bar p^2}
\leq\frac1{2\|W_1\|}\leq\alpha_1.
\]
Set \(C_G=\max\{\|G-P_+\|,\|\widetilde G\|\}\).
Lemma~\ref{lem:chain-terminal-correction} implies
\(C_G\leq C_\star\).
The choices established in Appendix~\ref{app:horizon-uniform-bounds}
are
\begin{equation}\label{eq:M_c-M_P}
\begin{aligned}
\alpha_c&=\frac1{2\|W_1\|},\\
M_c&=\sqrt{\kappa(W_1)}+\|W_2\|\kappa(W_1)C_G,\\
M_P&=\|P_+\|+\kappa(W_1)C_G.
\end{aligned}
\end{equation}
Substituting the preceding bounds proves
\eqref{eq:chain-uniform-stability-constants}.
All quantities on the right depend only on the fixed physical
and cost parameters, which proves the asserted uniformity.
\end{proof}

\subsection{Conditioning and quantum complexity}
\label{app:chain-conditioning-complexity}

We complete the proof of Corollary~\ref{cor:main-mass-spring}
by checking the remaining parameters in the main results.
The coefficient bounds give
\[
\|A\|\leq a_0:=\frac{1+\delta}{m_-}+k_+,
\qquad
\|Q\|\leq q_+,\qquad
\|G\|\leq\gamma q_+.
\]
The Hamiltonian splits into its block-diagonal and off-diagonal
parts, so
\begin{equation}\label{eq:chain-H-norm-bound}
q_-\leq\|Q\|\leq\|H\|
\leq\|A\|+\max\{\|D\|,\|Q\|\}
\leq a_0+\max\{\mu^{-1}\|\mathcal F\|^2,q_+\}=:H_*.
\end{equation}
Each row and column of \(A\) has at most four nonzero entries;
those of \(Q,G\) have at most three, and those of \(B,R\)
have at most one. Thus the common sparsity bound \(s=4\) is
valid for every chain. Also,
\[
d=2,\qquad \kappa(R)=1,\qquad
\|B\|_{\max}=1,\qquad \|G\|_{\max}\leq\gamma q_+.
\]

The constant \(M\) in Section~\ref{sec:main-results} therefore
satisfies
\begin{equation}\label{eq:chain-M-bound}
M\leq\overline M
:=\overline M_c^2
\left(1+\frac{\overline M_c^2}{2\mu\alpha_*}\right)^2
(1+\overline M_P)^2.
\end{equation}
For any fixed step size \(h>0\),
Theorem~\ref{thm:main-conditioning} gives
\begin{equation}\label{eq:chain-conditioning-bound}
\begin{aligned}
\|L_{\rm ex}^{-1}\|
&\leq\frac{\overline M}{1-e^{-\alpha_*h}},\\
\kappa(L_{\rm ex})
&\leq
\frac{\overline M}{1-e^{-\alpha_*h}}
\left(1+\max\left\{e^{hH_*},
\sqrt{1+\gamma^2q_+^2}\right\}\right).
\end{aligned}
\end{equation}
Both bounds are independent of \(n_m\) and \(T\).
If \(hH_*<1/(4\pi)\), choose \(p\) so that
\[
(4p)^p\geq\frac{32\overline M}{1-e^{-\alpha_*h}}.
\]
Then \eqref{eq:Delta_h-ineq} ensures
\(\Delta_h\|L_{\rm ex}^{-1}\|\leq1/2\), and hence
\(\kappa(L)\leq3\kappa(L_{\rm ex})\).
Thus a stage number independent of \(n_m,T\) suffices to
preserve the conditioning bound at fixed \(h\).

For output estimation, use the grid \eqref{eq:m-scaling} and
choose \(p\) according to the desired accuracy as in
Appendix~\ref{app:algorithm-discretization} and the subsequent
output proofs. Equation~\eqref{eq:chain-H-norm-bound} gives
\(m=\mathcal O(1+T)\), while
\eqref{eq:algorithm-inverse-scaling} yields
\[
K_{\rm ex}
=\mathcal O\!\left(
\overline M(1+\alpha_*^{-1}H_*)\right)
=\mathcal O(1),
\]
uniformly also for short horizons.
Consequently, the complexity parameter
\(\mathcal C\) in Section~\ref{sec:main-results} is bounded
independently of \(n_m\) and \(T\).
Write \(\delta_{\rm fail}\in(0,1)\) for the failure probability,
to distinguish it from the damping coefficient \(\delta\).
Under Assumption~\ref{ass:be},
Theorems~\ref{thm:main-cost} and~\ref{thm:main-control} give
the respective query bounds
\[
\widetilde{\mathcal O}\!\left(
\frac{\|x^0\|^2}{\varepsilon}\log\frac2{\delta_{\rm fail}}\right),
\qquad
\widetilde{\mathcal O}\!\left(
\frac{\|x^0\|}{\varepsilon}\log\frac4{\delta_{\rm fail}}\right).
\]
The corresponding additional gate counts have the further factor
\(\log^2(2n)\log^2(2+T)\), up to the suppressed logarithmic
factors. If \(\|x^0\|\) is uniformly bounded, these estimates
have polylogarithmic dependence on \(n\) and \(T\).

\end{document}